\documentclass[journal]{IEEEtran}
\ifCLASSINFOpdf
\else
\fi

\usepackage{color}
\usepackage{amsmath,amssymb,bm}
\usepackage{cite}
\usepackage{multirow}
\usepackage{url}
\usepackage{graphicx}
\usepackage{slashbox}
\usepackage{epstopdf}
\usepackage{algorithm}
\usepackage{algorithmic}
\usepackage{subcaption}
\usepackage{cite}

\newcommand*\xbar[1]{%
  \hbox{%
    \vbox{%
      \hrule height 0.5pt 
      \kern0.3ex
      \hbox{%
        \kern-0.05em
        \ensuremath{#1}%
        \kern-0.05em
      }%
    }%
  }%
}

\newcommand{\argmin}{\mathop{\rm arg~min}\limits}
\newtheorem{rmk}{Remark}

\newtheorem{lem}{Lemma}
\newtheorem{defin}{Definition}
\newtheorem{Prop}{Proposition}
\def\thline{\noalign{\hrule height 1pt}}

\def\thline{\noalign{\hrule height 1pt}}

\def\l({\left(}
\def\r){\right)}
\begin{document}
%
\title{Directional Analytic Discrete Cosine Frames}

\author{Seisuke Kyochi~\IEEEmembership{Member,~IEEE},
        Taizo Suzuki, and
        Yuichi Tanaka,~\IEEEmembership{Senior Member,~IEEE}
\thanks{S. Kyochi is with the Department of Computer Science, Kogakuin University, Tokyo 163-8677, Japan (e-mail: kyochi@cc.kogakuin.ac.jp).}
\thanks{T. Suzuki is with the Faculty of Engineering, Information and Systems, University of Tsukuba, Tsukuba 305-8573 Japan (e-mail:taizo@cs.tsukuba.ac.jp).}
\thanks{Y.~Tanaka is with the Department of Computer and Information Sciences, Tokyo University of Agriculture and Technology, Koganei, Tokyo, 184-8588  Japan (e-mail: ytnk@cc.tuat.ac.jp).}
}

\markboth{}%
{Kyochi: Directional Discrete Cosine Frames}
\maketitle

\begin{abstract}
Block frames called \textit{directional analytic discrete cosine frames} (DADCFs) are proposed for sparse image representation. 
In contrast to conventional overlapped frames, the proposed DADCFs require a reduced amount of 1) computational complexity, 2) memory usage, and 3) global memory access. These characteristics are strongly required for current high-resolution image processing. Specifically, we propose two DADCFs based on discrete cosine transform (DCT) and discrete sine transform (DST). The first DADCF is constructed from parallel separable transforms of DCT and DST, where the DST is permuted by row. The second DADCF is also designed based on DCT and DST, while the DST is customized to have no DC leakage property which is a desirable property for image processing. Both DADCFs have rich directional selectivity with slightly different characteristics each other and they can be implemented as non-overlapping block-based transforms, especially they form Parseval block frames with low transform redundancy. We perform experiments to evaluate our DADCFs and compare them with conventional directional block transforms in image recovery.
\end{abstract}


\IEEEpeerreviewmaketitle

\section{Introduction}
\label{sec:intro}
\IEEEPARstart{S}{parse} representation by \textit{frames} has been an essential technique for image analysis and processing \cite{Elad2010, Starck2010}. Various kinds of signal recovery tasks, e.g., {denoising}, {deblurring}, and {restoration from compressive samples}, can also be realized by incorporating sparse image representation in convex and non-convex optimization algorithms\cite{Combettes2008,Afonso2011}. 

Significant efforts have been made to construct efficient frames for sparse image representation. Of particular focus has been directional frames, such as {curvelet} \cite{Starck2002}, {contourlet} \cite{Do2005}, directional filter banks (FBs) \cite{Easley2008,Tanaka2009,Lim2013,Zhiyu2014,Shi2014}, and dual-tree complex wavelet transforms (DTCWTs) \cite{Selesnick2005, Lili2013, KYOCHI2010, Kyochi2009, Ishibashi2019} for 2D signals. Directional atoms\footnote{In this paper, atom is referred to as an element $\mathbf{d}_n$ $(n=0,\ldots ,N-1)$ of a frame $\mathbf{D} = \begin{bmatrix}
\mathbf{d}_0 & \cdots & \mathbf{d}_{N-1}
\end{bmatrix} \in \mathbb{R}^{M\times N}$ \cite{Mallat2008}.} are crucial for sparse image representation since images usually contain edges and textures lying along oblique directions. Extended versions of these frames for higher-dimensional signals, such as videos, have also been proposed \cite{Lu2007, Yang2009, Nguyen2010, Held2010}. In addition to directional frames, more general systems such as {dictionary} \cite{Aharon2006, Rubinstein2013} and graph WTs/FBs \cite{Tanaka2014} that capture highly complex structures and non-local similarity have been proposed. Those transforms can provide sparse image representation for various fine components.

Although these existing frames and dictionaries have been successfully applied to image processing, problems and limitations have been recognized in practical situations. First, computational complexity for calculating sparse coefficients is typically high due to the complicated algorithms involved, such as 2D filtering \cite{Do2005}, sparse coding with various iterative schemes \cite{Aharon2006, Rubinstein2013}, and eigenvalue decomposition of a large-scale graph Laplacian \cite{Tanaka2014}. Second, they typically require high transform redundancy which leads to a large amount of memory usage to store the coefficients. Third, since the supports of their atoms in those frames are overlapped with each other, they require global memory access, which disrupts parallel computation. Although recent digital devices have been increasing their computational power, the resolution of captured images has also been increasing and sometimes multiple images will also be taken at once for producing visually pleasant images like those having low-noise and/or high-dynamic range. Hence, the computational cost for image processing has to be kept as small as possible for avoiding installing extra hard/software modules in such devices.

Block-based bases and frames, whose supports are identical or disjoint, are thus highly desired due to their computational efficiency. They are still a key for many image processing applications like video coding. In addition, patch-based techniques based on block-based transforms, such as BM3D and redundant (type-II) discrete cosine transform (DCT) \cite{Danielyan2012, Fujita2015, Wang2017}, show their effectiveness in image recovery.
Nevertheless, directional block frames have received less attention compared with overlapped frames and, unfortunately, they have been believed that they cannot provide rich directional selectivity\footnote{In this paper, ``directional selectivity'' is measured by the number of distinguishable directional subbands for an $M$-channel 2D transform.}. However, we can realize such block-based transforms by carefully choosing their building blocks. 

In this paper, we focus on directional block frames and propose directional analytic discrete cosine frames (DADCFs) based on DCT \cite{Rao1990} and (type-II) discrete sine transform (DST) \cite{Martucci1994}. They have the following advantages against alternative directional block frames:
\begin{itemize}
    \item Directional selectivity of the DADCFs is much richer than that of existing directional block frames.
    \item DADCFs can be designed by appending the DST (or a DST-like transform) and simple extra operations to the DCT, and thus are compatible with the DCT.
\end{itemize}
We introduce two types of DADCF, both forming Parseval frames. The first DADCF contains the DCT and a row-wise permuted version of the DST, and the second DADCF contains the DCT and a DST without DC leakage. The second one is called regularity-constrained DADCF (RDADCF). In order to realize RDADCF, we propose the DST without DC leakage, regularity-constrained DST (RDST), for the first time. The DADCF and the RDADCF have different advantages. The DADCF provides richer directional selectivity while the DC energy will be distributed over several subbands.
As it will be described later, the DC leakage can be avoided by integrating the DADCF with Laplacian pyramid at the expense of redundancy. In contrast, the RDADCF can structurally avoid the DC leakage problem by incorporating the proposed RDST as its building block instead of the row-wise permuted DST. We numerically compare two DADCFs with some existing approaches in compressive sensing reconstruction.

The rest of this paper is organized as follows. Section \ref{sec:relatedworks} summarizes related works. Section \ref{sec:prelim} reviews the conventional directional block bases and the analyticity for images. Section \ref{sec:DirDCT} explains the definition and a customization for preventing DC leakage of the DADCF. Section \ref{sec:RDDCF} introduces the RDADCF. Section \ref{sec:experiments} evaluates the proposed DADCFs in compressive sensing reconstruction. Section \ref{sec:conclusion} concludes with a brief summary.    
\subsection{Notations} Bold-faced lower-case letters and upper-case letters are vectors and matrices, respectively. The subscripts $h$ and $v$ are used to indicate variables corresponding to horizontal and vertical directions, respectively. The other mathematical notations are summarized in Table \ref{tab:Notations}.
\section{Related Works}\label{sec:relatedworks}
Directional block bases and frames can be classified into two categories. One is the fixed class, i.e., transforms equipped with directionally oriented bases. This class of transforms includes discrete Fourier transform (DFT) \cite{Cooley1965}, discrete Hartley transform (DHT) \cite{Bracewell1984}, and real-valued conjugate-symmetric Hadamard transform \cite{Kyochi2014}. The other is the adaptive class, i.e., the application of a non-directional block transform (e.g., the DCT) along suitable oblique directions provided by preprocessing (e.g., edge analysis) for each block \cite{Xu2007, Zeng2008}. Applications of the latter class are relatively limited because transform directions have to be determined from an input signal in advance. For example, in signal recovery, degraded observations make it difficult to find suitable directions. Our directional block frames correspond to the fixed class.

The main problem with DFT and its variants is that they contain duplicated atoms along the same direction in their basis and hence cannot provide rich directional selectivity (i.e., the number of directional orientations in a basis or a frame). This degrades the efficiency of signal analysis and processing. In order to achieve richer directional selectivity, in this paper, we extend the DCT to the DADCF. Definitely, the DCT is one of the most effective block transforms for image processing tasks and is already integrated into many digital devices. For example, video coding standards, e.g. HEVC \cite{Sullivan2012} and VVC \cite{Bross2021}, employ the various sizes of the (integer) DCT. However, since it does not contain obliquely oriented atoms in its basis, it cannot achieve rich directional selectivity.

In this paper, we reveal that by appending some extra modules, i.e., DST and scaling/addition (and subtraction)/permutation (SAP) operations, to the DCT, the resulting transform provides directionally oriented atoms and thus leads to rich directional selectivity. Furthermore, since the DST can be designed by the (row-wise) flipped and sign-altered version of the DCT, the implementation cost for the proposed transform can be kept low, i.e., the total procedure can be fully implemented by using the DCT and a few SAP operations.

A preliminary version of this work was presented in \cite{Kyochi2016}, which provides a basic structure of the DADCF. In this paper, we newly introduce theory and design algorithm of the RDADCF, and comprehensive experiments. 
\begin{table}[t]
\caption{Notations}
\vspace{-0.2cm}
\label{tab:Notations}
\begin{center}
\scalebox{0.8}{
\begin{tabular}{c|c}
\thline
Notation & Terminology\\ \thline
$\mathbb{R}$ & Real numbers\\\hline
$j$ & $\sqrt{-1}$\\\hline
$\Omega_{N}$, $\Omega_{N_1,N_2}$ & $\{0,\ldots,N\}$, $\{N_1,\ldots,N_2\}$\\\hline
$\mathbb{R}^N$& $N$-dimensional real-valued vector space\\\hline
$\mathbb{R}^{N_v\times N_h}$ & $N_v \times N_h$ real-valued matrices\\\hline
$\mathbf{I}$, $\mathbf{O}$& Identity matrix, zero matrix\\\hline
$\mathbf{A}_{N}$& $N \times N$ square matrix\\\hline
$\mathbf{A}^{\top}$& Transpose of $\mathbf{A}$\\\hline
$x_n$ and $[\mathbf{x}]_n$& $n$-th element of a vector $\mathbf{x}$\\\hline
$X_{m,n}$ and $[\mathbf{X}]_{m,n}$& $(m,n)$-th element of a matrix $\mathbf{X}$\\\hline
$\mathbf{X}^{(m,n)} \in \mathbb{R}^{M \times M}$ &$(m,n)$-th $M \times M$ subblock of $\mathbf{X} \in \mathbb{R}^{ML_v \times ML_h}$\\\hline
$\mathrm{vec}(\mathbf{X}) \in \mathbb{R}^{N_v N_h}$ & \begin{tabular}{c} Vectorization of $\mathbf{X}\in \mathbb{R}^{N_v \times N_h}$\\ $x_{N_v n + m} = X_{m,n}$\end{tabular} \\\hline
$\mathrm{bvec}(\mathbf{X})$ & \begin{tabular}{c}
Block-wise vectorization of $\mathbf{X}\in \mathbb{R}^{ML_v \times ML_h}$\\$[ \mathrm{vec}(\mathbf{X}^{(0,0)})^{\top}\  \ldots\ \mathrm{vec}(\mathbf{X}^{(L_v-1,L_h-1)})^{\top} ]^{\top}$
\end{tabular}
\\ \hline
$\otimes$& Kronecker product \\\hline
$\mathcal{N}(\mathbf{A})$ & Null space of a matrix: $\mathbf{A}$\\\hline
 \begin{tabular}{c}$\mathrm{diag}(a_0, \ldots , a_{N-1})$ \\ 
 $\mathrm{diag}( \mathbf{A}^{(\mathrm{0})}, \ldots , \mathbf{A}^{({N-1})})$
 \end{tabular}& Diagonal/block-diagonal matrices.\\\hline
 $H(z)$ & $z$-transform: $\sum_{n}h(n)z^{-n}$\\\hline
 $H(\omega)$, $ H(e^{j\omega})$, $\mathcal{F}[h] $ & Discrete-time Fourier transform: $\sum_{n}h(n)e^{-j\omega n}$\\\hline
 $H_k(z)$ ($k = 0, \ldots, M-1$) &  $M$-channel subband filters\\ \hline
 $H({\bm \omega})$& $H(e^{j\omega_v},e^{j\omega_h})$\\\hline
 $H_{k_v,k_h}({\bm \omega})$ & 
 \begin{tabular}{c}
$H_{k_v}(\omega_v)H_{k_h}(\omega_h)$\\ 
\end{tabular}\\\hline
$\|\mathbf{x}\|_{p}$ ($\mathbf{x}\in \mathbb{R}^N$, $p \in [1 , \infty )$)  & \begin{tabular}{c} $\ell_p$-norm, $\|\mathbf{x}\|_{p}=\left(\sum^{N-1}_{n=0} |x_n|^p\right)^{\tfrac{1}{p}}$
 \end{tabular}\\\hline
  \begin{tabular}{c} $\|\mathbf{x}\|_{1,2}$ ($\mathbf{x}\in \mathbb{R}^{KN}$) \end{tabular} &  \begin{tabular}{c} $\ell_{1,2}$-mixed norm, \\$\|\mathbf{x}\|_{1,2} = \sum^{K-1}_{k=0} \left(\sum^{N-1}_{n=0}|{x}_{Nk+n}|^2\right)^{\tfrac{1}{2}}$ \end{tabular}
\\
\hline
\end{tabular}
}
\end{center}
\vspace{-0.4cm}
\end{table}

\section{Preliminaries}
\label{sec:prelim}
\subsection{Conventional Block Bases}\label{subsec:conv}
The DCT \cite{Rao1990} is one of the most popular time-frequency transforms. Its transform matrix $\mathbf{F}^{(\mathrm{C})} \in \mathbb{R}^{M\times M}$ ($M=2^m$, $m \geq 1$)\footnote{For simplicity, we restrict the sizes of all the block transforms to $M=2^m$ throughout this paper, but it is easily extended to the general $M$.} is defined as
\begin{align} \label{def:dct}
[ \mathbf{F}^{(\mathrm{C})} ]_{k,n}=&\ 
\alpha_k \sqrt{\frac{2}{M}}\cos ( \theta_{k,n} )
\end{align}
where $k$ and $n$ are subband and time indices ($k,\ n\in \Omega_{M-1}$), $\theta_{k,n} = \frac{\pi}{M}k\left(n+\frac{1}{2}\right)
$, $\alpha_k = \tfrac{1}{\sqrt{2}}$ for $k=0$ and $\alpha_k = 1$ for otherwise. For $\mathbf{x} = \mathrm{vec}(\mathbf{X})$ ($\mathbf{X} \in \mathbb{R}^{M \times M}$), the 2D DCT is given by $\mathbf{F}^{(\mathrm{C})} \otimes \mathbf{F}^{(\mathrm{C})}\ \in \mathbb{R}^{M^2\times M^2}$. Since the DCT is the approximation of the Karhunen-Lo\`{e}ve transform for a first-order Markov process with a correlation coeffcient $\rho$ when $\rho \rightarrow 1$, the 2D DCT coefficients of natural images tend to be sparse (i.e., its $\ell_1$ norm $\| (\mathbf{F}^{(\mathrm{C})} \otimes \mathbf{F}^{(\mathrm{C})}) \mathbf{x}\|_1$ is small). Thus, the DCT is widely applied to many applications, especially for source coding. However, it is a separable transform and hence it lacks directional selectivity. Formally, its 2D atom $\mathbf{B}^{(k_v,k_h)} \in \mathbb{R}^{M\times M}$ in the DCT basis forms
\begin{align}
\label{eq:2DDCF}
{B}^{(k_v,k_h)}_{n_v,n_h} = \alpha_{k_v} \alpha_{k_h} \frac{2}{M} 
\cos \left(\theta_{k_v,n_v}\right)\cos \left(\theta_{k_h,n_h}\right),
\end{align}
where $k_d$ and $n_d$ ($d \in \{h, v\}$) denote subband and spatial indices, respectively ($k_d,n_d\in \Omega_{M-1}$). Fig. \ref{Dirbasis1}(a) shows an example of the 2D DCT atoms\footnote{In Figs. \ref{Dirbasis1}, \ref{eq:2DCSMFB}, and \ref{fig:2DRDDCF1}, each atom is enlarged for visualization.}. Clearly, they ``mix'' two diagonal components along $45^\circ$ and $-45^\circ$ which reduce directional selectivity.

The 2D DFT can be regarded as a block transform with directional selectivity because it is a complex-valued transform. Its 2D atoms $\mathbf{B}^{(k_v,k_h)} \in \mathbb{R}^{M\times M}$ are represented as
\begin{align}
\label{eq:2DDFT}
{B}^{(k_v,k_h)}_{n_v,n_h} 
= \frac{1}{M} e^{j\left(\varphi_{k_v,n_v}+\varphi_{k_h,n_h}\right)},
\end{align}
where $\varphi_{k,n} = \frac{2\pi}{M}kn $. As shown in Fig. \ref{Dirbasis1}(b), the DFT bases can decompose diagonal components into different subbands. 

There are some real-valued variants of the DFT \cite{Bracewell1984,Kyochi2014} that provide directionally oriented atoms. For example, the DHT \cite{Bracewell1984} can form a directionally oriented basis by modifying some of the original DFT atoms ${B}^{(k_v,k_h)}_{n_v,n_h} =\frac{1}{M}\mathrm{cas}(\varphi_{k_v,n_v})\mathrm{cas}(\varphi_{k_h,n_h})$, $(\mathrm{cas}(\varphi_{k,n}) = \cos(\varphi_{k,n}) + \sin(\varphi_{k,n}))$ to ${B}^{(k_v,k_h,\pm 1)}_{n_v,n_h}$ as
\begin{align}
\label{eq:2DDHT}
{B}^{(k_v,k_h,\pm 1)}_{n_v,n_h} =& \frac{1}{2}\left({B}^{(k_v,k_h)}_{n_v,n_h} \pm {B}^{(M-k_v,M-k_h)}_{n_v,n_h}\right) \nonumber\\
=&\  \frac{1}{M}
\begin{cases}
\cos(\varphi_{k_v,n_v} - \varphi_{k_h,n_h}) \\
\sin(\varphi_{k_v,n_v} + \varphi_{k_h,n_h}) 
\end{cases},
\end{align}
where we assume $M\geq 4$ and $k_v,\ k_h \neq 0, M/2$. 

One problem shared by these conventional directional block transforms is that they contain multiple atoms along the same direction in their basis. For the $M\times M$ DFT and DHT, the number of distinguishable directions is $2\left(\tfrac{M-2}{2}\right)^2$ compared to the number of atoms $M^2$. They cannot provide rich directional selectivity, as shown in Fig. \ref{Dirbasis1}(b). 
\begin{figure}[t]
\centering
\begin{subfigure}[b]{0.5\linewidth}
\centering
\scalebox{0.4}{\includegraphics[keepaspectratio=true]{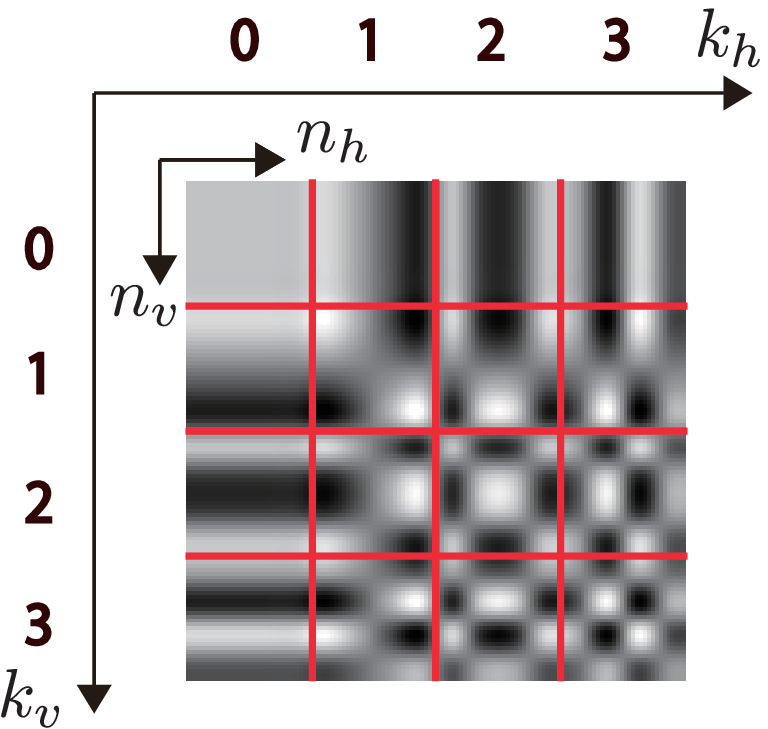}}
\caption{DCT}
\end{subfigure}%
\begin{subfigure}[b]{0.5\linewidth}
\centering
\scalebox{0.4}{\includegraphics[keepaspectratio=true]{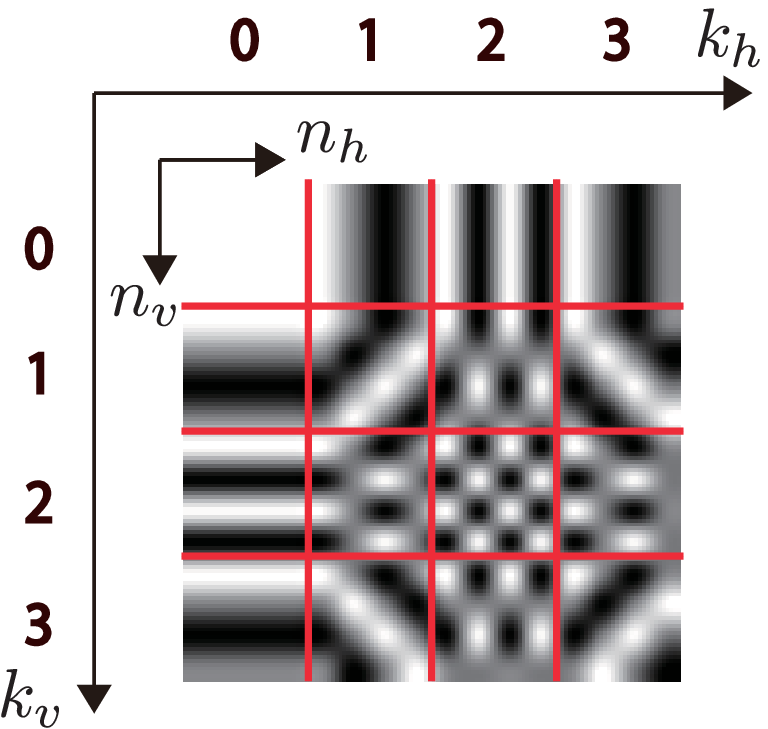}}
\caption{DFT (real part)}
\end{subfigure}%
\caption{Atoms ${B}_{n_v,n_h}^{(k_v,k_h)}$ in basis ($M=4$).}
\vspace{-0.3cm}
\label{Dirbasis1}
\end{figure}

\subsection{Analyticity for Directional Selectivity}\label{sec:analMD}
As explained in Section \ref{subsec:conv}, the 2D DCT cannot provide a directional image representation. We explain this phenomenon in the 2D frequency domain. Let $H_k(\omega)$ be a frequency spectrum of the $k$-th row of the DCT, i.e., $H_k(\omega) = \mathcal{F}[[\mathbf{F}_C]_{k,\cdot}]$. Since $H_k(\omega)$ is the frequency response of a real-valued filter, its spectrum symmetrically distributes in both positive and negative $\omega$ (Fig. \ref{fig:anal}(a)). Thus, the frequency spectrum of the 2D separable DCT $H_{k_v,k_h}({\bm \omega})$ always has nonzero frequency responses in four quadrants, as in Fig. \ref{fig:anal}(c), and it mixes $\pm 45^\circ$ frequency spectra for example.

In contrast, any spectrum of the DFT $U_k(\omega) = \frac{1}{\sqrt{M}}\mathcal{F}[e^{-j\varphi_{k,\cdot}}]$ (complex-valued filter) has a frequency response in only positive (or negative) $\omega$, as in Fig. \ref{fig:anal}(b). This property is referred to as \textit{analyticity} \cite{Selesnick2005}, i.e., $|U_k(\omega)| \approx 0$ for $\omega < 0$ (or $\omega > 0$). Thus, frequency spectra of the 2D separable DFT $U_{k_v,k_h}({\bm \omega})$ are localized in one quadrant (Fig. \ref{fig:anal}(d)), which indicates the directional subband.

Conventional separable directional WTs/FBs utilize analyticity. For example, DTCWTs consist of two $M$-channel filter banks $\{H_k(\omega)\}_{k=0}^{M-1}$ and $\{G_k(\omega)\}_{k=0}^{M-1}$, where those complex combination satisfies analyticity as follows:
\begin{align}
H_{k}({ \omega})=&\  \frac{1}{2}\left(U_{k}({\omega}) + \overline{U_{k}({\omega})}\right) ,\nonumber\\
G_{k}({ \omega})=&\  \frac{1}{2j}\left( U_{k}({\omega}) - \overline{U_{k}({\omega})}\right), \nonumber\\
U_{k}({\omega})=&\   H_{k}({ \omega})+jG_{k}({ \omega}), \ |U_{k}({\omega})| \approx 0 \ (\omega < 0 )
\end{align}
Here, we assume that the frequency spectrum $U_{k}({\omega})$ distributes in the positive frequency domain (Fig. \ref{fig:anal}(b)). Then, by using the 2D frequency spectra of the complex-valued filters $U_{k_v,k_h}({\bm \omega})=U_{k_v}(\omega_v)U_{k_h}(\omega_h)$ and $U_{k_v,\overline{k_h}}({\bm \omega}):= U_{k_v}({\omega_v})\overline{U_{k_h}({\omega}_h)}$,
the 2D directional frequency spectrum of the real-valued filter can be designed as follows:
\begin{align}\label{eq:posi}
&\ \frac{1}{2}\left(U_{k_v,k_h}({\bm \omega}) + \overline{U_{k_v,k_h}({\bm \omega})}\right) \nonumber\\&=\ H_{k_v}({\omega}_v)H_{k_h}({\omega}_h) - G_{k_v}({\omega}_v)G_{k_h}({\omega}_h),\nonumber\\
&\ \frac{1}{2}\left(U_{k_v,\overline{k_h}}({\bm \omega})+ \overline{U_{k_v,\overline{k_h}}({\bm \omega})}\right)\nonumber\\
&=\ H_{k_v}({\omega}_v)H_{k_h}({\omega}_h) + G_{k_v}({\omega}_v)G_{k_h}({\omega}_h).
\end{align}
Considering \eqref{eq:posi}, a directional frequency decomposition can be realized by two 2D separable FBs followed by addition/subtraction, as in Fig. \ref{2DPlaneDTCWT}. $M$-channel DTCWTs can distinguish $2M^2$ directional subbands.

\begin{figure}[t] 
\centering
\begin{subfigure}[b]{0.5\linewidth}
\centering
\scalebox{0.5}{\includegraphics[keepaspectratio=true]{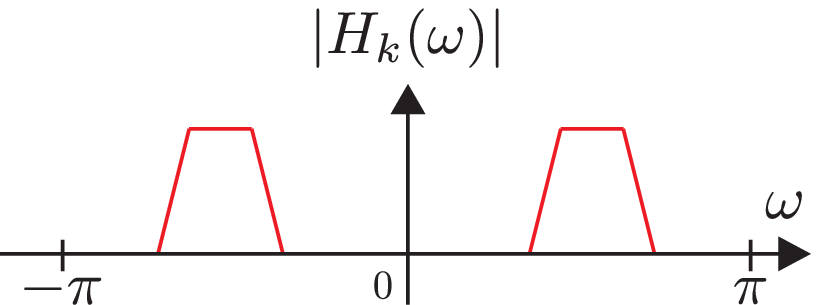}}
\caption{Real-valued filter}
\end{subfigure}%
\begin{subfigure}[b]{0.5\linewidth}
\centering
\scalebox{0.5}{\includegraphics[keepaspectratio=true]{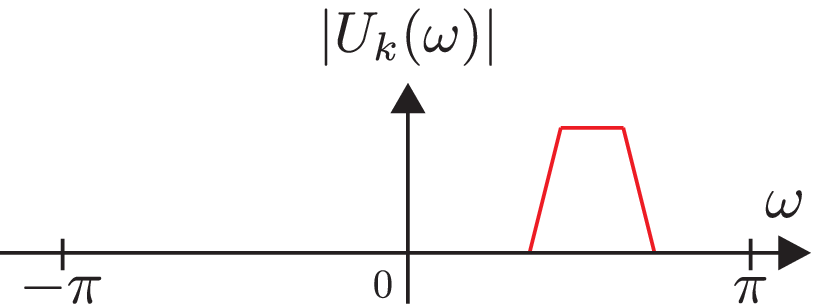}}
\caption{Complex-valued filter}
\end{subfigure}%
\\
\begin{subfigure}[b]{0.3\linewidth}
\centering
\scalebox{0.25}{\includegraphics[keepaspectratio=true]{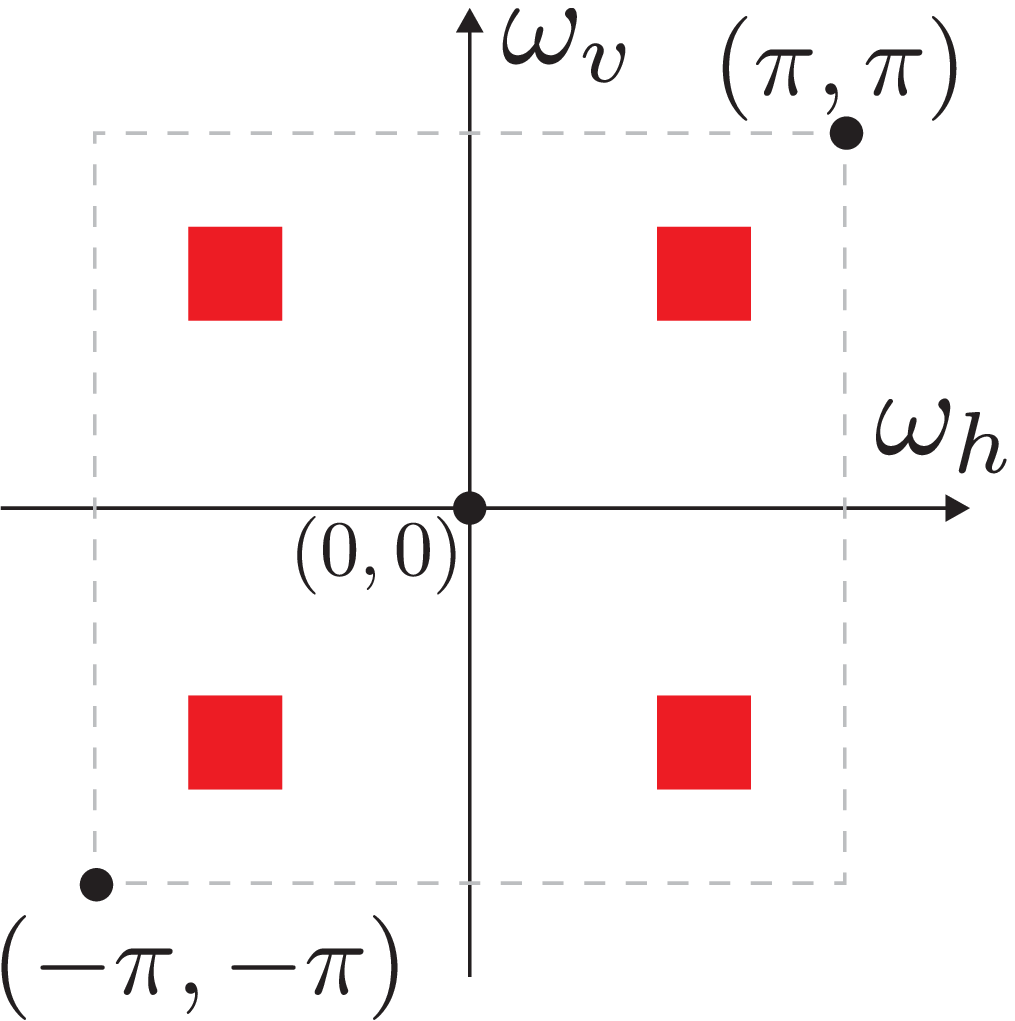}}
\caption{$H_{k_v,k_h}({\bm \omega})$}
\end{subfigure}%
\begin{subfigure}[b]{0.3\linewidth}
\centering
\scalebox{0.25}{\includegraphics[keepaspectratio=true]{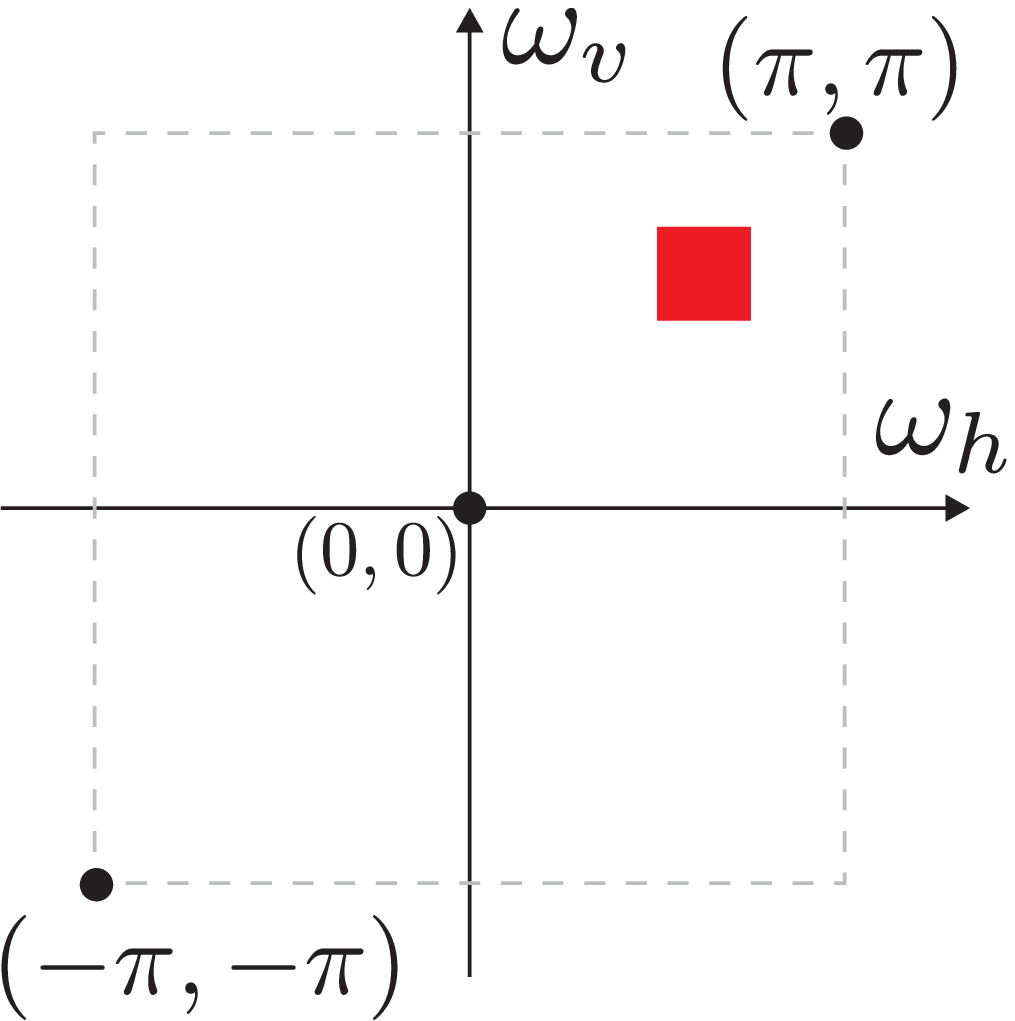}}
\caption{$U_{k_v,k_h}({\bm \omega})$}
\end{subfigure}%
\begin{subfigure}[b]{0.4\linewidth}
\centering
\scalebox{0.25}{\includegraphics[keepaspectratio=true]{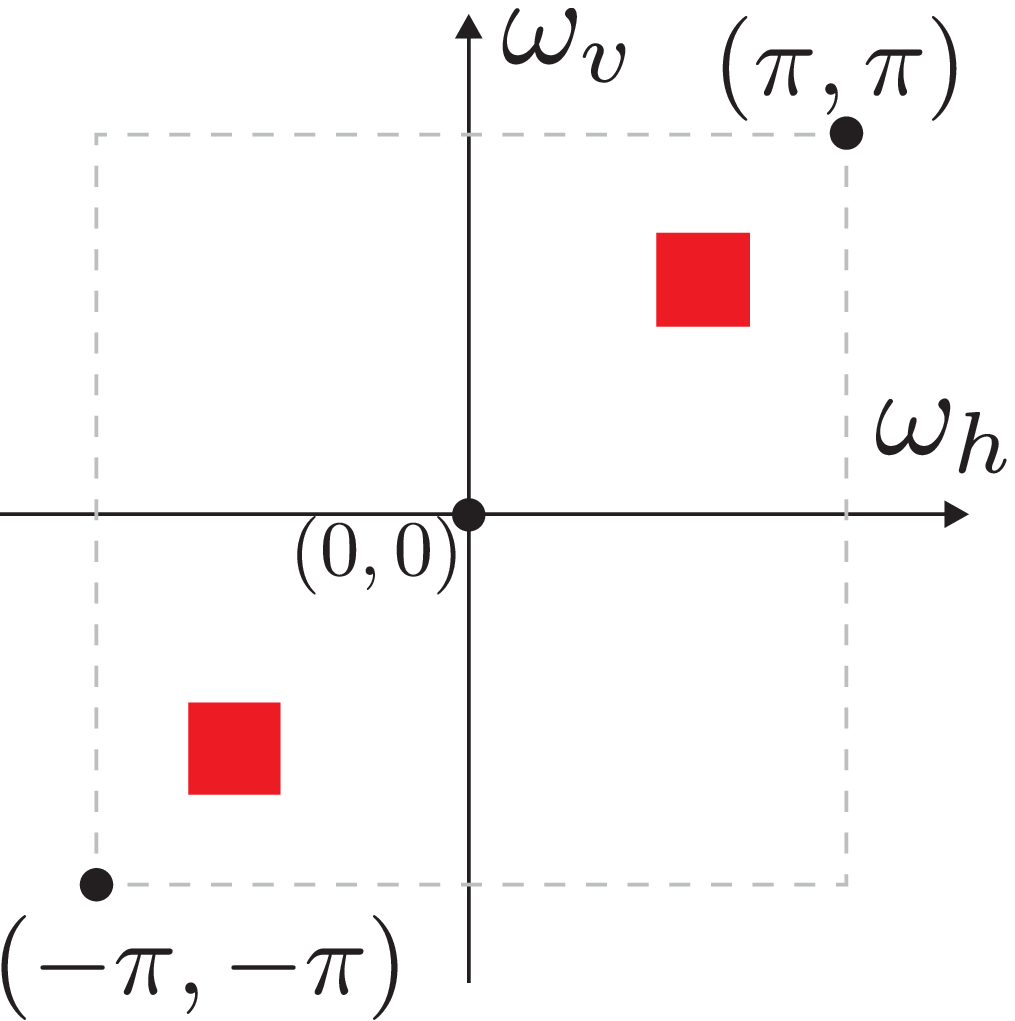}}
\caption{$U_{k_v,k_h}+\overline{U_{k_v,k_h}}$}
\end{subfigure}%
\caption{Example of frequency spectra (analytic and non-analytic filters).}
\vspace{-0.3cm}
\label{fig:anal}
\end{figure}
\begin{figure}[t]
\begin{center}
\begin{minipage}{\linewidth}
\begin{center}
\scalebox{0.85}{
\includegraphics[width=\linewidth,keepaspectratio=true]{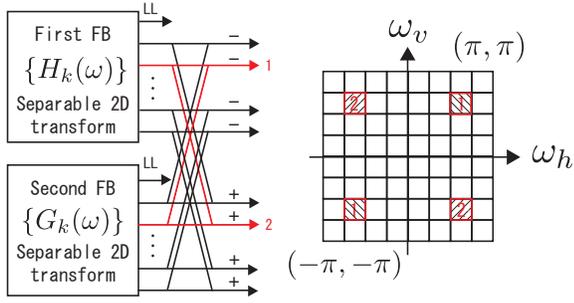}
}
\end{center}
\end{minipage}
\end{center}
\caption{Configurations for 2D DTCWTs. For $M=4$, 32 directional subbands can be distinguished.}
\label{2DPlaneDTCWT}
\end{figure}

\section{Directional Analytic Discrete Cosine Frames}
This section introduces the DADCF. The definition of the DADCF is given in Section \ref{subsec:DDDCF}. Directional selectivity of the DADCF is then discussed by analyzing its atoms in Section \ref{subsec:ADDCF}. As it will be shown in Section \ref{subsec:DC}, the DADCF suffers from the DC leakage problem. One solution is given by constructing the DADCF pyramid (the DADCF with Laplacian pyramid) in Section \ref{subsec:DC}.
\label{sec:DirDCT}
\subsection{Definition of Directional Analytic Discrete Cosine Frame}
\label{subsec:DDDCF}
This section introduces DADCFs for 2D signals by extending the conventional DCT. 
\begin{defin}
The analysis operator of the DADCF $\mathbf{F}^{(\mathrm{D})} \ \in\ \mathbb{R}^{2M^2\times M^2}$ is defined as
\begin{align} 
\label{eq:DDCF}
\mathbf{F}^{(\mathrm{D})} := &\ 
\mathbf{P}^{(\mathrm{I})\top}
\mathbf{W}^{(\mathrm{I})}
 \mathbf{P}^{(\mathrm{I})}
\begin{bmatrix}
\mathbf{F}^{(\mathrm{C})} \otimes \mathbf{F}^{(\mathrm{C})} \\
\mathbf{F}^{(\mathrm{S})} \otimes \mathbf{F}^{(\mathrm{S})}
\end{bmatrix},\nonumber\\
\mathbf{W}^{(\mathrm{I})} =& \ 
\mathrm{diag}\left(\frac{1}{\sqrt{2}}\mathbf{I}_{2M-1},\frac{1}{2}
\begin{bmatrix}
  \mathbf{I}_{(M-1)^2} & -\mathbf{I}_{(M-1)^2}\\
 \mathbf{I}_{(M-1)^2} & \mathbf{I}_{(M-1)^2}
\end{bmatrix}
\right),
\end{align}
where $\mathbf{F}^{(\mathrm{C})}$ is defined in \eqref{def:dct} and $\mathbf{P}^{(\mathrm{I})} \in \mathbb{R}^{M^2\times M^2}$ is a permutation matrix that places the $2M-1$ DCT and $2M-1$ DST coefficients associated with the subband indices $k_v = 0$ or $k_h = 0$ to the first part, and the other $2(M-1)^2$ coefficients associated with the subband indices $k_v \neq 0$ and $k_h \neq 0$ to the last (see Fig. \ref{eq:2DCSMFB}(a)). $\mathbf{F}^{(\mathrm{S})}\ \in \mathbb{R}^{M\times M}$ is defined as
\begin{align}\label{def:dst}
[ \mathbf{F}^{(\mathrm{S})} ]_{k,n}=&\ 
\begin{cases}
\sqrt{\frac{1}{M}}\sin \left(\pi\left(n+\frac{1}{2}\right)\right) & (k=0) \\
\sqrt{\frac{2}{M}}\sin \left(\frac{\pi}{M}k\left(n+\frac{1}{2}\right)\right) & (k \neq 0 )
\end{cases}.
\end{align}
\end{defin}

$ \mathbf{F}^{(\mathrm{S})}$ is nothing but the row-wise permuted version of the DST. In this paper, we simply denote the row-wise permuted DST as the DST. Because the DCT ($\mathbf{F}^{(\mathrm{C})}$) and the DST ($\mathbf{F}^{(\mathrm{S})}$) are orthogonal matrices, the DADCF is a Parseval block frame: $\mathbf{F}^{(\mathrm{D})\top}\mathbf{F}^{(\mathrm{D})} = \mathbf{I}_{M^2}$.

The construction flow of the DADCF is illustrated in Fig. \ref{eq:2DCSMFB}(a). The DADCF requires two block transforms, additions and subtractions between two transforms, and scaling operations. Its computational cost is slightly higher than conventional block transforms due to the SAP operations but much lower than other overlapped frames and dictionaries, as mentioned in Section \ref{sec:intro}. Its redundancy ratio is 2: It is the same as the DFT and the DTCWTs \cite{Selesnick2005, Kyochi2014, KYOCHI2010, Kyochi2009}, and thus it can reduce the amount of memory usage compared with highly redundant frames and dictionaries like those in \cite{Aharon2006, Rubinstein2013}.
\begin{rmk}
According to the basic knowledge on the DCT/DST, the DST $\mathbf{F}^{(\mathrm{S})}\in \mathbb{R}^{M \times M}$ can be implemented as the permuted and sign-altered version of the DCT $\mathbf{F}^{(\mathrm{C})}\in \mathbb{R}^{M \times M}$, i.e., $\mathbf{F}^{(\mathrm{S})} =\mathbf{P}^{(\mathrm{I\hspace{-.1em}I})}\mathbf{F}^{(\mathrm{C})}\mathrm{diag}(1,-1,\ldots,1,-1)$, where $\mathbf{P}^{(\mathrm{I\hspace{-.1em}I})} \in \mathbb{R}^{M \times M}$ denotes the permutation matrix that arranges the rows of matrices in reverse order. Thus, the DADCF can be implemented by the DCT with a few trivial SAP operations.
\end{rmk}
\begin{figure}[t]
\centering
\begin{subfigure}{1\linewidth}
\centering
\scalebox{0.5}{\includegraphics[keepaspectratio=true]{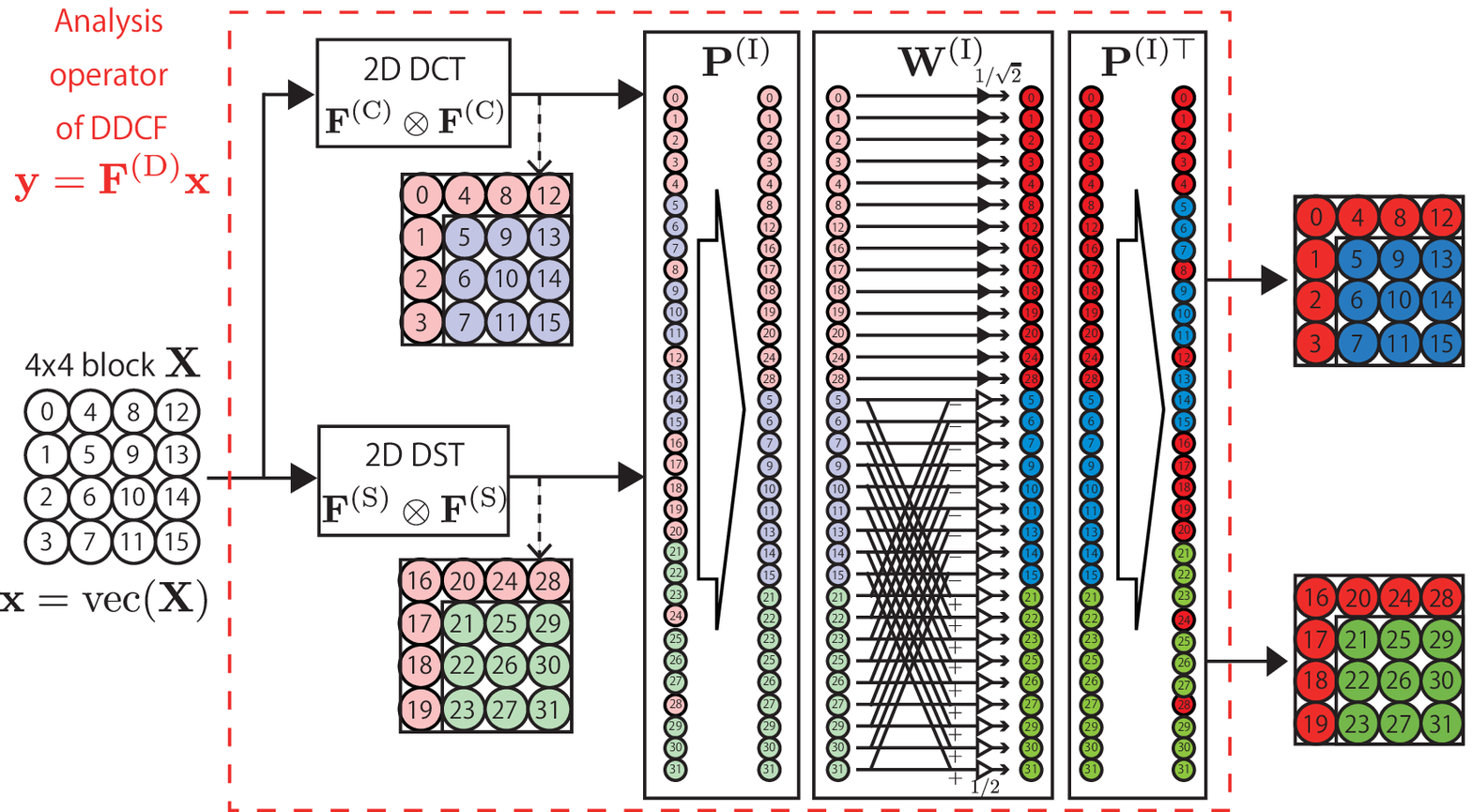}}
\caption{}
\end{subfigure}
\begin{subfigure}{0.45\linewidth}
\centering
\scalebox{0.4}{\includegraphics[keepaspectratio=true]{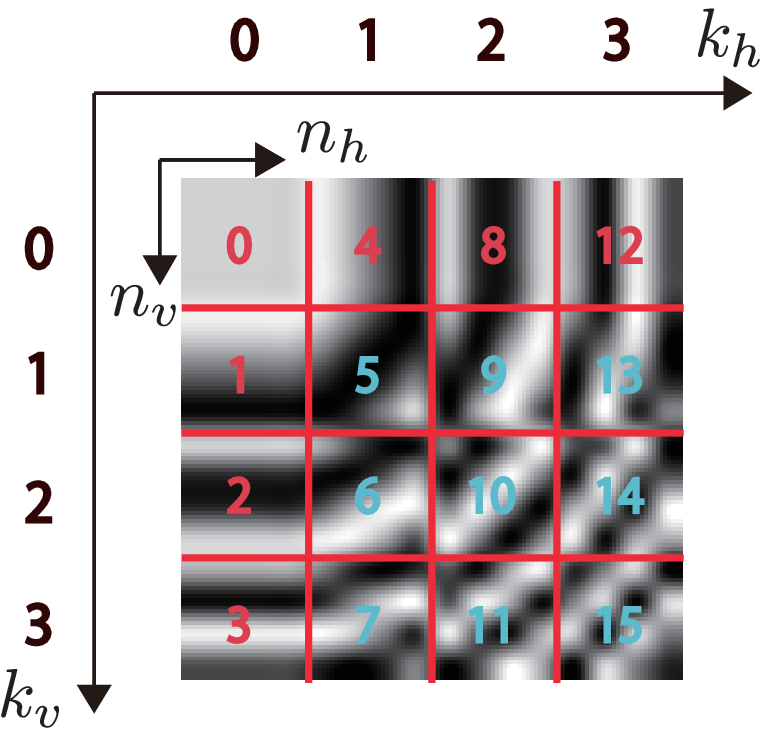}}
\caption{${C}_{n_v,n_h}^{(k_v,k_h, 1)}$, ${B}_{n_v,n_h}^{(k_v,k_h, -1)}$}
\end{subfigure}%
\begin{subfigure}{0.45\linewidth}
\centering
\scalebox{0.4}{\includegraphics[keepaspectratio=true]{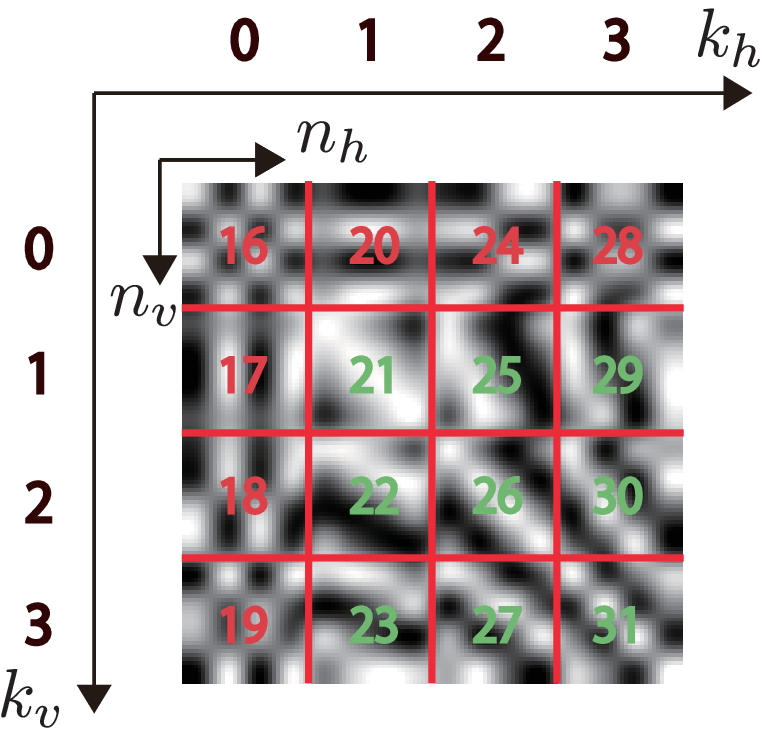}}
\caption{${S}_{n_v,n_h}^{(k_v,k_h, 1)}$, ${B}_{n_v,n_h}^{(k_v,k_h, 1)}$}
\end{subfigure}%
\caption{(a) Procedure of the DADCF ($M=4$). (b) and (c): Atoms ${C}_{n_v,n_h}^{(k_v,k_h, 1)}$, ${S}_{n_v,n_h}^{(k_v,k_h, 1)}$ (red), and ${B}_{n_v,n_h}^{(k_v,k_h, \pm 1)}$ (blue and green) in the DADCF. The numbers indicate the rightmost subband indices in (a).}
\label{eq:2DCSMFB}
\vspace{-0.3cm}
\end{figure}
\vspace{-0.2cm}
\subsection{Directional Atoms in DADCF}
\label{subsec:ADDCF}
Here, we examine the directional selectivity of the DADCF defined in \eqref{eq:DDCF}. The frequency spectra of the $k$-th rows of the DCT \eqref{def:dct} and the DST \eqref{def:dst} are given by
$
H_{k}({ \omega}):= \mathcal{F}[[\mathbf{F}^{(\mathrm{C})}]_{k,\cdot}] ,
G_{k}({ \omega}):= \mathcal{F}[[\mathbf{F}^{(\mathrm{S})}]_{k,\cdot}]
$,
where $k\geq 1$. Their complex combination
\begin{align}
H_{k}({ \omega})+j G_{k}({ \omega})=  \sqrt{\frac{2}{M}}\sum_{n=0}^{M-1} e^{j\theta_{k,n}}e^{-j\omega n},
\end{align}
which is the spectra of (9), approximately satisfies the analyticity, as shown in Fig. \ref{fig:DC1}(c). As a result, the DADCF is a directional transform with real coefficients from the 2D DCT and DST followed by addition/subtraction operations. Note that the frequency spectrum of $H_{0}({ \omega})+j G_{0}({ \omega})$, i.e., low-pass spectrum, does not satisfy the analyticity. As a result, the DADCF can distinguish $2(M-1)^2$ directional subbands.
\begin{figure}[t]
\centering
\begin{subfigure}{0.29\linewidth}
\centering
\scalebox{0.18}{\includegraphics[keepaspectratio=true]{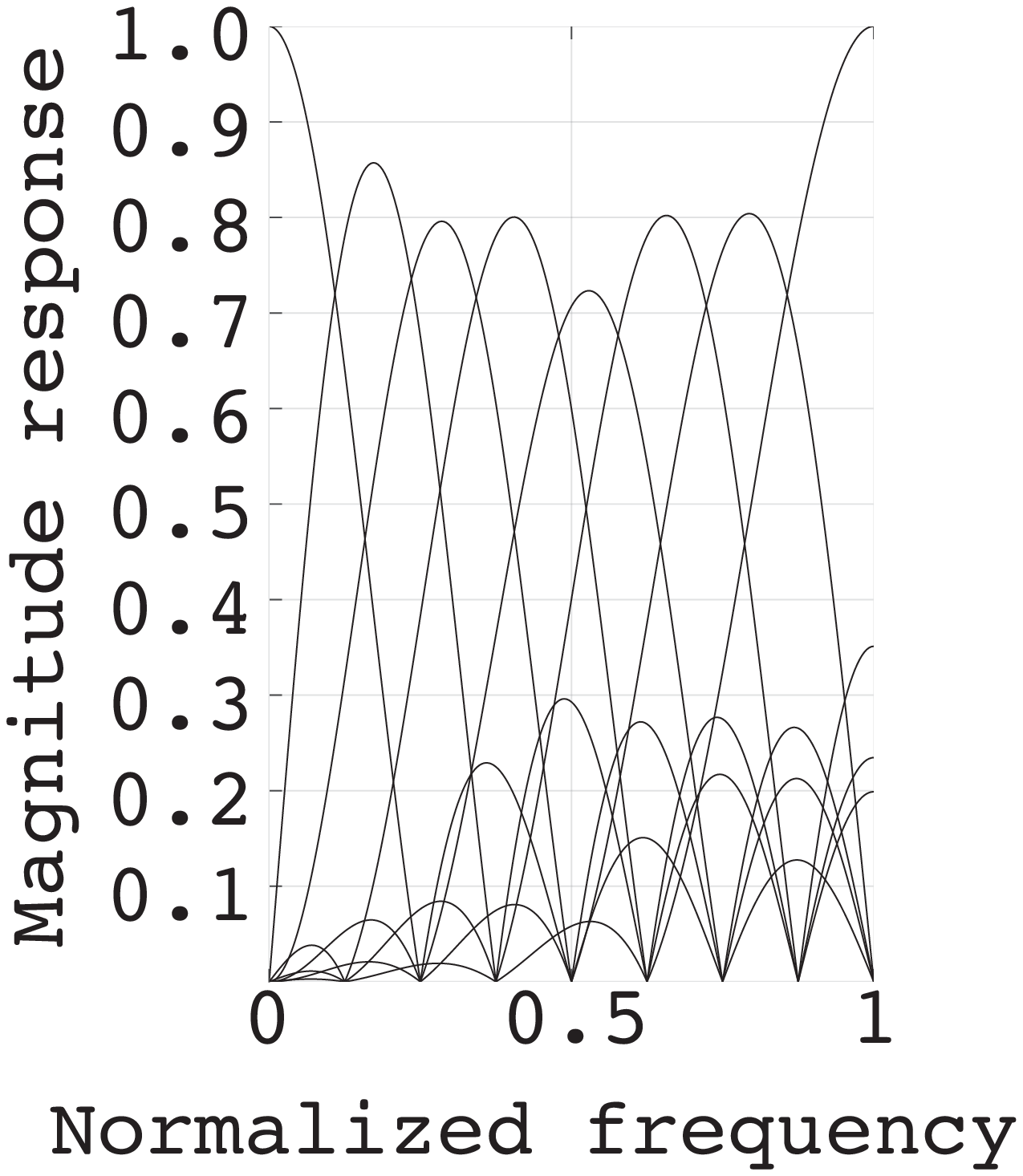}}
\caption{$H_{k}({ \omega})$}
\end{subfigure}%
\begin{subfigure}{0.29\linewidth}
\centering
\scalebox{0.18}{\includegraphics[keepaspectratio=true]{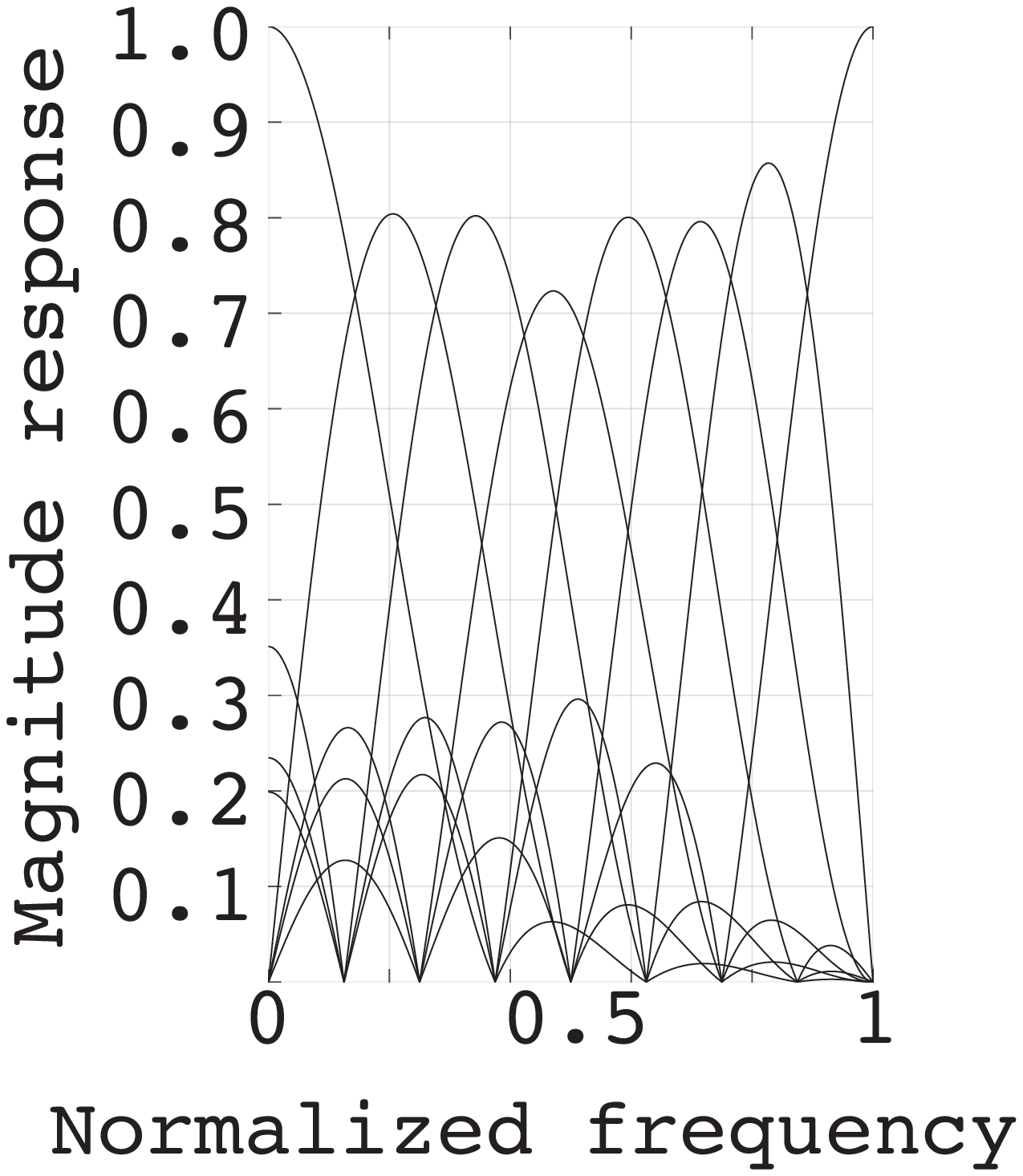}}
\caption{$G_{k}({ \omega})$}
\end{subfigure}%
\begin{subfigure}{0.44\linewidth}
\centering
\scalebox{0.18}{\includegraphics[keepaspectratio=true]{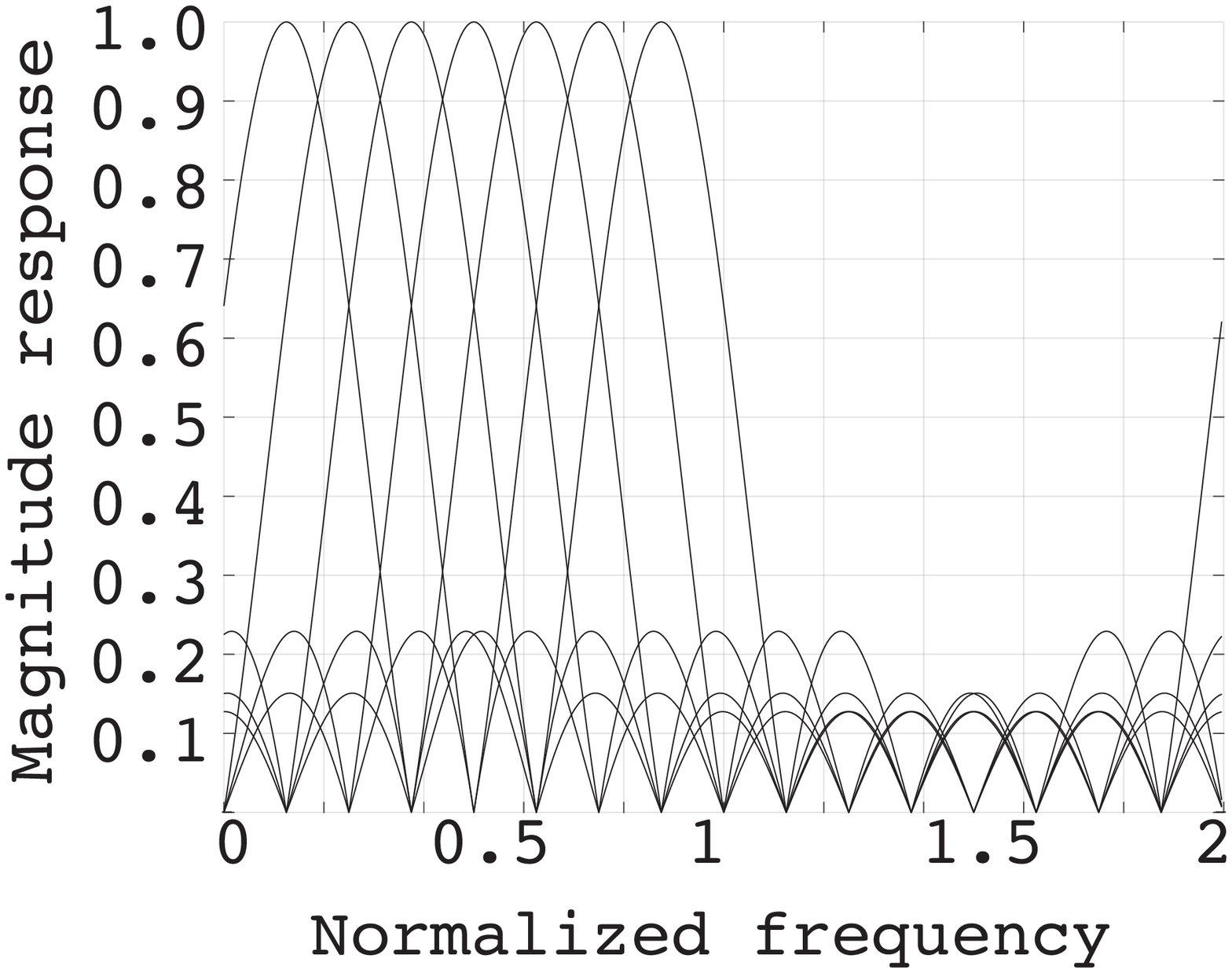}}
\caption{$H_{k}({ \omega})+j G_{k}({ \omega})$}
\end{subfigure}%
\caption{Frequency spectra (frequency: $[0, 2\pi]$, $M=8$): (a) DCT, (b) DST, (c) the complex combination ($1\leq k \leq 7$).}
\label{fig:DC1}
\end{figure}

Next, we show the atoms of the DADCF. Because the DADCF forms a Parseval block frame, it is enough to examine the synthesis transform $\begin{bmatrix} \mathbf{f}_0 & \ldots & \mathbf{f}_{2M^2-1} \end{bmatrix}^{\top}:= \mathbf{F}^{(\mathrm{D})\top}$. From \eqref{eq:DDCF}, $\mathbf{F}^{(\mathrm{D})\top}$ is composed of 1) an atom in the 2D DCT basis, 2) an atom in the 2D DST basis, or 3) directional atoms arising from the addition/subtraction of 2D DCT/DST atoms. Let $\mathbf{B}^{(k_v,k_h, 1)},\ \mathbf{B}^{(k_v,k_h, -1)} \in \mathbb{R}^{M\times M}$ be two directional atoms of the DADCF that correspond to the subband $(k_v,k_h) \in \Omega_{1,M-1} \times \Omega_{1,M-1}$. These atoms can be represented as
\begin{align} 
\label{eq:DDCFframe}
{B}^{(k_v,k_h, \pm 1)}_{n_v,n_h} 
 =&\ {C}^{(k_v,k_h)}_{n_v,n_h}   \pm {S}^{(k_v,k_h)}_{n_v,n_h} 
 \nonumber \\
=&\ 
\frac{2}{M} 
\cos \left(\theta_{k_v,n_v} \mp \theta_{k_h,n_h}\right),
\end{align}
where ${C}^{(k_v,k_h)}_{n_v,n_h}=[\mathbf{F}^{(\mathrm{C})}]_{k_v,n_v}[\mathbf{F}^{(\mathrm{C})}]_{k_h,n_h}$ and ${S}^{(k_v,k_h)}_{n_v,n_h}=[\mathbf{F}^{(\mathrm{S})}]_{k_v,n_v}[\mathbf{F}^{(\mathrm{S})}]_{k_h,n_h}$. In contrast to the DFT and the DHT bases \eqref{eq:2DDFT} and \eqref{eq:2DDHT}, these 2D atoms lie along various oblique directions, as illustrated in Figs. \ref{eq:2DCSMFB}(b) and (c).
\subsection{Lack of Regularity of DADCF}
\label{subsec:drawback}
As previously shown, some 2D frequency responses $U_{k_v,k_h}({\bm \omega}) + \overline{U_{k_v,k_h}({\bm \omega})}$ and $U_{k_v,\overline{k_h}}({\bm \omega}) + \overline{U_{k_v,\overline{k_h}}({\bm \omega})}$ obtained from the DADCF do not decay at ${\bm \omega} = (0,0)$ which leads to DC leakage. 

Figs. \ref{fig:DC}(a) and (b) show an image decomposition example. The image used is \textit{Zoneplate} $\{\mathbf{X}^{(i,j)}\}_{i,j \in \Omega_{31}}$ ($\mathbf{X} \in \mathbb{R}^{256\times 256}$) and its (half of the arranged) DADCF coefficients $\{\mathbf{x}_2\}_{i,j \in \Omega_{31}}$ with $M=8$, where $[ \mathbf{x}_1^{\top}\ \mathbf{x}_2^{\top} ]^{\top} =\mathbf{F}^{(\mathrm{D})}\ \mathrm{vec}(\mathbf{X}^{(i,j)})$ are shown in Fig. \ref{fig:DC}(b). We observe that the DC leakage has been appeared and it leads to the reduction of the sparsity.

The DC leakage is due to the fact that the DST $\mathbf{F}^{(\mathrm{S})}$ loses regularity, as mathematically explained in the following. For a block transform $\mathbf{F} \in \mathbb{R}^{M\times M}$, regularity condition \cite{strang1996} is formulated as
\begin{align}\label{eq:regF}
\begin{bmatrix}
c &
0 &
\cdots &
0
\end{bmatrix}^{\top}
= \mathbf{F}\mathbf{1},
\end{align}
where $c$ is some constant and $\mathbf{1} = \begin{bmatrix} 1& 1& \ldots& 1 \end{bmatrix}^{\top}$. As shown Fig. \ref{fig:DC1}(b), the DST $\mathbf{F}^{(\mathrm{S})}$ leads to the DC leakage. It can be theoretically verified as in the following proposition.
\begin{Prop}\label{lem:S1} 
Let vectors $\{\mathbf{s}_{k}\}_{k=0}^{M-1}$ be the basis of $M\times M$ DST, i.e., $ \begin{bmatrix} \mathbf{s}_0 & \ldots & \mathbf{s}_{M-1} \end{bmatrix} = \mathbf{F}^{(\mathrm{S})\top}$. Then,
\begin{align*}
\langle \mathbf{s}_{k} , \mathbf{1} \rangle = \begin{cases}
\frac{\sqrt{2}}{\sqrt{M}\sin\left(\frac{\pi}{2M}k\right) } & (k = 2\ell+1) \\
0 & (\mathrm{otherwise})
\end{cases},
\end{align*}
where $\ell \in \Omega_{\frac{M}{2}-1}$.
\end{Prop}
\begin{proof}
It is clear that $\langle \mathbf{s}_0, \mathbf{1} \rangle = 0$. For the other cases,
\begin{align*}
\langle \mathbf{s}_k,  \mathbf{1} \rangle =&\  \sqrt{\frac{2}{M}} \sum_{n=0}^{M-1} \sin \left(\frac{\pi}{M}k\left(n+\frac{1}{2}\right)\right) \nonumber\\
=&\  \sqrt{\frac{2}{M}} \mathcal{I} \left[\sum_{n=0}^{M-1} e^{ j\frac{\pi}{M}k\left(n+\frac{1}{2}\right)}\right] 
 = \frac{(1-(-1)^k)}{\sqrt{2M}\sin\left(\frac{\pi}{2M}k\right) },
\end{align*}
where $\mathcal{I}$ takes the imaginary part of a complex number. 
\end{proof}
From the above proposition, the odd rows $(k = 2\ell+1)$ of the DST produce nonzero responses for a constant-valued signal, i.e., DC leakage.
\subsection{DADCF Pyramid}
\label{subsec:DC} 
To obtain sparser coefficients, we introduce the DADCF pyramid inspired by \cite{Do2005}. The analysis operator of the DADCF pyramid $\mathbf{F}^{(\mathrm{DP})}$ is defined as:
\begin{align}\label{eq:DDCFP}
\mathbf{F}^{(\mathrm{DP})} = \begin{bmatrix} (\mathbf{D}\mathbf{M})^\top & (\mathbf{F}^{(\mathrm{D})}(\mathbf{I}- \widetilde{\mathbf{M}}\mathbf{D}^{\top}\mathbf{D}\mathbf{M}))^\top \end{bmatrix}^\top, 
\end{align}where $\mathbf{D} = \begin{bmatrix} 1 & 0 & \cdots & 0 \end{bmatrix}$ is the downsampling operator (and thus $\mathbf{D}^{\top}$ corresponds to the upsampling operator), $\mathbf{M} = \mathbf{M}^{(\mathrm{0})} \otimes \mathbf{M}^{(\mathrm{0})}\ \in \mathbb{R}^{M^2\times M^2}$ is the averaging operator, where $[\mathbf{M}^{(\mathrm{0})}]_{k,n} = \frac{1}{M}$, $\widetilde{\mathbf{M}} = M {\mathbf{M}}$.  By applying the DADCF pyramid to the input block $\mathrm{vec}(\mathbf{X}^{(i,j)})$, we can obtain its average value (denoted as ${x}_L$) and the DADCF coefficients of the DC-subtracted input block (denoted as $\mathbf{x}_1$ and $\mathbf{x}_2$) as $\begin{bmatrix} {{x}}_L & {\mathbf{x}}_1^\top &  {\mathbf{x}}_2^\top \end{bmatrix}^\top = \mathbf{F}^{(\mathrm{DP})}\mathrm{vec}(\mathbf{X}^{(i,j)})$. 

For example, Fig. \ref{fig:DC}(c) shows the (half of) the transformed coefficients $\{\mathbf{x}_2\}_{i,j \in \Omega_{31}}$. It is clear that sparser coefficients can be obtained and the DADCF pyramid $\mathbf{F}^{(\mathrm{DP})}$ is still invertible. By this operation, however, the number of transformed coefficients is slightly increased from $2N^2$ to $2N^2+(N/M)^2$ for $N \times N$ input images. 

\begin{figure}[t]
\centering
\begin{subfigure}[b]{0.33\linewidth}
\centering
\scalebox{0.22}{\includegraphics[keepaspectratio=true]{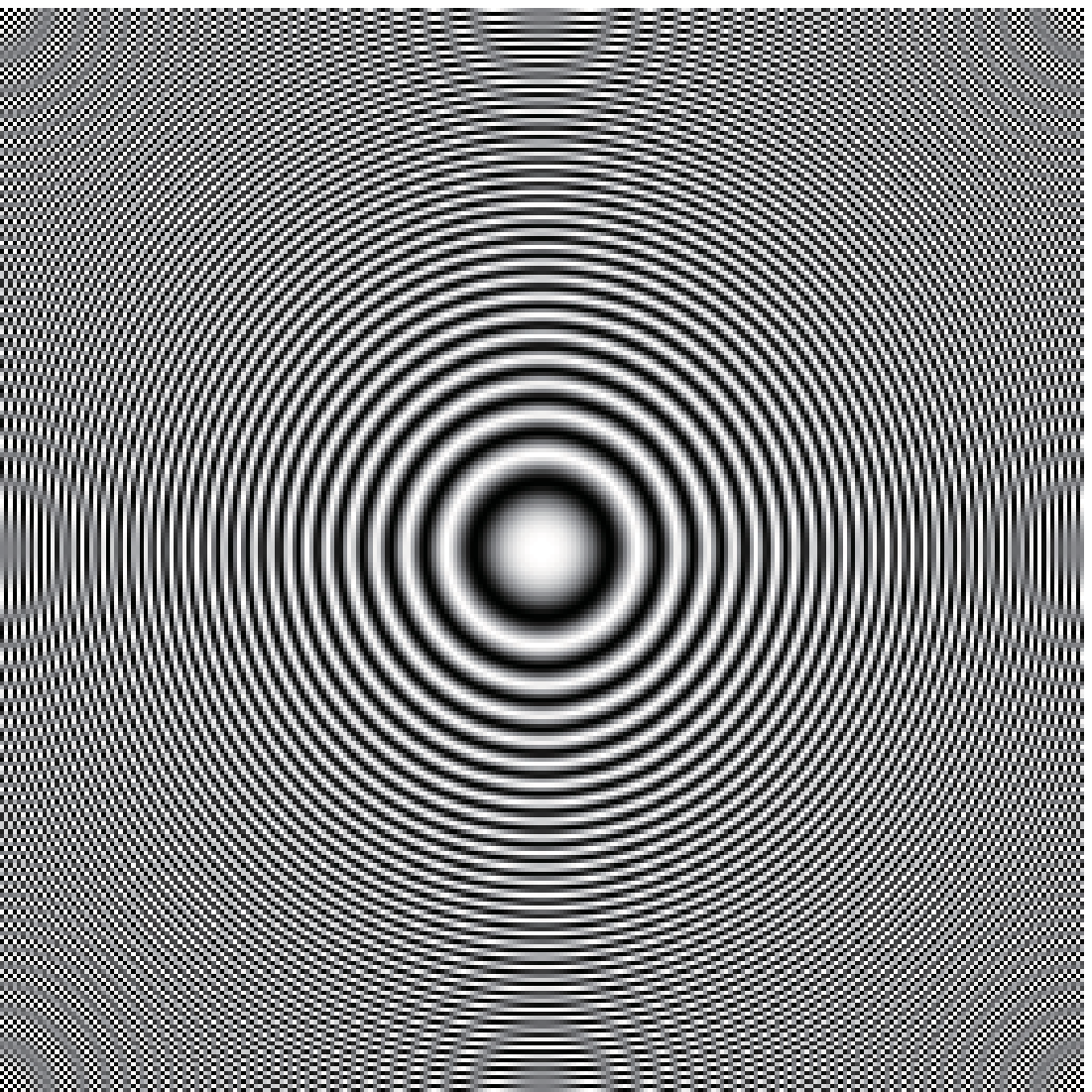}}
\caption{}
\end{subfigure}%
\begin{subfigure}[b]{0.33\linewidth}
\centering
\scalebox{0.22}{\includegraphics[keepaspectratio=true]{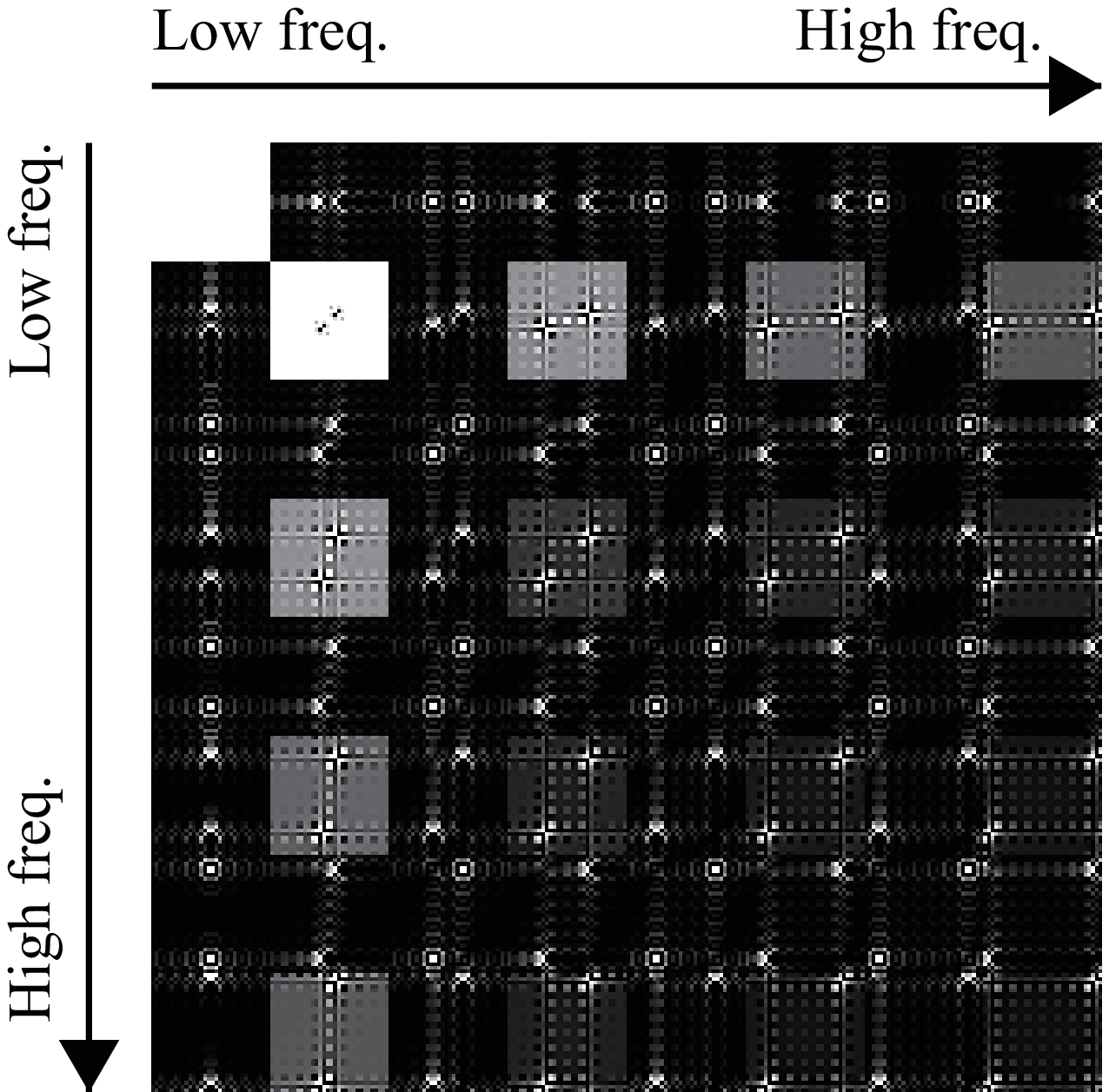}}
\caption{}
\end{subfigure}%
\begin{subfigure}[b]{0.33\linewidth}
\centering
\scalebox{0.22}{\includegraphics[keepaspectratio=true]{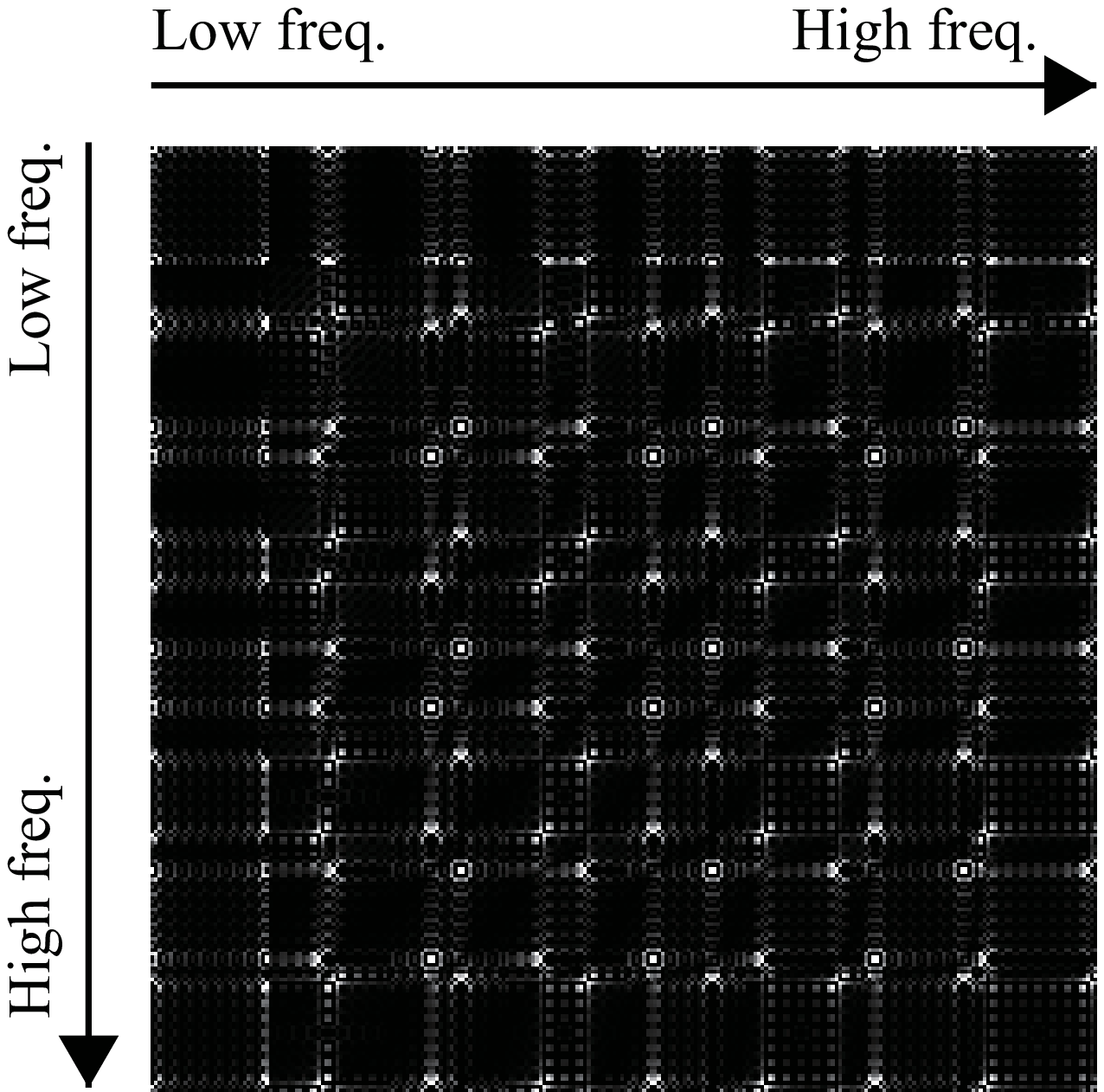}}
\caption{}
\end{subfigure}%
\caption{(a): \textit{Zoneplate}, (b) (Half of) DADCF coefficients ($M=8$), (c) (Half of) DADCF coefficients in the DADCF pyramid ones ($M=8$).}
\label{fig:DC}
\end{figure}
\section{Regularity-Constrained DADCF}\label{sec:RDDCF}
In this section, we introduce another DADCF, called RDADCF. We introduce a RDST in Section \ref{subsec:RDST} and \ref{subsec:imp_rdst}. Then, in Section \ref{subsec:DDCF_DTMRD}, we propose the RDADCF, which overcomes the problem of the DADCF, i.e., DC leakage, and saves the number of the transformed coefficients fewer than the DADCF pyramid. 
\subsection{Design of RDST}\label{subsec:RDST}
This section introduces a modified DST without DC leakage for constructing RDADCF. For notation simplicity, we present steps for constructing the RDST matrix $\mathbf{F}^{(\mathrm{RS})}$.
\begin{description}
\item[\textbf{Step 1}:]\ First, we define a modified DST $\mathbf{S} \in \mathbb{R}^{M\times M}$.
\begin{align}\label{eq:defS}
[ \mathbf{S}]_{k,n}=&\ 
\begin{cases}
\sqrt{\frac{1}{M}} & (k=0) \\
\sqrt{\frac{2}{M}}\sin \left(\frac{\pi}{M}k\left(n+\frac{1}{2}\right)\right) & (k \neq 0 )
\end{cases}\nonumber\\
= & \begin{bmatrix}
\mathbf{s}_0 & \mathbf{s}_1 & \cdots & \mathbf{s}_{M-1}
\end{bmatrix}^{\top}.
\end{align}
\end{description}
In short, it is constructed by replacing the 0-th row of the DST with that of the DCT. The modified DST satisfies the following property (see Appendix \ref{ap:prop:S} for its proof).
\begin{Prop}\label{prop:S}
 $\mathrm{rank}(\mathbf{S}) = M-1$. 
\end{Prop}
Then, we further modify $\mathbf{S}$ in \eqref{eq:defS}. From \eqref{eq:regF}, in order to impose the regularity condition on $\mathbf{S}$, $\{\mathbf{s}_k\}_{k=1}^{M-1}$ should be orthogonal to $\mathbf{s}_0$. Now, we orthogonalize the odd rows of $\mathbf{S}$ in the following way.
\begin{description}
\item[\textbf{Step 2}:]\  Set $\widetilde{\mathbf{S}}^{(0)}  = \begin{bmatrix}
\mathbf{s}_{0}& \mathbf{0} & \mathbf{s}_{2} &  \ldots & \mathbf{s}_{M-1}
\end{bmatrix}^{\top}$.
\end{description}
Here, $\widetilde{\mathbf{S}}^{(0)}$ satisfies the following proposition (see Appendix \ref{ap:prop:rankS1} for its proof).
\begin{Prop}\label{prop:rankS1}
$ \mathrm{rank}(\widetilde{\mathbf{S}}^{(0)}) = M-1$.
\end{Prop}
From Proposition \ref{prop:rankS1}, there is only one zero singular value and its corresponding right-singular vector (denoted as $\mathbf{v}^{(0)}$) belongs to the null space of $\widetilde{\mathbf{S}}^{(0)}$. It implies that $\widetilde{\mathbf{S}}^{(0)}\mathbf{v}^{(0)} = \mathbf{0}$, i.e., $\mathbf{v}^{(0)}$ satisfies the regularity condition. $\widetilde{\mathbf{S}}^{(0)}$ is updated by replacing $\mathbf{0}$ to $\mathbf{v}^{(0)}$.
\begin{description}
\item[\textbf{Step 3}:] \  Set $\mathbf{S}^{(0)} = \begin{bmatrix}
 \mathbf{s}_{0} & \mathbf{v}^{(0)} & \mathbf{s}_{2} & \ldots & \mathbf{s}_{M-1}
\end{bmatrix}^{\top}$.
\end{description}
Note that $\mathbf{v}^{(0)}$ can be explicitly represented as $[\mathbf{v}^{(0)}]_n=\sqrt{\frac{1}{M}}(-1)^n = [\mathbf{F}^{(\mathrm{S})}]_{0,n} \left(= \sqrt{\frac{1}{M}}\sin \left( \pi\left(n+\frac{1}{2}\right)\right)\right)$ because the row of the DST $[\mathbf{F}^{(\mathrm{S})}]_{0,n}$ corresponding to the highest frequency subband is orthogonal to $\{\mathbf{s}_0,\ \mathbf{s}_2,\ \ldots,\ \mathbf{s}_{M-1}\}$.
It clearly follows that $\mathrm{rank}(\mathbf{S}^{(0)}) = M$.

Consequently, by repeating Steps 1 and 2, we can obtain the orthogonal matrix $\mathbf{S}^{(M/2-1)}$ whose odd rows are replaced by the different ones from the initial $\mathbf{S}^{(0)}$. A summary of the algorithm is given in Algorithm \ref{alg:RDST}. Hereafter $\mathbf{F}^{(\mathrm{RS})}:=\mathbf{S}^{(M/2-1)}$ is termed as a RDST. 
\begin{algorithm}[t]
    \caption{The design procedure for RDST}
    \label{alg:RDST}
    \begin{algorithmic}[1]
        {\footnotesize
            \STATE Set $\mathbf{S}$ is as in \eqref{eq:defS}.
            \FOR{$k=0$ to $M/2-1$}
            \STATE 
            Set $\widetilde{\mathbf{S}}^{(k)} = \begin{bmatrix}
 \ldots& \mathbf{s}_{2k}& \mathbf{0} & \mathbf{s}_{2k+2} & \ldots
\end{bmatrix}^{\top}$. 
            \STATE Find the right-singular vector $\mathbf{v}^{(k)}$ corresponding to zero singular value 
            .
            \STATE Set $\mathbf{S}^{(k)} = \begin{bmatrix}
 \ldots&\mathbf{s}_{2k}& \mathbf{v}^{(k)} & \mathbf{s}_{2k+2} & \ldots
\end{bmatrix}^{\top}$. 
            \ENDFOR
            \STATE Output $\mathbf{S}^{(M/2-1)}$.}
    \end{algorithmic}
\end{algorithm}

The RDST satisfies the following properties (see Appendix \ref{ap:prop:RDST} for its proof).
\begin{Prop}\label{prop:RDST}
Let $\mathbf{F}^{(\mathrm{RS})} \in \mathbb{R}^{M\times M}$ be the RDST.
\begin{enumerate}
\item This satisfies the regularity condition, i.e., $$\begin{bmatrix}
c &
0 &
\cdots &
0
\end{bmatrix}^{\top} = \mathbf{F}^{(\mathrm{RS})}\mathbf{1}.$$
\item Some rows of $\mathbf{F}^{(\mathrm{RS})} \in \mathbb{R}^{M\times M}$ are identical with those in the DST matrix $\mathbf{F}^{(\mathrm{S})}$:
$
 [\mathbf{F}^{(\mathrm{RS})}]_{0,n} = \sqrt{\frac{1}{M}}$, $[\mathbf{F}^{(\mathrm{RS})}]_{1,n} = [\mathbf{F}^{(\mathrm{S})}]_{0,n}$, $[\mathbf{F}^{(\mathrm{RS})}]_{2\ell,n} = [\mathbf{F}^{(\mathrm{S})}]_{2\ell,n},
$
where $\mathbf{F}^{(\mathrm{S})} \in \mathbb{R}^{M\times M}$ is the DST matrix and $\ell \in  \Omega_{\frac{M}{2}-1}$.
\item The passband of the spectrum $\mathcal{F}[[\mathbf{F}^{(\mathrm{RS})}]_{2\ell+1,\cdot}]$ is the same as that of the DST $\mathcal{F}[[\mathbf{F}^{(\mathrm{S})}]_{2\ell+1,\cdot}]$ $(\ell \geq 2)$.
\end{enumerate}
\end{Prop}

In Fig. \ref{fig:FRRDST}(a), the red lines show the frequency spectra of the newly updated rows ($k=0,3,5,7$) in the RDST ($M=8$) and the dashed gray lines show those of the corresponding rows in the DST (the rest frequency spectra of the RDST are identical to those of the DST). The frequency spectra of the RDST approximate those of the original DST, but decay at zero frequency.
\begin{figure}[t]
        \begin{center}
        \begin{subfigure}{0.5\linewidth}
        \centering
                \scalebox{0.2}{\includegraphics[keepaspectratio=true]{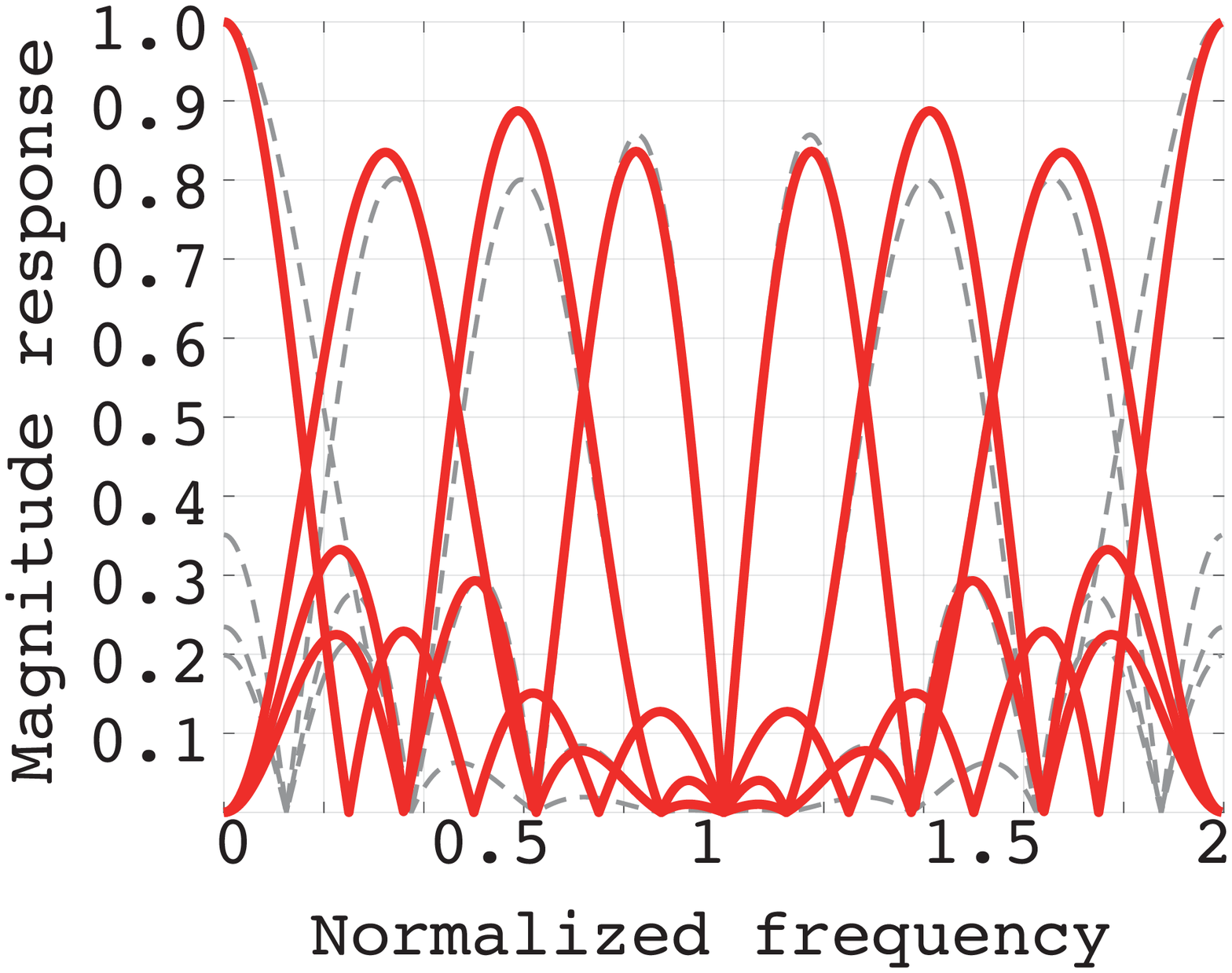}}
                \caption{}
        \end{subfigure}%
                \begin{subfigure}{0.5\linewidth}
        \centering
                \scalebox{0.2}{\includegraphics[keepaspectratio=true]{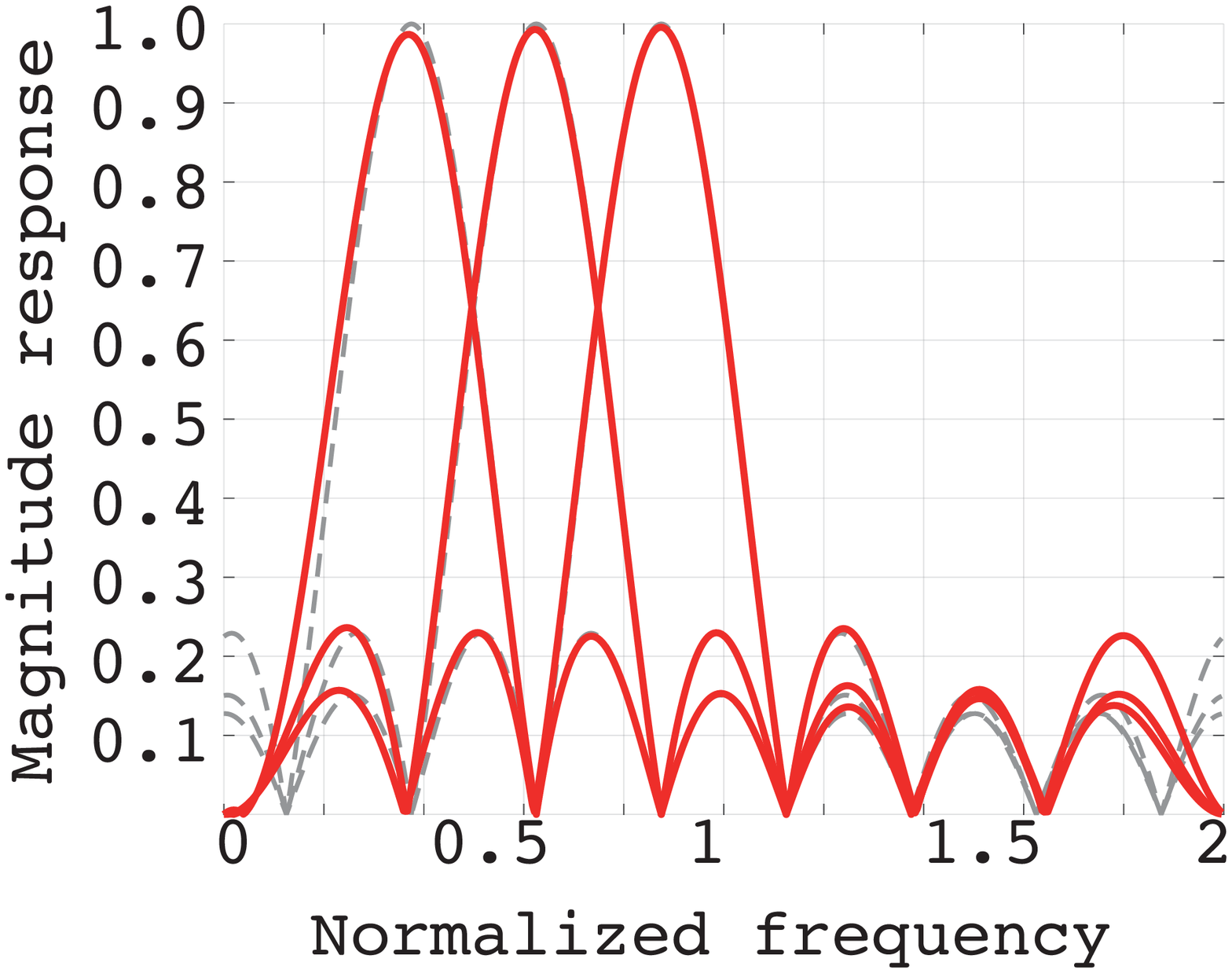}}
                                \caption{}
        \end{subfigure}%
        \end{center}
                               \caption{Frequency spectra (frequency: $[0, 2\pi]$, $M=8$): (a) red lines: $\mathcal{F}[[\mathbf{F}^{(\mathrm{RS})}]_{k,\cdot}]$, dashed gray lines: $\mathcal{F}[[\mathbf{F}^{(\mathrm{S})}]_{k,\cdot}]$ ($k=0,3,5,7$), (b) red lines: $\mathcal{F}[[\mathbf{F}^{(\mathrm{C})}]_{k,\cdot}]+j\mathcal{F}[[\mathbf{F}^{(\mathrm{RS})}]_{k,\cdot}]$, dashed gray lines: $\mathcal{F}[[\mathbf{F}^{(\mathrm{C})}]_{k,\cdot}]+j\mathcal{F}[[\mathbf{F}^{(\mathrm{S})}]_{k,\cdot}]$ ($k=3,5,7$).}
                \label{fig:FRRDST}
\end{figure}
\subsection{Implementation of RDST}\label{subsec:imp_rdst}
From Proposition \ref{prop:RDST}, the $\frac{M}{2}$ rows of the RDST $\mathbf{F}^{(\mathrm{RS})}\in \mathbb{R}^{M\times M}$ are the same as the rows of the original DST $\mathbf{F}^{(\mathrm{S})}\in\mathbb{R}^{M\times M}$, and both matrices are orthogonal. Thus, we can derive that the RDST can be implemented by the cascade of the DST and an orthogonal matrix as in the following. 

Let $\mathbf{F}^{(\mathrm{S,e})},\ \mathbf{F}^{(\mathrm{S,o})} \in \mathbb{R}^{\frac{M}{2} \times M}$ be the even and odd rows of the $\mathbf{F}^{(\mathrm{S})} \in \mathbb{R}^{M \times M}$, respectively. Then, the RDST $\mathbf{F}^{(\mathrm{RS})} \in \mathbb{R}^{M \times M}$ can be expressed as:
\begin{align}
\mathbf{F}^{(\mathrm{RS})} =&\  \mathbf{P}^{(\mathrm{I\hspace{-.1em}I\hspace{-.1em}I})}\begin{bmatrix}
\mathbf{F}^{(\mathrm{S,e})} \\
\widetilde{\mathbf{F}}^{(\mathrm{S,o})}
\end{bmatrix} = \mathbf{P}^{(\mathrm{I\hspace{-.1em}I\hspace{-.1em}I})}\mathrm{diag}(\mathbf{I}_{\frac{M}{2}},\mathbf{\Gamma}_{\frac{M}{2}}) 
\begin{bmatrix}
\mathbf{F}^{(\mathrm{S,e})} \\
\mathbf{F}^{(\mathrm{S,o})}
\end{bmatrix} \nonumber\\ =&\  \mathbf{P}^{(\mathrm{I\hspace{-.1em}I\hspace{-.1em}I})}\mathrm{diag}(\mathbf{I}_{\frac{M}{2}},\mathbf{\Gamma}_{\frac{M}{2}})  \mathbf{P}^{(\mathrm{I\hspace{-.1em}V})} {\mathbf{F}}^{(\mathrm{S})}, 
\end{align}
where $\mathbf{P}^{(\mathrm{I\hspace{-.1em}I\hspace{-.1em}I})},\  \mathbf{P}^{(\mathrm{I\hspace{-.1em}V})} \in \mathbb{R}^{M \times M}$ are the permutation matrices, and the matrix $\mathbf{\Gamma}_{\frac{M}{2}}$ is guaranteed to be an orthogonal matrix because of orthogonality of the RDST and the DST. Since $\mathbf{\Gamma}_{\frac{M}{2}}$ is an orthogonal matrix, it can be factorized into $\frac{M(M-2)}{8}$ rotation matrices. Thus, the RDST is still a hardware-friendly transform that can be implemented by the $\mathbf{F}^{(\mathrm{C})}$ with some trivial operations.
\subsection{Design of RDADCF}\label{subsec:DDCF_DTMRD}
Finally, a RDADCF $\mathbf{F}^{(\mathrm{RD})}$ is defined using the RDST as follows.
\begin{defin}
Let $\mathbf{F}^{(\mathrm{RD})}  \in \mathbb{R}^{2M^2\times M^2}$ be the analysis operator of the RDADCF defined as:
\begin{align} 
\label{eq:RDDCF}
\mathbf{F}^{(\mathrm{RD})}  := &\ 
\mathbf{P}^{(\mathrm{V})\top}
\mathbf{W}^{(\mathrm{I\hspace{-.1em}I})} 
 \mathbf{P}^{(\mathrm{V})}
\begin{bmatrix}
\mathbf{F}^{(\mathrm{C})} \otimes \mathbf{F}^{(\mathrm{C})} \\
\mathbf{F}^{(\mathrm{RS})} \otimes \mathbf{F}^{(\mathrm{RS})}
\end{bmatrix},\nonumber\\
\mathbf{W}^{(\mathrm{I\hspace{-.1em}I})} =& \ 
\mathrm{diag}\left(\frac{1}{\sqrt{2}}\mathbf{I}_{4M-4},\frac{1}{2}
\begin{bmatrix}
  \mathbf{I}_{(M-2)^2} & -\mathbf{I}_{(M-2)^2}\\
 \mathbf{I}_{(M-2)^2} & \mathbf{I}_{(M-2)^2}
\end{bmatrix}
\right),
\end{align}
where $\mathbf{P}^{(\mathrm{V})} \in \mathbb{R}^{2M^2\times 2M^2}$ is a permutation matrix. $\mathbf{P}^{(\mathrm{V})}$ places the $4M-4$ DCT and DST coefficients associated with the subband indices $k_v \in \{ 0,1 \}$ or $k_h \in \{ 0,1 \}$ to the first part, and the other $2(M-2)^2$ coefficients to the last (see Fig. \ref{fig:2DRDDCF1}(a)). Due to the orthogonality of the RDST, the RDADCF clearly forms a Parseval block frame, i.e., $\mathbf{F}^{(\mathrm{RD})\top} \mathbf{F}^{(\mathrm{RD})} = \mathbf{I}_M$.
\end{defin}

Now, we discuss the capability of the directional subband decomposition based on the DCT and the RDST. Let $\mathbf{F}^{(\mathrm{C})},\ \mathbf{F}^{(\mathrm{RS})} \in \mathbb{R}^{M\times M}$ be the DCT and the RDST matrices. As discussed in Section \ref{sec:analMD}, the complex combination $[\mathbf{F}^{(\mathrm{C})}]_{k,\cdot}\pm j [\mathbf{F}^{(\mathrm{RS})}]_{k,\cdot}$ should have a one-sided frequency spectrum for directional subband decomposition. In the case of even $k\ (\geq 2)$, the rows of $[\mathbf{F}^{(\mathrm{RS})}]_{k,\cdot}$ are identical to those of the DST. Therefore, the frequency spectrum $[\mathbf{F}^{(\mathrm{C})}]_{k,\cdot}\pm j [\mathbf{F}^{(\mathrm{RS})}]_{k,\cdot}$ is one-sided. In the case of odd $k\ (\geq 3)$, where the rows $[\mathbf{F}^{(\mathrm{RS})}]_{k,\cdot}$ are newly designed in Algorithm \ref{alg:RDST}, the frequency spectrum $[\mathbf{F}^{(\mathrm{C})}]_{k,\cdot}\pm j [\mathbf{F}^{(\mathrm{RS})}]_{k,\cdot}$ can be one-sided  (Fig. \ref{fig:FRRDST}(b)). 
 
Analyticity of the RDADCF can be explained as follows. Let $\{\mathbf{s}_{\ell}^{(s)}\}$ and $\{\mathbf{s}_{\ell}^{(r)}\}$ be the rows of the DST and the RDST, respectively. For any odd $k$ ($\geq 3$), $\mathbf{s}_{k}^{(r)}$ can be obtained by applying orthogonal projection to $\mathbf{s}_{k}^{(s)}$ onto the orthogonal complement of $\{ \mathbf{s}_{\ell}^{(s)} \}_{\Omega_{M-1}\setminus \{k\}}$, as
\begin{align}
\mathbf{s}_{k}^{(r)} = \frac{\pm 1}{\eta_k}\left(\mathbf{s}_{k}^{(s)} - \sum_{\ell \in \Omega_{M-1}\setminus \{k\}} \langle \mathbf{s}_{\ell}^{(r)} , \mathbf{s}_{k}^{(s)} \rangle \mathbf{s}_{\ell}^{(r)}\right),
\end{align} 
where $\eta_k$ is the normalization factor for $\mathbf{s}_{k}^{(r)}$ having unit norm. Let $ [\mathbf{F}^{(\mathrm{W})}]_{k,n} = \frac{1}{\sqrt{M}}e^{-j\frac{2\pi}{M}kn}$ denote the DFT. Because $\mathbf{F}^{(\mathrm{W})}\mathbf{s}_{m}^{(r)}$ and $\mathbf{F}^{(\mathrm{W})}\mathbf{s}_{k}^{(r)}$ have different passbands, $|\langle \mathbf{s}_{m}^{(r)}, \mathbf{s}_{k}^{(r)} \rangle| = |\langle \mathbf{F}^{(\mathrm{W})}\mathbf{s}_{m}^{(r)},\mathbf{F}^{(\mathrm{W})}\mathbf{s}_{k}^{(r)} \rangle|$ is small. Therefore, the spectrum of the $\mathbf{s}_{k}^{(r)}$ can approximate $\mathbf{s}_{k}^{(s)}$ over the passband of $\mathbf{s}_{k}^{(s)}$.

The atoms of the RDADCF lie along the $2(M-2)^2$ frequency directions, as shown in Figs. \ref{fig:2DRDDCF1}(b) and (c), where ${C}^{(k_v,k_h)}_{n_v,n_h}=[\mathbf{F}^{(\mathrm{C})}]_{k_v,n_v}[\mathbf{F}^{(\mathrm{C})}]_{k_h,n_h}$, ${S}^{(k_v,k_h)}_{n_v,n_h}=[\mathbf{F}^{(\mathrm{RS})}]_{k_v,n_v}[\mathbf{F}^{(\mathrm{RS})}]_{k_h,n_h}$, and ${B}^{(k_v,k_h,\pm 1)}_{n_v,n_h}={C}^{(k_v,k_h)}_{n_v,n_h}\pm {S}^{(k_v,k_h)}_{n_v,n_h}$. The number of directional selectivities of the RDADCF is slightly less than the original DADCF. Since the RDADCF with $M=2$ cannot ensure directional selectivity, we recommend $M=2^m$ where $m \geq 2$. 
\begin{figure}[t]
\centering
  \begin{subfigure}{1\linewidth}
  \centering
   \scalebox{0.52}{\includegraphics[keepaspectratio=true]{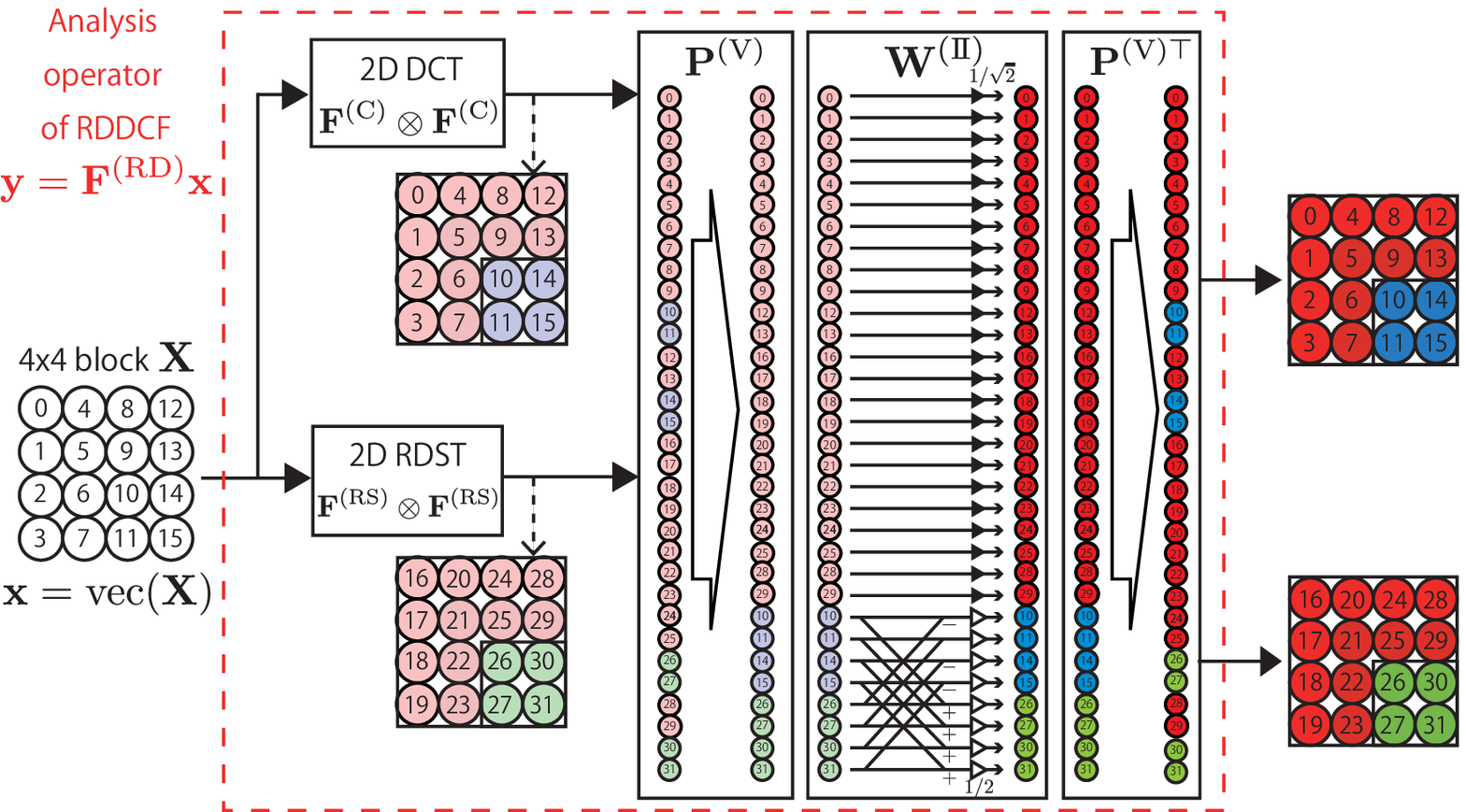}}
   \caption{}
   \end{subfigure}
   \\
        \begin{subfigure}{0.45\linewidth}
        \centering
                \scalebox{0.4}{\includegraphics[keepaspectratio=true]{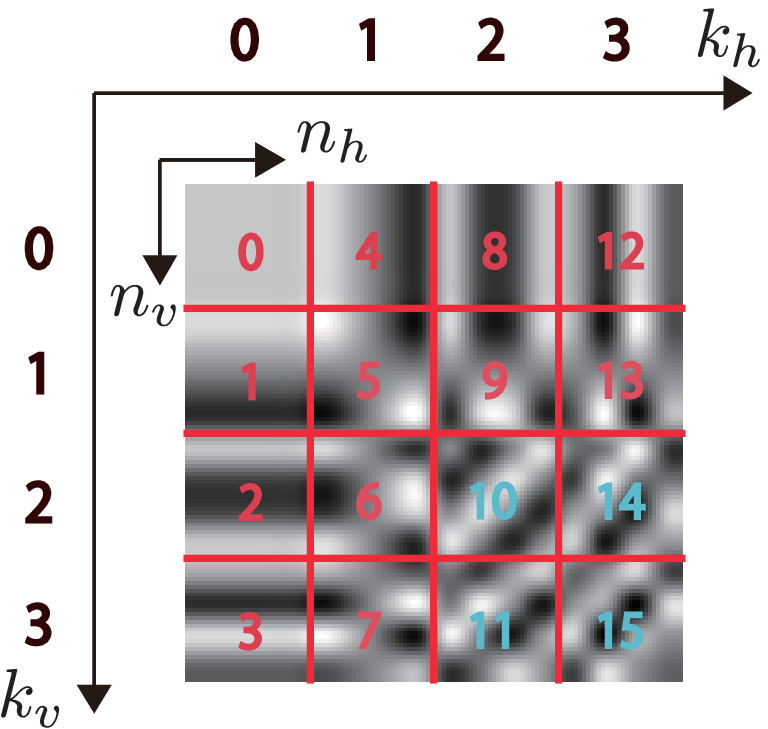}}
                \caption{${C}_{n_v,n_h}^{(k_v,k_h, 1)}$, ${B}_{n_v,n_h}^{(k_v,k_h, -1)}$}
        \end{subfigure}%
                \begin{subfigure}{0.45\linewidth}
        \centering
                \scalebox{0.4}{\includegraphics[keepaspectratio=true]{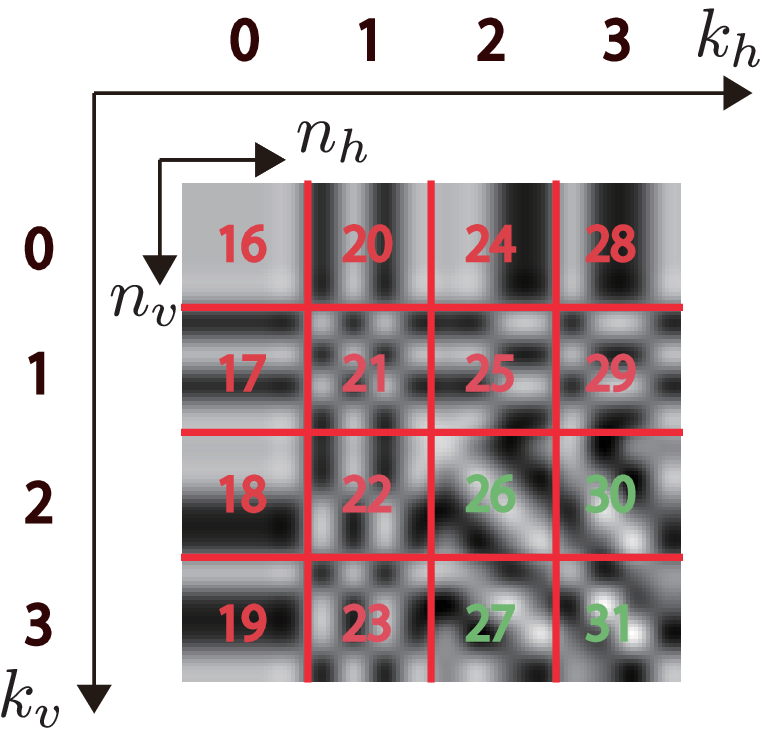}}
                                \caption{${S}_{n_v,n_h}^{(k_v,k_h, 1)}$, ${B}_{n_v,n_h}^{(k_v,k_h, 1)}$}
        \end{subfigure}%
\caption{(a) Procedure of the RDADCF ($M=4$). (b) and (c): Atoms ${C}_{n_v,n_h}^{(k_v,k_h)}$, ${S}_{n_v,n_h}^{(k_v,k_h)}$ (red), and ${B}_{n_v,n_h}^{(k_v,k_h, \pm 1)}$ (blue and green) in the RDADCF. The numbers indicate the rightmost subband indices in (a).}
\label{fig:2DRDDCF1}
\vspace{-0.3cm}
\end{figure}
\section{Experimental results}
\label{sec:experiments}
We evaluated the performance of the proposed DADCF pyramid (Section \ref{subsec:DC}) and RDADCF in compressive image sensing reconstruction \cite{Afonso2011}, as an example of image processing applications. $512\times 512$ pixel images in Fig. \ref{fig:expim} were used as the test set. Each incomplete observation ($\mathbf{y} = \mathrm{vec}(\mathbf{Y})  \in \mathbb{R}^{512^2}$) is obtained by Noiselet transform \cite{COIFMAN2001} ($\mathbf{\Phi} \in \mathbb{R}^{512^2 \times 512^2}$) followed by random sampling of 30\%, 40\%, 50\%, and 60\% pixels ($\mathbf{R}_{\mathrm{samp}} \in \mathbb{R}^{\mathcal{R}(512^2p) \times 512^2}$ where $\mathcal{R}$ is the rounding operator and $p = 0.3,\ 0.4,\  0.5,\ 0.6$) in the presence of additive white Gaussian noise ($\mathbf{n} \in \mathbb{R}^{512^2}$) with the standard derivation $\sigma =0.1$ as $\mathbf{y} = \mathbf{R}_{\mathrm{samp}}\mathbf{\Phi}\mathbf{x} + \mathbf{n}$, ($\mathbf{x}=\mathrm{vec}(\mathbf{X})  \in \mathbb{R}^{512^2}$). Figs. \ref{fig:expim}(a)--(d) indicate the estimated latent images by using the Moore-Penrose pseudo inverse of $\widetilde{\mathbf{\Phi}}^\dag=\mathbf{\Phi}^\top\mathbf{R}_{\mathrm{samp}}^\top$ ($\widetilde{\mathbf{\Phi}}=\mathbf{R}_{\mathrm{samp}}\mathbf{\Phi}$) in the case of $p = 0.3$. 

Up to now, many block transform-based methods for image recovery have been proposed, such as BM3D, patch-based redundant DCT approaches, and so on \cite{Danielyan2012, Fujita2015, Wang2017}.  For fair comparison, we simply evaluate directional block transforms in two image recovery problems: 
\begin{itemize}
\item{Problem 1:} image recovery based on sparsity of block-wise transformed coefficients.\\
\item{Problem 2:} image recovery based on sparsity of block-wise transformed coefficients and weighted total variation (WTV) for block boundaries as presented in \cite{Alter2004}. 
\end{itemize}
The cost function for these two problems is described as follows:
\begin{align} 
\label{eq:opt}
\mathbf{x}^{\star} =\argmin_{\mathbf{x} \in \mathbb{R}^{512^2}}& \|\mathbf{F}\mathbf{P}_{\mathrm{v2bv}}{\mathbf{x}}\|_1 + \rho \|\widetilde{\mathbf{W}}_{\mathrm{b}}\mathbf{D}_{\mathrm{hv}}\mathbf{x}\|_{1,2} \nonumber\\ &+ \iota_{C_{[0, 1]}} (\mathbf{x})+ F_{\mathbf{y}}(\widetilde{\mathbf{\Phi}}\mathbf{x}),
\end{align}
where $\mathbf{P}_{\mathrm{v2bv}}$ is the permutation matrix permuting a vectorized version of a matrix to a block-wise-vectorized one $\mathbf{P}_{\mathrm{v2bv}}{\mathbf{x}} = \mathrm{bvec}(\mathbf{X})$,  $\mathbf{F} = \mathbf{I}^{(\mathrm{V})}\otimes \mathbf{F}^{(\mathrm{DP})}$ or $\mathbf{F} = \mathbf{I}^{(\mathrm{V})}\otimes  \mathbf{F}^{(\mathrm{RD})}$ ($\mathbf{I}^{(\mathrm{V})} = \mathbf{I}_{512^2/M^2}$), and $\iota_{A}(\mathbf{x})$ is the indicator function\footnote{Indicator function of set $A$ is defined as $\iota_{A}(\mathbf{x}) = 0,\ (\mathbf{x} \in A), \quad \iota_{A}(\mathbf{x}) = \infty,\ (\mathbf{x} \notin A)$.} of a set $A$. $C_{[0, 1]}$ is the set of vectors whose entries are within $[0, 1]$. The data-fidelity function was set as $F_{\mathbf{y}} = \iota_{\mathcal{B}(\mathbf{y},\epsilon)}$ ($\mathcal{B}(\mathbf{y},\epsilon) := \{ \mathbf{x} \in \mathbb{R}^{M} | \| \mathbf{x} - \mathbf{y}\|_2 \leq \epsilon\}$) is the indicator function defined by the $\ell_2$-norm ball. The radius was set as $\epsilon = \|\mathbf{x}_o-\mathbf{y}\|_2$, where $\mathbf{x}_o$ is an original image. $\mathbf{D}_{\mathrm{hv}} = \begin{bmatrix} \mathbf{D}_{\mathrm{v}}^\top & \mathbf{D}_{\mathrm{h}}^\top \end{bmatrix}^\top \in \mathbb{R}^{(2\cdot 512^2) \times 512^2}$ denotes the vertical and horizontal difference operator. $\widetilde{\mathbf{W}}_b = \mathbf{I}^{(\mathrm{V})}\otimes \mathbf{W}_b$, where $\mathbf{W}_b \in \mathbb{R}^{M\times M}$ is the weighting matrix for block boundary as $[\mathbf{W}_b]_{m,n} = 0$ $(1\leq m,n \leq M-2)$,  $[\mathbf{W}_b]_{m,n} = 1$ (otherwise)
.  The cost functions with $\rho = 0$ and $\rho = 1$ correspond to Problem 1 and 2, respectively. The detailed algorithm used in the experiments is given in Appendix \ref{sec:AIR}. 

For comparison, we also used the DCT, the DFT, and the DHT in \eqref{eq:opt}. The block size is set to $M = 8,\ 16,\ 32$.

Fig. \ref{fig:expres} shows the resulting images of the proposed and conventional transforms obtained in the case of sampling rate $p = 0.4$. As these figures show (particularly in the dashed red boxes), the DCT cannot recover directional textures precisely. Table \ref{tab:error} shows the numerical results. In most cases, the DADCF pyramid or the RDADCF outperformed the DCT, the DFT, and the DHT in terms of the reconstruction error (PSNR). The RDADCF recovers the images better than the DADCF pyramid, especially for \textit{Monarch} and \textit{Parrot} (smooth images),  due to its regularity property. In fact, the DCT is superior to the DADCF pyramid and the RDADCF in some cases. However, since the DADCF pyramid and the RDADCF are compatible with the DCT, we can select the DCT, the DADCF pyramid, or the RDADCF by using or bypassing the DST/RDST and the SAP operations, depending on the input image.
\begin{figure}[t]
        \centering
                \begin{subfigure}{0.24\linewidth}
        \centering
                \scalebox{0.188}{\includegraphics[keepaspectratio=true]{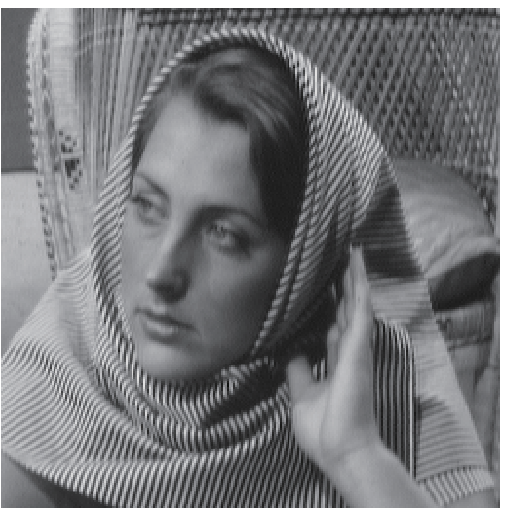}}
                \scalebox{0.188}{\includegraphics[keepaspectratio=true]{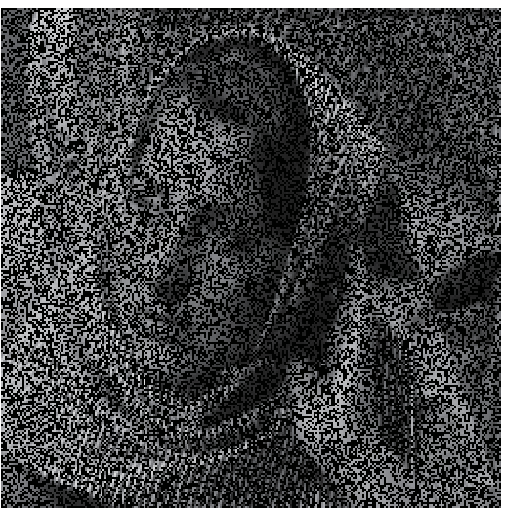}}
                \caption{\textit{Barbara}}
                \label{fig:texture}
        \end{subfigure}
                \begin{subfigure}{0.24\linewidth}
        \centering
                \scalebox{0.188}{\includegraphics[keepaspectratio=true]{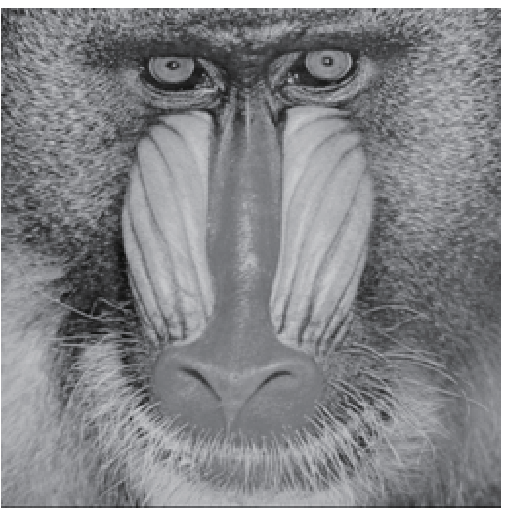}}
                \scalebox{0.188}{\includegraphics[keepaspectratio=true]{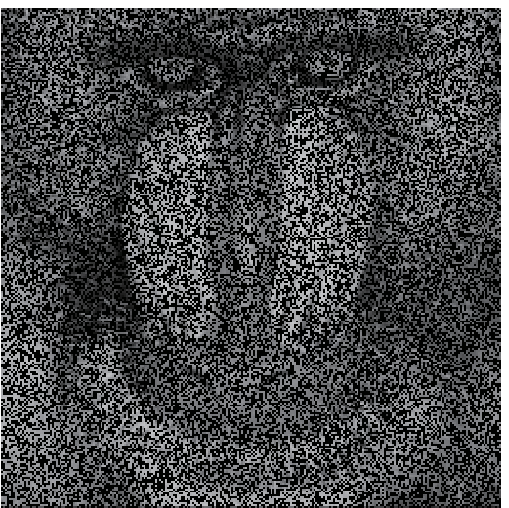}}        
                \caption{\textit{Mandrill}}
                \label{fig:house}
        \end{subfigure}
                \centering
                \begin{subfigure}{0.24\linewidth}
        \centering
                \scalebox{0.106}{\includegraphics[keepaspectratio=true]{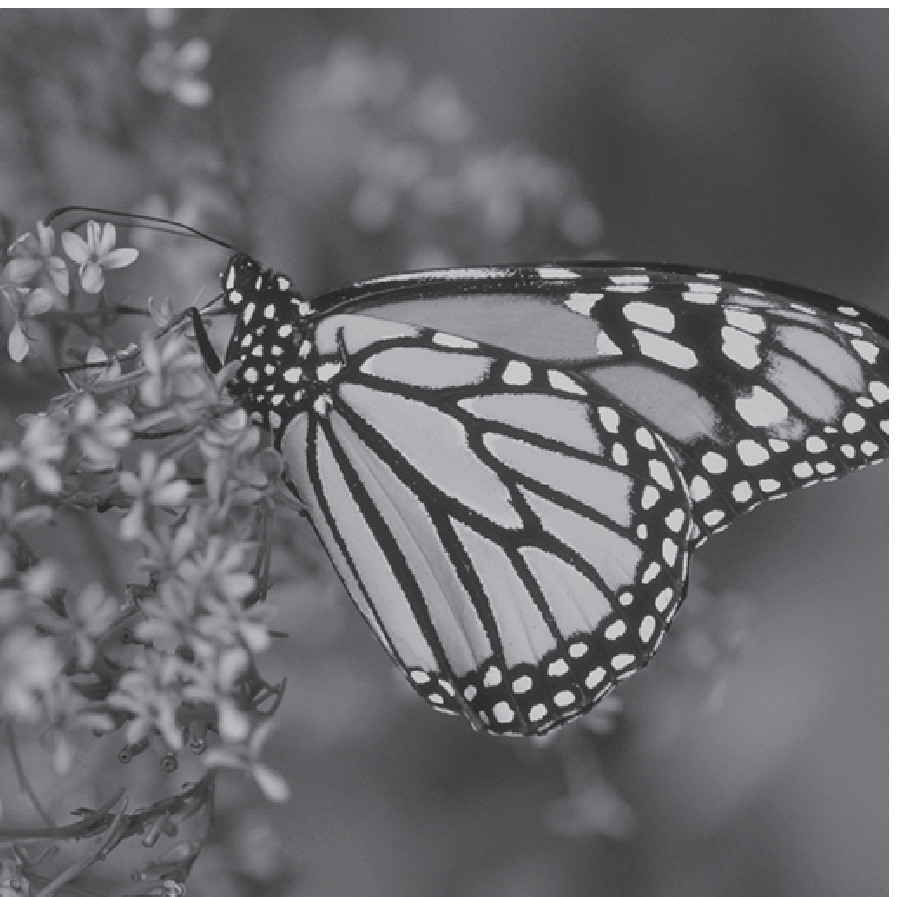}}
                \scalebox{0.106}{\includegraphics[keepaspectratio=true]{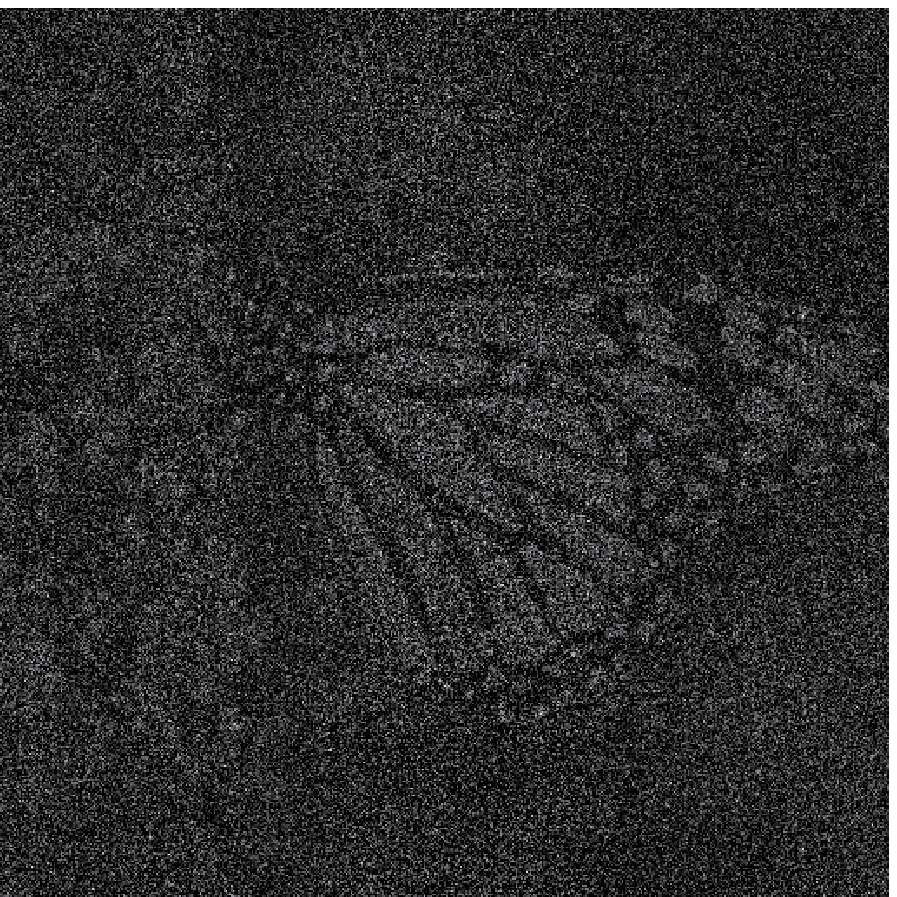}}
                \caption{\textit{Monarch}}
                \label{fig:texture}
        \end{subfigure}
                \centering
                \begin{subfigure}{0.24\linewidth}
        \centering
                \scalebox{0.106}{\includegraphics[keepaspectratio=true]{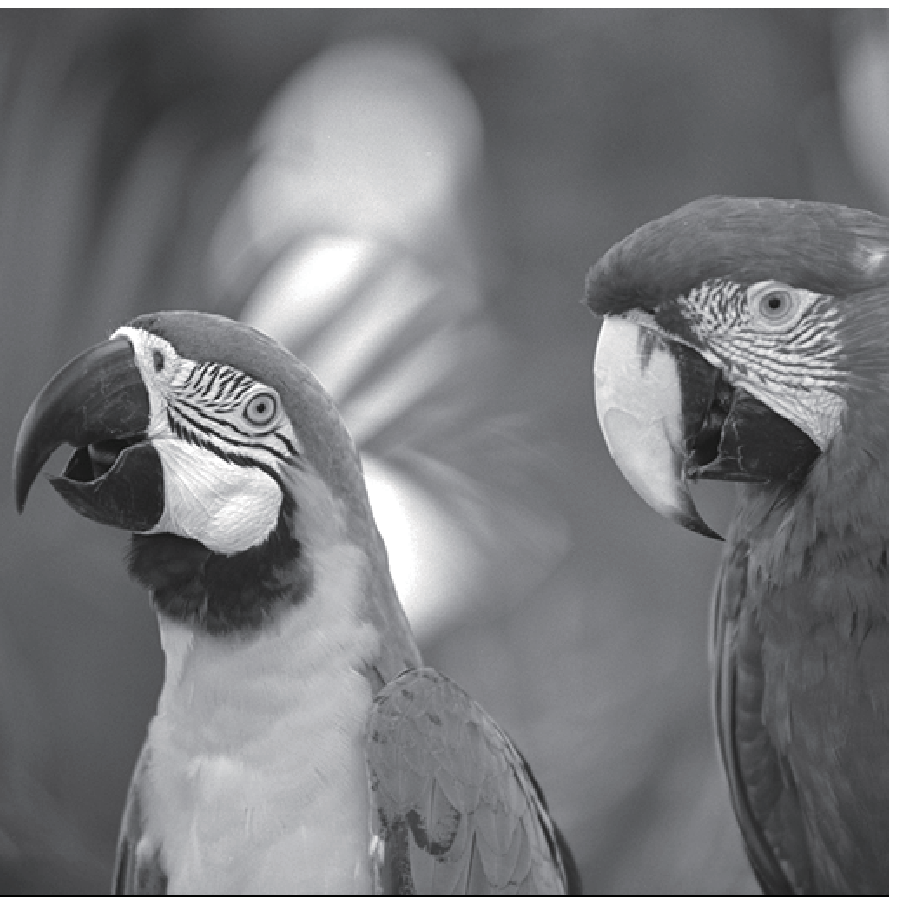}}
		 \scalebox{0.106}{\includegraphics[keepaspectratio=true]{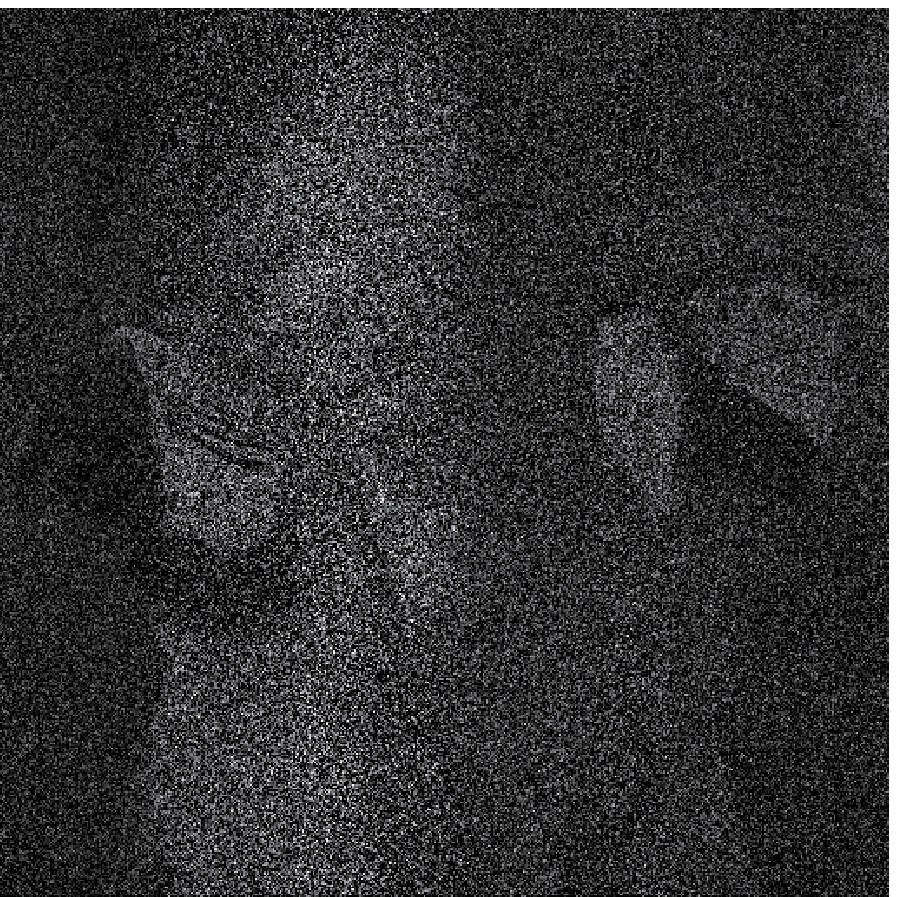}}
                \caption{\textit{Parrot}}
                \label{fig:house}
        \end{subfigure}
        \caption{(a)--(d): Original images  (256 $\times$ 256) and  recovered images by $\widetilde{\Phi}^\dag$ (sampling rate $p = 0.3$).}\label{fig:expim}
        \vspace{-0.2cm}
\end{figure}
\begin{figure}[t]
	\centering
	\begin{subfigure}{0.45\linewidth}
		\centering
               	\scalebox{0.4}{\includegraphics[keepaspectratio=true]{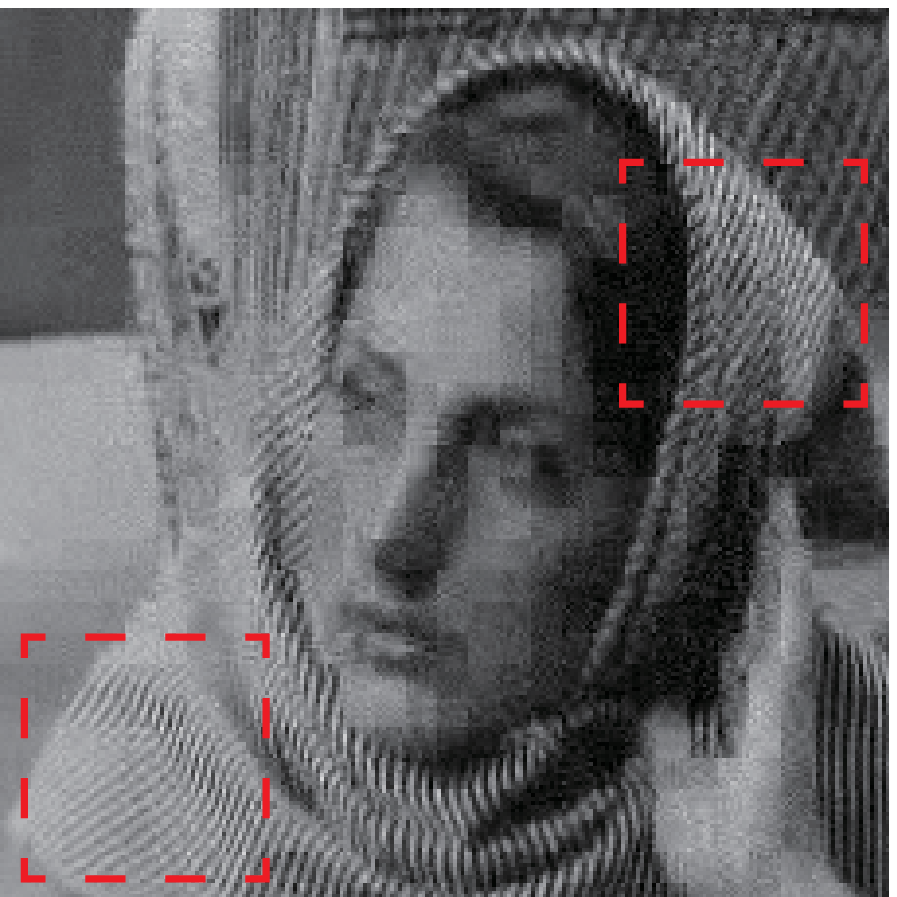}}
               	 \caption{DCT}
                	\label{fig:house}
       	\end{subfigure}
        \begin{subfigure}{0.45\linewidth}
        		\centering
                	\scalebox{0.4}{\includegraphics[keepaspectratio=true]{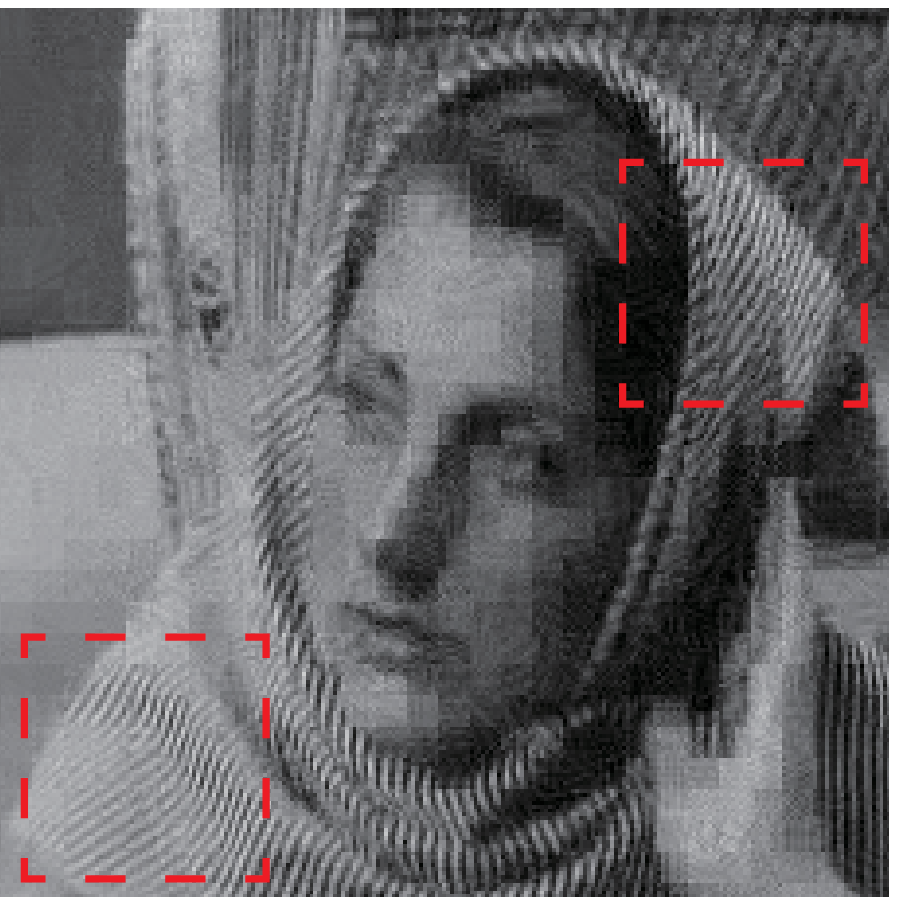}}
                	\caption{DFT}
                	\label{fig:barbara}
        \end{subfigure}
        \\
	\centering
                \begin{subfigure}{0.45\linewidth}
		\centering
                \scalebox{0.4}{\includegraphics[keepaspectratio=true]{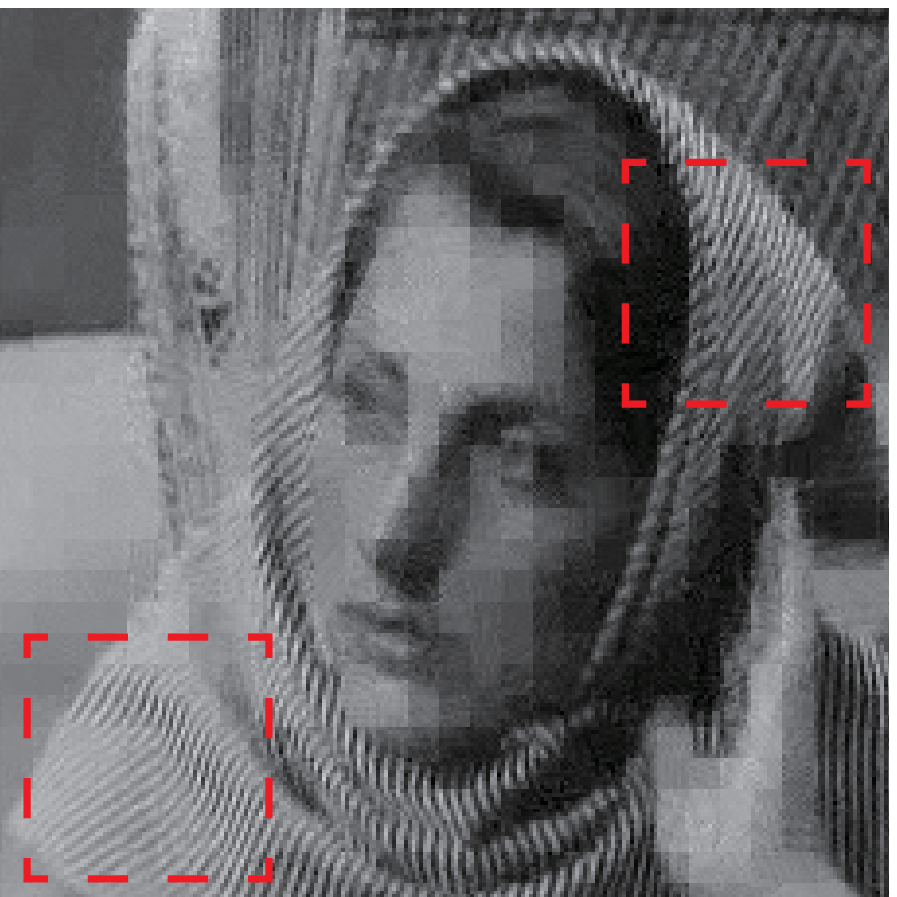}}
                \caption{DADCF}
                \label{fig:texture}
        \end{subfigure}
                \begin{subfigure}{0.45\linewidth}
        \centering
                \scalebox{0.4}{\includegraphics[keepaspectratio=true]{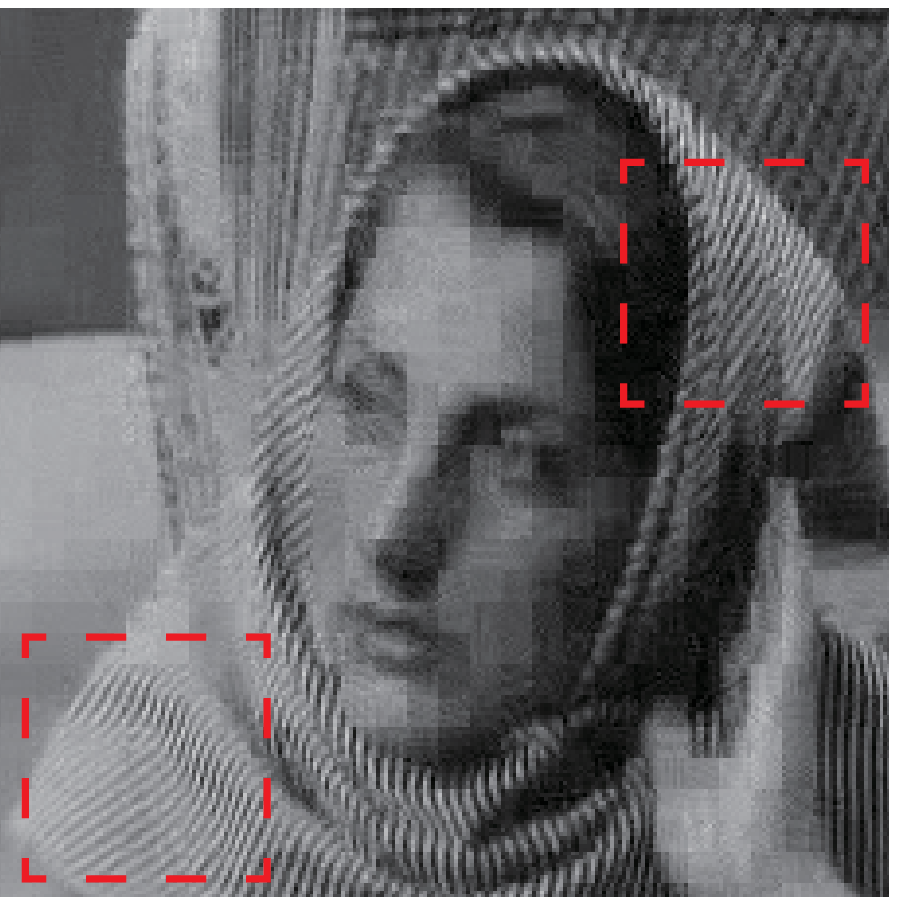}}
                \caption{RDADCF}
                \label{fig:house}
        \end{subfigure}
\\
        	\centering
	\begin{subfigure}{0.45\linewidth}
		\centering
               	\scalebox{0.4}{\includegraphics[keepaspectratio=true]{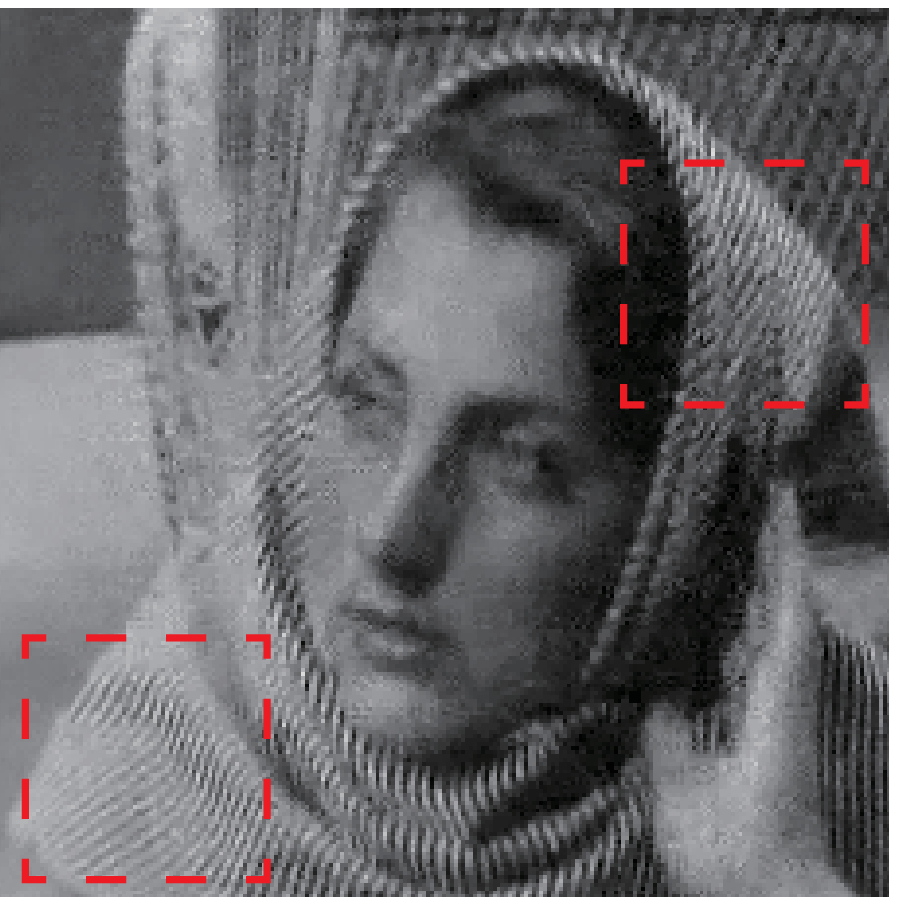}}
               	 \caption{DCT+WTV}
                	\label{fig:house}
       	\end{subfigure}
        \begin{subfigure}{0.45\linewidth}
        		\centering
                	\scalebox{0.4}{\includegraphics[keepaspectratio=true]{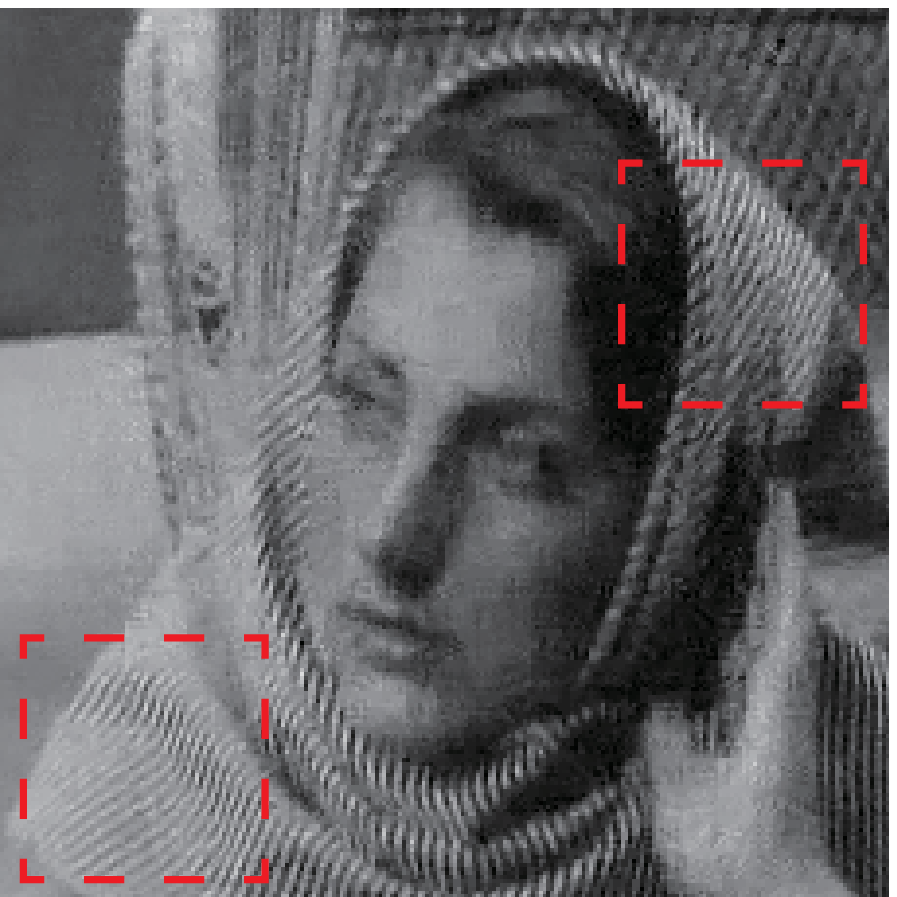}}
                	\caption{DFT+WTV}
                	\label{fig:barbara}
        \end{subfigure}
        \\
	\centering
                \begin{subfigure}{0.45\linewidth}
		\centering
                \scalebox{0.4}{\includegraphics[keepaspectratio=true]{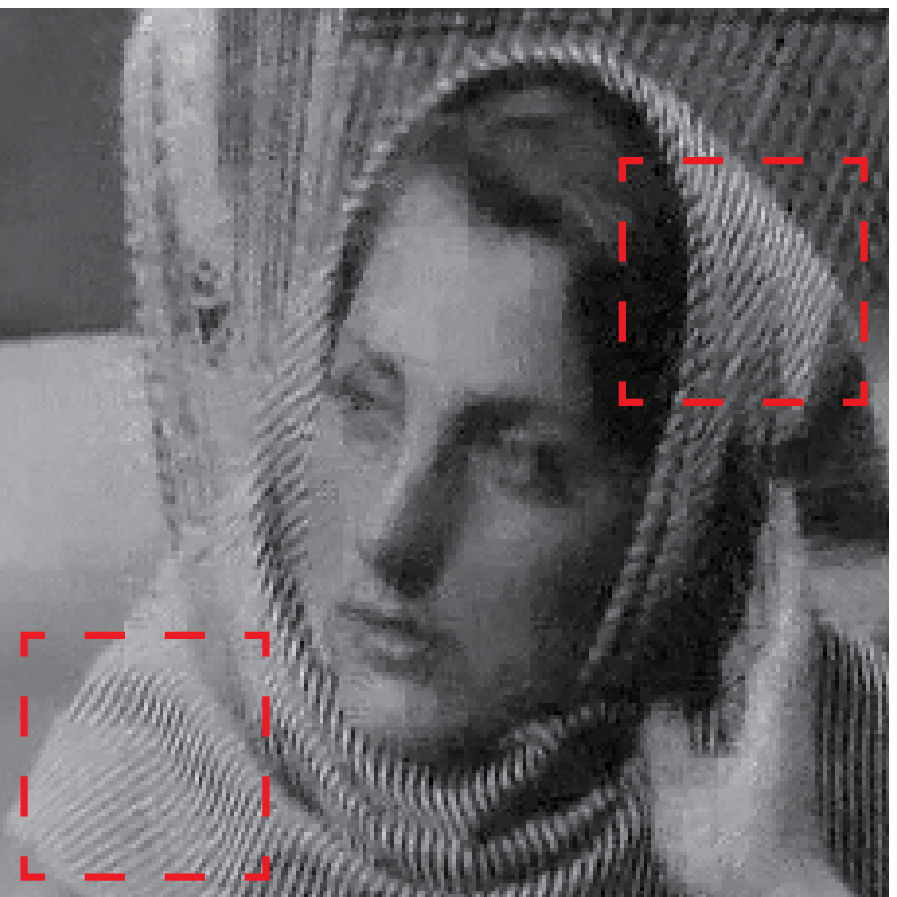}}
                \caption{DADCF+WTV}
                \label{fig:texture}
        \end{subfigure}
                \begin{subfigure}{0.45\linewidth}
        \centering
                \scalebox{0.4}{\includegraphics[keepaspectratio=true]{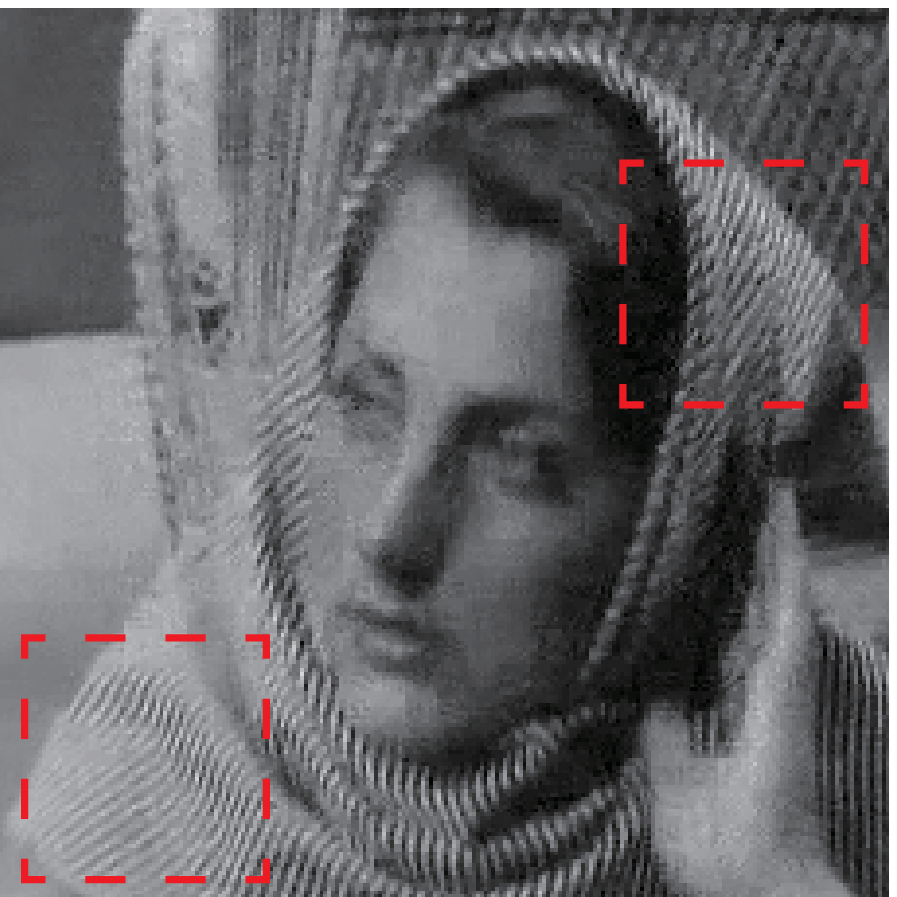}}
                \caption{RDADCF+WTV}
                \label{fig:house}
        \end{subfigure}
                  \caption{Zoomed resulting images reconstructed from 60\% noiselet coefficients ((a)--(d)) and 30\% noiselet coefficients ((e)--(h)). The size and the decomposition level of the transforms is $M=8$ and $J=2$, respectively.}\label{fig:expres}
        \vspace{-0.2cm}
\end{figure}
\begin{table*}[t]
\caption{Numerical results of compressive sensing reconstruction.}
\vspace{-0.2cm}
\label{tab:error}
\begin{center}
\scalebox{0.85}{
\begin{tabular}{cccccccccccccccc}
\thline
 \multicolumn{15}{c}{Problem 1: Image recovery based on sparsity of transformed coefficients}  \\ \hline
 \multicolumn{1}{c|}{} & \multicolumn{5}{c|}{Block size: $M=8$} &    \multicolumn{5}{c|}{Block size: $M=16$} & \multicolumn{5}{c}{Block size: $M=32$} \\ \hline
\multicolumn{1}{c|}{PSNR [dB]} & \multicolumn{1}{c}{DCT} &  \multicolumn{1}{c}{DFT} & \multicolumn{1}{c}{DHT}& \multicolumn{1}{c}{DADCFP} & \multicolumn{1}{c|}{RDADCF} &     \multicolumn{1}{c}{DCT} & \multicolumn{1}{c}{DFT} & \multicolumn{1}{c}{DHT} & \multicolumn{1}{c}{DADCFP} &  \multicolumn{1}{c|}{RDADCF} & \multicolumn{1}{c}{DCT} &  \multicolumn{1}{c}{DFT} & \multicolumn{1}{c}{DHT}& \multicolumn{1}{c}{DADCFP} & \multicolumn{1}{c}{RDADCF}\\\hline
\multicolumn{15}{c}{Image: \textit{Barbara}}\\
 \hline
\multicolumn{1}{c|}{60\%: 10.77} & 32.76 & 31.76 & 31.00 & 32.51 & \multicolumn{1}{c|}{\textbf{32.97}} & 33.81 & 32.54 & 31.71 & 33.14 & \multicolumn{1}{c|}{\textbf{33.89}} & 33.81 & 33.06 & 32.17 & 33.61 & \multicolumn{1}{c}{\textbf{34.19}} \\
\multicolumn{1}{c|}{50\%: 10.22} & 30.10 & 29.63 & 28.84 & 30.40 & \multicolumn{1}{c|}{\textbf{30.69}} & 31.43 & 30.47 & 29.61 & 31.16 & \multicolumn{1}{c|}{\textbf{31.85}} & 31.45 & 30.96 & 30.14 & 31.57 & \multicolumn{1}{c}{\textbf{32.14}} \\
\multicolumn{1}{c|}{40\%: 9.689} & 27.59 & 27.58 & 26.78 & 28.47 & \multicolumn{1}{c|}{\textbf{28.62}} & 29.10 & 28.44 & 27.62 & 29.25 & \multicolumn{1}{c|}{\textbf{29.91}} & 29.26 & 28.87 & 28.04 & 29.75 & \multicolumn{1}{c}{\textbf{30.25}} \\
\multicolumn{1}{c|}{30\%: 9.255} & 25.11 & 25.46 & 24.66 & \textbf{26.51} & \multicolumn{1}{c|}{26.46} & 26.70 & 26.38 & 25.58 & 27.25 & \multicolumn{1}{c|}{\textbf{27.78}} & 27.02 & 26.79 & 25.95 & 27.70 & \multicolumn{1}{c}{\textbf{28.24}} \\
\hline
\multicolumn{15}{c}{Image: \textit{Mandrill}}\\
 \hline
\multicolumn{1}{c|}{60\%: 10.77} & 25.68 & 25.95 & 25.29 & \textbf{26.64} & \multicolumn{1}{c|}{26.54} & 25.99 & 26.47 & 25.76 & 27.04 & \multicolumn{1}{c|}{\textbf{27.05}} & 25.99 & 26.66 & 25.94 & 27.25 & \multicolumn{1}{c}{\textbf{27.32}}  \\
\multicolumn{1}{c|}{50\%: 10.22} & 23.89 & 24.21 & 23.55 & \textbf{25.05} & \multicolumn{1}{c|}{24.88} & 24.27 & 24.78 & 24.08 & \textbf{25.44} & \multicolumn{1}{c|}{25.43} & 24.32 & 24.91 & 24.25 & 25.66 & \multicolumn{1}{c}{\textbf{25.70}}  \\
\multicolumn{1}{c|}{40\%: 9.689} & 22.18 & 22.57 & 21.97 & \textbf{23.59} & \multicolumn{1}{c|}{23.34} & 22.68 & 23.16 & 22.55 & \textbf{23.94} & \multicolumn{1}{c|}{23.91} & 22.77 & 23.30 & 22.68 & 24.15 & \multicolumn{1}{c}{\textbf{24.20}}  \\
\multicolumn{1}{c|}{30\%: 9.255} & 20.54 & 21.01 & 20.40 & \textbf{22.21} & \multicolumn{1}{c|}{21.87} & 21.20 & 21.68 & 21.15 & \textbf{22.56} & \multicolumn{1}{c|}{22.51} & 21.39 & 21.82 & 21.25 & 22.76 & \multicolumn{1}{c}{\textbf{22.80}} \\
\hline
\multicolumn{15}{c}{Image: \textit{Monach}}\\
 \hline
\multicolumn{1}{c|}{60\%: 10.77} & \textbf{37.43} & 35.40 & 34.37 & 36.31 & \multicolumn{1}{c|}{37.22} & \textbf{37.23} & 35.14 & 34.09 & 36.09 & \multicolumn{1}{c|}{36.77} & 36.52 & 35.39 & 34.36 & 36.19 & \multicolumn{1}{c}{\textbf{36.59}}  \\
\multicolumn{1}{c|}{50\%: 10.22} & 34.70 & 32.72 & 31.66 & 33.99 & \multicolumn{1}{c|}{\textbf{34.73}} & \textbf{34.61} & 32.62 & 31.52 & 33.78 & \multicolumn{1}{c|}{34.41} & 34.10 & 32.93 & 31.91 & 33.94 & \multicolumn{1}{c}{\textbf{34.31}}  \\
\multicolumn{1}{c|}{40\%: 9.689} & 31.65 & 29.88 & 28.77 & 31.61 & \multicolumn{1}{c|}{\textbf{32.06}} & 31.93 & 30.10 & 29.03 & 31.53 & \multicolumn{1}{c|}{\textbf{32.04}} & 31.66 & 30.60 & 29.57 & 31.88 & \multicolumn{1}{c}{\textbf{32.21}}  \\
\multicolumn{1}{c|}{30\%: 9.255} & 28.23 & 26.92 & 25.79 & 29.05 & \multicolumn{1}{c|}{\textbf{29.19}} & 29.04 & 27.44 & 26.44 & 29.07 & \multicolumn{1}{c|}{\textbf{29.44}} & 28.92 & 28.09 & 27.15 & 29.49 & \multicolumn{1}{c}{\textbf{29.77}}  \\
\hline
\multicolumn{15}{c}{Image: \textit{Parrot}}\\
 \hline
\multicolumn{1}{c|}{60\%: 10.77} & 39.25 & 38.03 & 37.28 & 38.46 & \multicolumn{1}{c|}{\textbf{39.40}} & \textbf{39.46} & 37.94 & 37.15 & 38.05 & \multicolumn{1}{c|}{39.35} & 39.10 & 37.96 & 37.11 & 37.98 & \multicolumn{1}{c}{\textbf{39.11}} \\
\multicolumn{1}{c|}{50\%: 10.22} & 37.13 & 35.91 & 34.95 & 36.67 & \multicolumn{1}{c|}{\textbf{37.58}} & 37.47 & 35.84 & 34.91 & 36.27 & \multicolumn{1}{c|}{\textbf{37.58}} & 37.07 & 35.85 & 34.90 & 36.13 & \multicolumn{1}{c}{\textbf{37.30}} \\
\multicolumn{1}{c|}{40\%: 9.689} & 34.57 & 33.48 & 32.49 & 34.54 & \multicolumn{1}{c|}{\textbf{35.35}} & 35.19 & 33.57 & 32.57 & 34.23 & \multicolumn{1}{c|}{\textbf{35.56}} & 34.78 & 33.63 & 32.62 & 34.11 & \multicolumn{1}{c}{\textbf{35.28}} \\
\multicolumn{1}{c|}{30\%: 9.255} & 31.60 & 30.79 & 29.69 & 32.30 & \multicolumn{1}{c|}{\textbf{32.85}} & 32.47 & 31.10 & 30.11 & 32.04 & \multicolumn{1}{c|}{\textbf{33.26}} & 32.21 & 31.09 & 30.07 & 31.88 & \multicolumn{1}{c}{\textbf{33.09}} \\
\thline
\\
\end{tabular}
}
\scalebox{0.85}{
\begin{tabular}{cccccccccccccccc}
\thline
 \multicolumn{15}{c}{Problem 2: Image recovery based on sparsity of transformed coefficients and weighted total variation for block boundaries}  \\ \hline
 \multicolumn{1}{c|}{} & \multicolumn{5}{c|}{Block size: $M=8$} &    \multicolumn{5}{c|}{Block size: $M=16$} & \multicolumn{5}{c}{ Block size: $M=32$} \\ \hline
\multicolumn{1}{c|}{PSNR [dB]} & \multicolumn{1}{c}{DCT} &  \multicolumn{1}{c}{DFT} & \multicolumn{1}{c}{DHT}& \multicolumn{1}{c}{DADCFP} & \multicolumn{1}{c|}{RDADCF} &  \multicolumn{1}{c}{DCT} & \multicolumn{1}{c}{DFT} & \multicolumn{1}{c}{DHT} & \multicolumn{1}{c}{DADCFP} &  \multicolumn{1}{c|}{RDADCF} & \multicolumn{1}{c}{DCT} &  \multicolumn{1}{c}{DFT} & \multicolumn{1}{c}{DHT}& \multicolumn{1}{c}{DADCFP} & \multicolumn{1}{c}{RDADCF}\\\hline
\multicolumn{15}{c}{Image: \textit{Barbara}}\\
 \hline
\multicolumn{1}{c|}{60\%: 10.77} & 32.27 & 31.98 & 31.48 & 32.32 & \multicolumn{1}{c|}{\textbf{32.62}} & 33.15 & 32.60 & 31.93 & 32.96 & \multicolumn{1}{c|}{\textbf{33.47}} & 33.38 & 33.15 & 32.30 & 33.49 & \textbf{33.97}\\
\multicolumn{1}{c|}{50\%: 10.22} & 29.95 & 30.00 & 29.52 & 30.29 & \multicolumn{1}{c|}{\textbf{30.53}} & 30.89 & 30.69 & 29.98 & 31.07 & \multicolumn{1}{c|}{\textbf{31.54}} & 31.09 & 31.14 & 30.32 & 31.53 & \textbf{31.99}\\
\multicolumn{1}{c|}{40\%: 9.689} & 27.91 & 28.17 & 27.70 & 28.49 & \multicolumn{1}{c|}{\textbf{28.64}} & 28.76 & 28.89 & 28.22 & 29.27 & \multicolumn{1}{c|}{\textbf{29.69}} & 29.05 & 29.24 & 28.45 & 29.78 & \textbf{30.16}\\
\multicolumn{1}{c|}{30\%: 9.255} & 26.06 & 26.42 & 26.01 & 26.74 & \multicolumn{1}{c|}{\textbf{26.82}} & 26.88 & 27.11 & 26.49 & 27.42 & \multicolumn{1}{c|}{\textbf{27.82}} & 27.11 & 27.30 & 26.53 & 27.86 & \textbf{28.30}\\
\hline
\multicolumn{15}{c}{Image: \textit{Mandrill}}\\
 \hline
\multicolumn{1}{c|}{60\%: 10.77} & 26.85 & 26.95 & 26.59 & \textbf{27.23} & \multicolumn{1}{c|}{27.21} & 26.46 & 26.91 & 26.37 & \textbf{27.26} & \multicolumn{1}{c|}{27.25} & 26.20 & 26.86 & 26.20 & 27.29 & \textbf{27.34}\\
\multicolumn{1}{c|}{50\%: 10.22} & 25.24 & 25.42 & 25.07 & \textbf{25.68} & \multicolumn{1}{c|}{25.65} & 24.89 & 25.35 & 24.82 & \textbf{25.67} & \multicolumn{1}{c|}{25.67} & 24.61 & 25.22 & 24.63 & 25.76 & \textbf{25.79}\\
\multicolumn{1}{c|}{40\%: 9.689} & 23.84 & 24.03 & 23.71 & \textbf{24.31} & \multicolumn{1}{c|}{24.25} & 23.47 & 23.90 & 23.43 & 24.25 & \multicolumn{1}{c|}{\textbf{24.26}} & 23.21 & 23.73 & 23.19 & 24.29 & \textbf{24.33}\\
\multicolumn{1}{c|}{30\%: 9.255} & 22.46 & 22.67 & 22.41 & \textbf{22.96} & \multicolumn{1}{c|}{22.86} & 22.16 & 22.58 & 22.18 & \textbf{22.92} & \multicolumn{1}{c|}{22.90} & 21.91 & 22.38 & 21.92 & 22.95 & \textbf{22.99}\\
\hline
\multicolumn{15}{c}{Image: \textit{Monach}}\\
 \hline
\multicolumn{1}{c|}{60\%: 10.77} & \textbf{40.60} & 39.55 & 39.18 & 39.49 & \multicolumn{1}{c|}{40.10} & \textbf{39.26} & 38.01 & 37.29 & 38.08 & \multicolumn{1}{c|}{38.67} & \textbf{37.63} & 36.96 & 36.09 & 37.20 & 37.59\\
\multicolumn{1}{c|}{50\%: 10.22} & \textbf{38.70} & 37.49 & 37.04 & 37.50 & \multicolumn{1}{c|}{38.13} & \textbf{37.01} & 35.68 & 34.89 & 35.87 & \multicolumn{1}{c|}{36.48} & \textbf{35.34} & 34.58 & 33.66 & 34.95 & 35.34\\
\multicolumn{1}{c|}{40\%: 9.689} & \textbf{36.49} & 35.12 & 34.62 & 35.21 & \multicolumn{1}{c|}{35.82} & \textbf{34.69} & 33.27 & 32.46 & 33.56 & \multicolumn{1}{c|}{34.12} & 33.04 & 32.23 & 31.37 & 32.85 & \textbf{33.17}\\
\multicolumn{1}{c|}{30\%: 9.255} & \textbf{33.76} & 32.33 & 31.84 & 32.58 & \multicolumn{1}{c|}{33.13} & \textbf{31.97} & 30.40 & 29.64 & 30.89 & \multicolumn{1}{c|}{31.44} & 30.39 & 29.59 & 28.71 & 30.34 & \textbf{30.64}\\
\hline
\multicolumn{15}{c}{Image: \textit{Parrot}}\\
 \hline
\multicolumn{1}{c|}{60\%: 10.77} & \textbf{40.83} & 40.24 & 39.91 & 40.20 & \multicolumn{1}{c|}{40.71} & \textbf{40.34} & 39.56 & 39.00 & 39.35 & \multicolumn{1}{c|}{40.19} & 39.66 & 39.07 & 38.37 & 38.84 & \textbf{39.69}\\
\multicolumn{1}{c|}{50\%: 10.22} & \textbf{39.33} & 38.73 & 38.34 & 38.68 & \multicolumn{1}{c|}{39.27} & \textbf{38.72} & 37.85 & 37.19 & 37.72 & \multicolumn{1}{c|}{38.62} & 37.89 & 37.21 & 36.40 & 37.09 & \textbf{38.02} \\
\multicolumn{1}{c|}{40\%: 9.689} & 37.43 & 36.83 & 36.41 & 36.84 & \multicolumn{1}{c|}{\textbf{37.45}} & 36.77 & 35.91 & 35.21 & 35.77 & \multicolumn{1}{c|}{\textbf{36.77}} & 35.77 & 35.17 & 34.32 & 35.09 & \textbf{36.07}\\
\multicolumn{1}{c|}{30\%: 9.255} & 35.32 & 34.69 & 34.25 & 34.70 & \multicolumn{1}{c|}{\textbf{35.32}} & 34.55 & 33.71 & 33.00 & 33.64 & \multicolumn{1}{c|}{\textbf{34.64}} & 33.47 & 32.77 & 31.91 & 32.85 & \textbf{33.95}\\
\thline
\end{tabular}
}
\end{center}
\vspace{-0.4cm}
\end{table*}

\section{Concluding Remarks}
\label{sec:conclusion}
In this paper, we proposed the DADCF and the RDADCF for directional image representation by extending the DCT. Since they are Parseval block frames with low redundancy, they can deliver computational efficiency for practical situations. Furthermore, unlike the conventional directional block transforms, they can finely decompose the frequency plane and provide richer directional atoms. Comparing both the DADCF and the RDADCF, the DADCF provides richer directional selectivity than the RDADCF. However, in practice, the slightly redundant DADCF pyramid should be used instead of the DADCF to avoid the DC leakage and perform good image processing, i.e., the RDADCF can save more amount of memory usage than the DADCF (pyramid). Also, they can be easily implemented by appending trivial operations (the DST or the RDST, and the SAP operations) to the DCT. Moreover, the DST can be implemented based on the DCT with the permutation and sign-alternation operations and the RDST based on the DST and one additional orthogonal matrix with the size of the half. Since the DCT is integrated into many existing digital devices, the system modification for the proposed method is minimal. Note that the DADCF and the RDADCF are compatible with the DCT. Depending on applications, we can switch the DCT/DADCF/RDADCF by using the DST/RDST and SAP operations.

We evaluated the DADCF pyramid and the RDADCF in compressive image sensing reconstruction as a practical application. The experimental results showed that, for both fine textures and smooth images, they could achieve higher numerical qualities than the DCT, the DFT, and the DHT. Furthermore, it was shown that the RDADCF could recover
smooth regions better than the DADCF pyramid due to its regularity property.
\appendices
\section{Proof for Proposition 2}\label{ap:prop:S}
\begin{proof}
We first introduce the following lemma.
\begin{lem}\label{lem:S}\mbox{}\\
\begin{enumerate}
\item The elements in the upper-right triangle $[\mathbf{S}\mathbf{S}^{\top}]_{k_v,k_h} = [\langle \mathbf{s}_{k_v},  \mathbf{s}_{k_h} \rangle]_{k_v,k_h}$ are expressed as
\begin{align*}
\langle \mathbf{s}_{k_v},  \mathbf{s}_{k_h} \rangle = \begin{cases}
\frac{\sqrt{2}}{M\sin\left(\frac{\pi}{2M}k_h\right) } & (k_v = 0\ \mathrm{and}\ k_h = 2\ell+1) \\
1 & (k_v=k_h) \\
0 & (\mathrm{otherwise})
\end{cases},
\end{align*}
where $\ell \in \Omega_{\frac{M}{2}-1}$. For example, for $M = 4$, 
\begin{align}
\mathbf{S}\mathbf{S}^{\top} \approx 
\begin{bmatrix}
    1  &  0.9239  &  0  &  0.3827 \\
    0.9239  &  1  &       0  & 0 \\
    0  &       0  &  1  &  0 \\
    0.3827  & 0  &  0  &  1
\end{bmatrix}.
\end{align}
\item $\sum^{M/2-1}_{\ell=0} \langle \mathbf{s}_0,  \mathbf{s}_{2\ell+1} \rangle ^2 = 1$.
\end{enumerate}
\end{lem}
\begin{proof}
\begin{enumerate}
\item  It is clear that $\langle \mathbf{s}_0,  \mathbf{s}_0 \rangle = 1$ and $\langle \mathbf{s}_{k_v},  \mathbf{s}_{k_h} \rangle = \delta(k_v-k_h)$ for $k_v, k_h \in \Omega_{1,M-1}$ because $\{\mathbf{s}_k\}_{k=1}^{M-1}$ are the rows of the DST $\mathbf{F}^{(\mathrm{S})}$. In the other cases, it is clear from Proposition \ref{lem:S1}.
\item $
\sum^{M/2-1}_{\ell=0} \langle \mathbf{s}_0,  \mathbf{s}_{2\ell+1} \rangle ^2 = \left\|\mathbf{F}^{(\mathrm{S})} \frac{1}{\sqrt{M}}\mathbf{1}\right\|_2^2 = 1.
$
\end{enumerate}
\end{proof}
Let $\widehat{\mathbf{S}}$ be $\widehat{\mathbf{S}} = \mathbf{S}\mathbf{S}^{\top} = \begin{bmatrix}
\widehat{\mathbf{s}}_0 & \widehat{\mathbf{s}}_1 & \cdots & \widehat{\mathbf{s}}_{M-1}
\end{bmatrix}^{\top}$. From Lemma \ref{lem:S} 2),
$
\widehat{\mathbf{s}}_0 - \sum^{M/2-1}_{\ell=0}\langle \mathbf{s}_0,  \mathbf{s}_{2\ell+1}\rangle\widehat{\mathbf{s}}_{2\ell+1} = \mathbf{0}.
$
This implies that $\mathrm{rank}(\widehat{\mathbf{S}}) = M-1$, and so $\mathrm{rank}(\mathbf{S}) = M-1$.
\end{proof}
\section{Proof for Proposition 3}\label{ap:prop:rankS1}
\begin{IEEEproof}
From Proposition \ref{lem:S}, the elements in the upper-right triangle of $\widehat{\mathbf{S}}^{(0)}=\mathbf{S}^{(0)}\mathbf{S}^{(0)\top} = \begin{bmatrix}
\widehat{\mathbf{s}}_0^{(0)}  &  \cdots & \widehat{\mathbf{s}}_{M-1}^{(0)}
\end{bmatrix}^{\top}$ are as follows:
 \begin{align*}
[\widehat{\mathbf{S}}^{(0)}]_{k_v,k_h}=
\begin{cases}
\frac{\sqrt{2}}{M\sin\left(\frac{\pi}{2M}k_h\right) } & (k_v = 0\ \mathrm{and}\ k_h = 2\ell+1) \\
1 & (k_v=k_h\  \mathrm{and}\  k_v \neq 1) \\
0 & (\mathrm{otherwise})
\end{cases},
\end{align*}
where $1\leq \ell \leq \frac{M}{2}-1$. For example, for $M = 4$, 
\begin{align}
\widehat{\mathbf{S}}^{(0)} \approx 
\begin{bmatrix}
    1  &  0  &  0  &  0.3827 \\
    0  &  0  &       0  & 0 \\
    0  &       0  &  1  &  0 \\
    0.3827  & 0  &  0  &  1
\end{bmatrix}.
\end{align}
Then, we can derive that
\begin{align}
\widehat{\mathbf{s}}_0^{(0)} - \sum^{M/2-1}_{\ell=1}\langle \mathbf{s}_0,  \mathbf{s}_{2\ell+1}\rangle\widehat{\mathbf{s}}_{2\ell+1}^{(0)} \neq \mathbf{0}.
\end{align}
Thus, it can be concluded that $\mathrm{rank}(\mathbf{S}^{(0)}) = M-1$.
\end{IEEEproof}
\section{Proof for Proposition 4}\label{ap:prop:RDST}
\begin{proof}
The statements 1) and 2) are clearly true. We only show the proof for 3). 

For any even row $2k \geq 2$, $\mathcal{F}[[\mathbf{F}^{(\mathrm{RS})}]_{2k,\cdot}]$ is exactly the same as $\mathcal{F}[[\mathbf{F}^{(\mathrm{S})}]_{2k,\cdot}]$. Thus, it is enough to show the case of odd rows $[\mathbf{F}^{(\mathrm{RS})}]_{2k+1,\cdot}$ $(2k+1 \geq 3)$.

$[\mathbf{F}^{(\mathrm{RS})}]_{2k+1,n}$ is the same as $[\mathbf{s}^{(k+1)}_{2k+1}]_n$ of $\mathbf{S}^{(k+1)} = \begin{bmatrix}
\mathbf{s}_{0}^{(k+1)}&  \ldots & \mathbf{s}_{M-1}^{(k+1)} \end{bmatrix}^{\top}$ in Step 5 of Algorithm \ref{alg:RDST}. Let $\mathbf{T}^{(k)} = \begin{bmatrix}
\mathbf{t}_{0}^{(k)}&  \ldots & \mathbf{t}_{M-1}^{(k)} \end{bmatrix}$ be the inverse matrix of $\mathbf{S}^{(k)}$. Since $\mathbf{s}^{(k+1)}_{2k+1}$ is designed to be orthogonal to $\{\mathbf{s}^{(k)}_{n}\}_{\Omega_{M-1}\setminus \{2k+1\}}$ in Algorithm \ref{alg:RDST}, $\mathbf{s}^{(k+1)}_{2k+1}$ can be expressed with a linear combination of  $\{\mathbf{t}^{(k)}_{n}\}_{\Omega_{M-1}}$ as:
\begin{align}
\mathbf{s}^{(k+1)}_{2k+1} =& \mathbf{T}^{(k)}\mathbf{S}^{(k)}\mathbf{s}^{(k+1)}_{2k+1} = \sum_{m = 0}^{M-1} \langle \mathbf{s}^{(k)}_{m} , \mathbf{s}^{(k+1)}_{2k+1} \rangle \mathbf{t}^{(k)}_{m} \nonumber\\ = & \langle \mathbf{s}^{(k)}_{2k+1} , \mathbf{s}^{(k+1)}_{2k+1} \rangle \mathbf{t}^{(k)}_{2k+1}.
\end{align}
Here, we use the following lemma (see its proof in Appendix \ref{ap:lem:pbT}).
\begin{lem}\label{lem:pbT}
The passband of the frequency response of  $\mathbf{t}^{(k)}_{\ell}$ of $\mathbf{T}^{(k)} = \begin{bmatrix}
\mathbf{t}_{0}^{(k)}&  \ldots & \mathbf{t}_{M-1}^{(k)} \end{bmatrix}$ is the same as that of $\mathbf{s}^{(k)}_{\ell}$ of $\mathbf{S}^{(k)} = \begin{bmatrix}
\mathbf{s}_{0}^{(k)}&  \ldots & \mathbf{s}_{M-1}^{(k)} \end{bmatrix}^{\top}$.
\end{lem}

From Lemma \ref{lem:pbT}, the passband of the frequency response of $\mathbf{s}^{(k+1)}_{2k+1}$ is located at the same position as that of $\mathbf{t}^{(k)}_{2k+1}$ and $\mathbf{s}^{(k)}_{2k+1} = ([\mathbf{F}^{(\mathrm{S})}]_{2k+1,\cdot})^{\top}$. Consequently, we conclude that statement 3) is true.
\end{proof}

\section{Proof for Lemma 2}\label{ap:lem:pbT}
\begin{proof}
First, consider the case of $k=1$. Since $\mathbf{T}^{(1)} = \begin{bmatrix} \mathbf{t}_{0}^{(1)}&  \ldots & \mathbf{t}_{M-1}^{(1)} \end{bmatrix}$ is the inverse of $\mathbf{S}^{(1)} = \begin{bmatrix} \mathbf{s}_{0}^{(1)}&  \ldots & \mathbf{s}_{M-1}^{(1)} \end{bmatrix}^{\top}$, each $\mathbf{t}_{m}^{(1)}$ can be represented as $\mathbf{t}_{m}^{(1)} = \sum_{n = 0}^{M-1} \langle \mathbf{t}_{n}^{(1)} , \mathbf{t}_{m}^{(1)} \rangle \mathbf{s}_{n}^{(1)}$. Therefore, it is enough to show that $|[(\mathbf{T}^{(1)})^{\top}\mathbf{T}^{(1)}]_{m,n}| = |\langle \mathbf{t}_{m}^{(1)} , \mathbf{t}_{n}^{(1)} \rangle | \ll |\langle \mathbf{t}_{m}^{(1)} , \mathbf{t}_{m}^{(1)} \rangle |$. 

For that, we consider the eigenvalue decomposition of $\mathbf{S}^{(1)}\mathbf{S}^{(1)\top}= \mathbf{U}^{(1)}\mathbf{D}^{(1)}\mathbf{U}^{(1)\top}$, then calculate $\mathbf{T}^{(1)\top}\mathbf{T}^{(1)} = \mathbf{U}^{(1)}(\mathbf{D}^{(1)})^{-1}\mathbf{U}^{(1)\top}$, where $\mathbf{U}^{(1)} = \begin{bmatrix} \mathbf{u}_{0}&  \ldots & \mathbf{u}_{M-1}\end{bmatrix}$ and $\mathbf{D}^{(1)} = \mathrm{diag}(\lambda_0,\ldots,\lambda_{M-1})$ are some orthogonal and diagonal matrices consisting of eigenvectors and eigenvalues, respectively. Similar to Lemma \ref{lem:S}, we can derive $\widehat{\mathbf{S}}^{(1)} = \mathbf{S}^{(1)}\mathbf{S}^{(1)\top}$ forms
\begin{align}\label{eq:Shat1}
\widehat{\mathbf{S}}^{(1)} =&\  
\begin{bmatrix}
1 & 0 & 0 & \widehat{s}_{3} & 0 & \widehat{s}_{5} & \cdots \\
0 & 1 & 0 & 0 & 0 & 0 & \cdots \\
0 & 0 & 1 & 0 & 0 & 0 &  \cdots \\
 \widehat{s}_{3} & 0 & 0 & 1 & 0 & 0 &  \cdots \\
 0 & 0 & 0 & 0 & 1 & 0 &  \cdots \\
   \widehat{s}_{5} & 0 & 0 & 0 & 0 & 1 &  \cdots \\
  \vdots & \vdots & \vdots & \vdots & \vdots & \vdots &  \ddots \\
\end{bmatrix},
\end{align} 
where $ \widehat{s}_{2\ell+1} = \frac{\sqrt{2}}{M\sin\left(\frac{\pi}{2M}(2\ell+1)\right) }\  (\ell \geq 1)$. Let us consider some eigenvalue $\lambda$ and its corresponding eigenvectors $\mathbf{u}$ of $\widehat{\mathbf{S}}^{(1)}$. Note that all the eigenvalues are positive $\lambda_n > 0$, since $\mathrm{rank}(\widehat{\mathbf{S}}^{(1)}) = M$. Suppose an eigenvalue $\lambda = 1$, then its eigenvector $\mathbf{u} = \begin{bmatrix} {u}_{0}&  \ldots & {u}_{M-1}\end{bmatrix}^{\top}$ should satisfy
\begin{align}
&\widehat{\mathbf{S}}^{(1)} \begin{bmatrix} {u}_{0}&  \ldots & {u}_{M-1} \end{bmatrix}^{\top} =   \begin{bmatrix} {u}_{0}&  \ldots & {u}_{M-1}\end{bmatrix}^{\top} \nonumber \\
 \Rightarrow & \begin{cases} {u}_{3} = - \sum_{\ell = 2}^{M/2-1} \frac{\widehat{s}_{2\ell+1}}{\widehat{s}_{3}}{u}_{2\ell+1}  \\
{u}_{0} = 0\end{cases} \nonumber \\
\Rightarrow & \mathbf{u} = \begin{bmatrix} 0 & u_1 & u_2 & - \sum_{\ell = 2}^{M/2-1} \frac{\widehat{s}_{2\ell+1}}{\widehat{s}_{3}}{u}_{2\ell+1} & {u}_{4} & \cdots  \end{bmatrix}^{\top} \nonumber \\
 \Rightarrow & \mathbf{u}\in \mathrm{span}\left\{\mathbf{u}_{1} ,\ldots, \mathbf{u}_{M-2} \right\},
\end{align}
where
\begin{align}
\mathbf{u}_m = &\ 
\begin{cases}
{\bm \delta_{m}} & (m = 1, 2) \\ 
{\bm \delta_{m+1}} & (m = 2\ell+1,  m \geq 3) \\ 
{\bm \delta_{m+1}}  - \frac{\widehat{s}_{m+1}}{\widehat{s}_{3}}{\bm \delta_{3}}& (m = 2\ell,  m \geq 3) \\ 
\end{cases},
\end{align}
where ${\bm \delta_{m}} \in \mathbb{R}^{M}$ $(m \in \Omega_{M-1})$ consists of $[{\bm \delta}_m]_m = 1$ and $[{\bm \delta}_m]_n = 0$ $(m \neq n)$. Since $ \mathbf{U}^{(1)}$ should be an orthogonal matrix, but the vectors $\{\mathbf{u}_m\}_{m = 2\ell,  m \geq 3}$ are not orthogonal yet, Gram-Schmidt orthonormalization is applied to them. 

Next, we consider the case of $\lambda \neq 1$. 
\begin{align}
&\widehat{\mathbf{S}}^{(1)} \begin{bmatrix} {u}_{0}&  \ldots & {u}_{M-1} \end{bmatrix}^{\top} =   \lambda\begin{bmatrix} {u}_{0}&  \ldots & {u}_{M-1}\end{bmatrix}^{\top} \nonumber \\
\Rightarrow & \begin{cases} \sum_{\ell = 1}^{M/2-1} \widehat{s}_{2\ell+1}  {u}_{2\ell+1}=  (\lambda -1 ){u}_{0}\ (u_0 \neq 0)\\
 \widehat{s}_{2\ell+1}{u}_{0} =   (\lambda -1 ){u}_{2\ell+1} \\  (\lambda -1 ){u}_{\ell} = 0\end{cases} \nonumber \\
 \Rightarrow & \lambda^2 -2\lambda + \left(1 -\sum_{\ell = 1}^{M/2-1} \widehat{s}_{2\ell+1}^2\right) = 0 
\end{align}
Thus, $\lambda_0 = 1 +  \sqrt{\sum_{\ell = 1}^{M/2-1} \widehat{s}_{2\ell+1}^2},\  \lambda_{M-1} = 1- \sqrt{\sum_{\ell = 1}^{M/2-1} \widehat{s}_{2\ell+1}^2}$. For $\lambda_0$ and $\lambda_{M-1}$, the eigenvectors $\mathbf{u}_{0}$ and $\mathbf{u}_{M-1}$ can be found as
\begin{align}
\mathbf{u}_{0}^{(1)} =&\  \frac{1}{\sqrt{2}}{\bm \delta}_0 + \sum_{\ell = 1}^{M/2-1} \frac{\widehat{s}_{2\ell+1}}{ \sqrt{2}\sqrt{\sum_{\ell = 1}^{M/2-1} \widehat{s}_{2\ell+1}^2}}{\bm \delta}_{2\ell+1}, \nonumber \\
\mathbf{u}_{M-1}^{(1)} =&\  \frac{1}{\sqrt{2}}{\bm \delta}_0 - \sum_{\ell = 1}^{M/2-1} \frac{\widehat{s}_{2\ell+1}}{ \sqrt{2}\sqrt{\sum_{\ell = 1}^{M/2-1} \widehat{s}_{2\ell+1}^2}}{\bm \delta}_{2\ell+1}.
\end{align}
$\{\lambda_n\}$ and $\{\mathbf{u}^{(1)}_n\}$ give us the eigenvalue decomposition of $\widehat{\mathbf{S}}^{(1)}$. For example, when $M=8$, ${\mathbf{D}}^{(1)} =  \mathrm{diag}(\lambda_0, 1, 1, 1, 1, 1, 1, \lambda_{M-1})$ and 
\begin{align}
\mathbf{U}^{(1)} =  & 
\begin{bmatrix}
\frac{1}{ \sqrt{2}} & 0 & 0 & 0 & 0 & 0 &0 & \frac{1}{ \sqrt{2}} \\
0 & 1 & 0 & 0 & 0 & 0 &0 & 0 \\
0 & 0 & 1 & 0 & 0 & 0 &0 & 0 \\
\frac{A_3}{ \sqrt{2}} & 0& 0& 0 & u_{3,4} & 0 &u_{3,6}& -\frac{A_3}{ \sqrt{2}}  \\
0 & 0 & 0 & 1 & 0 & 0 &0 & 0 \\
\frac{A_5}{ \sqrt{2}} & 0 & 0 & 0 &u_{5,4} &0 & 0  &  -\frac{A_5}{ \sqrt{2}} \\
0 & 0 & 0 & 0 & 0 & 1 & 0 & 0 \\
\frac{A_7}{ \sqrt{2}} & 0 & 0 & 0 & 0 & 0 &u_{7,6} &  -\frac{A_7 }{ \sqrt{2}}
\end{bmatrix},
\end{align}
where $A_{2\ell+1} = \frac{\widehat{s}_{2\ell+1}}{ \sqrt{ \widehat{s}_{3}^2+ \widehat{s}_{5}^2+ \widehat{s}_{7}^2}} \ (<1)$ and $u_{3,4},\ u_{3,6},\ u_{5,4},\ u_{7,6}$ are the elements after orthogonalization. Then, the elements in the upper-right triangle of $\widehat{\mathbf{S}}^{(1)}$ are
\begin{align}\label{eq:Shat1UDU}
&[\widehat{\mathbf{S}}^{(1)}]_{m,n} 
= [{\mathbf{U}}^{(1)}{\mathbf{D}}^{(1)}{\mathbf{U}}^{(1)\top}]_{m,n} \nonumber \\
 =&\  
 \begin{cases}
 \frac{1}{2}(\lambda_0 +\lambda_{M-1}) \quad (m= n = 1) \\
 \frac{A_n}{2}(\lambda_0 - \lambda_{M-1}) \quad (m=0,  n = 2\ell + 1  \geq 3 )\\ 
\frac{A_m^2}{2} (\lambda_0 + \lambda_{M-1})+ \Delta_{m,m}\quad (m= n = 2\ell+1 \geq 3)\\
\frac{A_mA_n}{2} (\lambda_0 + \lambda_{M-1})+ \Delta_{m,n} \\ 
 (m = 2\ell_m+1,\ n = 2\ell_n+1,\ \ell_m \neq \ell_n,\  m,\  n \geq 3)\\
1 \quad (m = n = 1,\  2,\  2\ell \  (\ell \geq 2) )\\
0 \quad \mathrm{otherwise}
 \end{cases},
\end{align}
where $\Delta_{m,n}$ contains the result of multiplication. The upper-right triangle elements of $\widehat{\mathbf{T}}^{(1)} = {\mathbf{U}}^{(1)}({\mathbf{D}}^{(1)})^{-1}{\mathbf{U}}^{(1)\top}$ are
\begin{align}\label{eq:That}
&[\widehat{\mathbf{T}}^{(1)}]_{m,n}
=  [{\mathbf{U}}^{(1)}({\mathbf{D}}^{(1)})^{-1}{\mathbf{U}}^{(1)\top}]_{m,n} \nonumber \\
 =&\  
 \begin{cases}
 \frac{1}{2}(\frac{1}{\lambda_0} + \frac{1}{\lambda_{M-1}})  \quad (m= n = 1)   \\
 \frac{A_n}{2}(\frac{1}{\lambda_0} - \frac{1}{\lambda_{M-1}}) \quad (m=0,  n = 2\ell + 1 \geq 3 ) \\ 
\frac{A_m^2}{2}(\frac{1}{\lambda_0} + \frac{1}{\lambda_{M-1}}) + \Delta_{m,m}\quad (m= n = 2\ell+1 \geq 3)\\ 
 \frac{A_mA_n}{2}(\frac{1}{\lambda_0} + \frac{1}{\lambda_{M-1}}) + \Delta_{m,n} \\
  (m = 2\ell_m+1,\ n = 2\ell_n+1,\ \ell_m \neq \ell_n,\  m,\  n \geq 3)\\
1 \quad(m = n = 1,\  2,\  2\ell\  (\ell \geq 2))\\
0 \quad  \mathrm{otherwise}
 \end{cases} .
\end{align}
From Lemma \ref{lem:S}, it follows that
\begin{align}
\frac{1}{\lambda_0\lambda_{M-1}} =
 \frac{1}{1 -  (\sum_{\ell = 1}^{M/2-1} \widehat{s}_{2\ell+1}^2)} = \frac{1}{\widehat{s}_{1}^2} = \frac{(M\sin(\frac{\pi}{2M}))^2}{2}.
\end{align}
Since $M\sin(\frac{\pi}{2M}) = \frac{\pi}{2}\frac{2M}{\pi}\sin(\frac{1}{\frac{2M}{\pi}})$, $x\sin(\frac{1}{x})$ monotonically increases over $[\frac{4}{\pi}, \infty ) $ and $x\sin(\frac{1}{x}) \xrightarrow{x\rightarrow \infty} 1$, then
\begin{align}\label{eq:11716}
1.1716 \approx  \frac{(4\sin(\frac{\pi}{8}))^2}{2} <&  \frac{(M\sin(\frac{\pi}{2M}))^2}{2}
   \xrightarrow{M\rightarrow \infty} \frac{\pi^2}{8} \approx 1.2337, \nonumber\\
   \frac{1}{\lambda_0\lambda_{M-1}} =&\  1 + \epsilon \quad (\epsilon < \frac{1}{4}),
\end{align}
and $\frac{1}{\lambda_0} + \frac{1}{\lambda_{M-1}} = \frac{\lambda_0+\lambda_{M-1}}{\lambda_0\lambda_{M-1}} = (\lambda_0+\lambda_{M-1})(1 + \epsilon)$, $\frac{1}{\lambda_0} - \frac{1}{\lambda_{M-1}} = -\frac{\lambda_0-\lambda_{M-1}}{\lambda_0\lambda_{M-1}} = -(\lambda_0-\lambda_{M-1})(1 + \epsilon)$. Thus, by substituting \eqref{eq:Shat1}, \eqref{eq:Shat1UDU}, and \eqref{eq:11716} into \eqref{eq:That}, we can derive
\begin{align}\label{eq:That1}
&[\widehat{\mathbf{T}}^{(1)}]_{m,n} 
=  [{\mathbf{U}}^{(1)}({\mathbf{D}}^{(1)})^{-1}{\mathbf{U}}^{(1)\top}]_{m,n} \nonumber \\
 =&\  
 \begin{cases}
 1 + \epsilon > 1 \quad (m= n = 1)  \\
 - \widehat{s}_{n} (1 + \epsilon)\quad (m=0,  n = 2\ell + 1 \geq 3 )\\ 
1+ \frac{A_m^2}{2}\epsilon > 1\quad (m= n = 2\ell+1 \geq 3) \\
\frac{A_mA_n}{2}\epsilon < \frac{1}{8} \\
  (m = 2\ell_m+1,\ n = 2\ell_n+1,\ \ell_m \neq \ell_n,\  m,\  n \geq 3)\\
1 \quad(m = n = 1,\  2,\  2\ell\  (\ell \geq 2))\\
0 \quad  \mathrm{otherwise} 
 \end{cases} .
\end{align}
Here, let $\rho(M,n) = \widehat{s}_{n} = \frac{\sqrt{2}}{(M\sin(\frac{\pi}{2M}n))}$. Since we assume that the size $M$ for the RDADCF is $M \geq 4$,
\begin{align}
\frac{2}{5}>0.3827 \approx  \rho(4,3) > \rho(M,n)  > \rho(M+1,n) ,\nonumber\\ 
\frac{2}{5}>0.3827 \approx  \rho(4,3) > \rho(M,n)  > \rho(M,n+1),
\end{align}
thus $|-\widehat{s}_{n} (1 + \epsilon)| < \frac{2}{5}\frac{5}{4} = \frac{1}{2}$. Finally, we conclude that $|[\mathbf{T}^{(1)\top}\mathbf{T}^{(1)}]_{m,n}| =  |\langle \mathbf{t}_{m}^{(1)} , \mathbf{t}_{n}^{(1)} \rangle | \ll |\langle \mathbf{t}_{m}^{(1)} , \mathbf{t}_{m}^{(1)} \rangle |$, which implies the passband of each $\mathbf{t}_m^{(1)}$ is the same as $\mathbf{s}_m^{(1)}$.

For $k = 2$, $\mathbf{S}^{(2)}$ forms as in \eqref{eq:Shat1} with $\widehat{s}_3=0$. With the same discussion when $k=1$, it can be derived that lower bounds of the diagonal elements of $|[\widehat{\mathbf{T}}^{(2)}]_{m,m}|$ are 1 and upper bounds of the elements $|[\widehat{\mathbf{T}}^{(2)}]_{m,n}|$ ($m\neq n$) are $\frac{1}{2}$ or $\frac{1}{8}$, as in \eqref{eq:That1}. Thus, $|\langle \mathbf{t}_{m}^{(2)} , \mathbf{t}_{n}^{(2)} \rangle |\ll |\langle \mathbf{t}_{m}^{(2)} , \mathbf{t}_{m}^{(2)} \rangle |$. This is the end of proof for Lemma \ref{lem:pbT}.
\end{proof}
\section{Detailed Algorithm of Image Recovery Used in Experiments}\label{sec:AIR}
To solve \eqref{eq:opt}, the primal-dual splitting (PDS) algorithm \cite{Condat2013,Vu2013} is used. Consider the following convex optimization problem to find
\begin{align}
\label{eq:pds}
\mathbf{x}^{\star} \in \argmin_{\mathbf{x}\in \mathbb{R}^{N_1}} f(\mathbf{x}) + g(\mathbf{L}\mathbf{x}),
\end{align}
where $f\in \Gamma_0(\mathbb{R}^{N_1})$, $g \in \Gamma_0(\mathbb{R}^{N_2})$ ($\Gamma_0(\mathbb{R}^{N_2})$ is the set of proper lower semicontinuous convex functions \cite{Bauschke2011} on $\mathbb{R}^{N}$), and $\mathbf{L} \in \mathbb{R}^{{N_2}\times {N_1}}$. Then, the optimal solution $\mathbf{x}^{\star}$, can be obtained as
\begin{align}
\label{eq:pdsalg}
\begin{cases}
\mathbf{x}^{(n+1)}:= \mathrm{prox}_{\gamma_1 f} [\mathbf{x}^{(n)} - \gamma_1 \mathbf{L}^{\top}\mathbf{z}^{(n)}] \\
\mathbf{z}^{(n+1)}:= \mathrm{prox}_{\gamma_2 g^{\ast}} [\mathbf{z}^{(n)} + \gamma_2 \mathbf{L}(2\mathbf{x}^{(n+1)} - \mathbf{x}^{(n)})]
\end{cases},
\end{align}
where $\mathrm{prox}$ denotes the {proximal operator} \cite{Bauschke2011}, $g^{\ast}$ is the conjugate function \cite{Bauschke2011} of $g$.  In the experiments, the parameters $\gamma_1$ and $\gamma_2$ in \eqref{eq:pdsalg}, are chosen as 0.01 and $\frac{1}{12\gamma_1}$. 
For Problem 2 ($\rho = 1$ in \eqref{eq:opt}), the functions $f$ and $g$, and the matrix $\mathbf{L}$ in \eqref{eq:pds}, are set as
\begin{align} 
\label{eq:setting}
&f(\mathbf{x}) = \iota_{C_{[0, 1]}}(\mathbf{x}), \nonumber\\  
&g([ \mathbf{z}_1^{\top}\ \mathbf{z}_2^{\top}\ \mathbf{z}_3^{\top}]^{\top}) = \|\mathbf{z}_1\|_{1} + \|\mathbf{z}_2\|_{1,2} + \iota_{\{\mathbf{y}\}}(\mathbf{z}_3),
\nonumber\\&\mathbf{z}_1 = \mathbf{F}\mathbf{P}_{\mathrm{v2bv}}\mathbf{x}, \ \mathbf{z}_2 = \widetilde{\mathbf{W}}_{\mathrm{b}}\mathbf{D}_{\mathrm{hv}}\mathbf{x},\  \mathbf{z}_3 = \widetilde{\mathbf{\Phi}}{\mathbf{x}},\nonumber\\ &\mathbf{L}=  
\begin{bmatrix}
(\mathbf{F}\mathbf{P}_{\mathrm{v2bv}})^\top & (\widetilde{\mathbf{W}}_{\mathrm{b}}\mathbf{D}_{\mathrm{hv}})^\top  & \widetilde{\mathbf{\Phi}}^\top 
\end{bmatrix}^\top.
\end{align}
The resulting solver for \eqref{eq:opt} is described in Algorithm \ref{alg2}\footnote{For $\mathbf{x} \in \mathbb{R}^N$, $[\mathrm{prox}_{\gamma \|\cdot \|_{1}}(\mathbf{x})]_i  = \mathrm{sign}(x_i) \max \{|x_i|-\gamma,0\}$ (soft-thresholding), $\mathrm{prox}_{\gamma \|\cdot \|_{1,2}}(\mathbf{x})$ is the group soft-thresholding \cite{Bach2012}, $\mathrm{prox}_{\iota_{C_{[0, 1]}}}(\mathbf{x})$ is the clipping operation to $[0,1]$, and $\mathrm{prox}_{\iota_{\{\mathbf{v}\}}}(\mathbf{x}) = \mathbf{y} $, where $\mathbf{y} \in \mathbb{R}^N$ is an observation. }. The stopping criterion is $\|\mathbf{x}^{(n+1)}-\mathbf{x}^{(n)}\|_2\leq 0.01$. The algorithm fof Problem 1 ($\rho = 0$ in \eqref{eq:opt}) can be designed by removing the terms and steps (Step 5, 8, and 10) relating to $\mathbf{z}_2$, $\mathbf{t}_2$, ${\mathbf{t}}_2^{(n)}$, and $\hat{\mathbf{t}}_2^{(n)}$ from Algorithm \ref{alg2}.
\begin{algorithm}[t]
    \caption{Solver for \eqref{eq:opt}}
    \label{alg2}
    \begin{algorithmic}[1]
        {\footnotesize
            \STATE set $n=0$ and choose $\mathbf{x}^{(0)}$, $\mathbf{z}_1^{(0)}$, $\mathbf{z}_2^{(0)}$, $\gamma_1$, $\gamma_2$.
            \WHILE{stop criterion is not satisfied}
            \STATE 
            $
            \mathbf{x}^{(n+1)}=\mathrm{prox}_{\gamma_1\iota_{C_{[0, 1]}}}(\mathbf{x}^{(n)}-\gamma_1((\mathbf{F}\mathbf{P}_{\mathrm{v2bv}})^\top\mathbf{z}_1^{(n)}+(\widetilde{\mathbf{W}}_{\mathrm{b}}\mathbf{D}_{\mathrm{hv}})^\top\mathbf{z}_2^{(n)}+\widetilde{\mathbf{\Phi}}^{\top}\mathbf{z}_3^{(n)}))
            $
            \STATE $\mathbf{t}_1^{(n)}=\mathbf{z}_1^{(n)}+\gamma_2\mathbf{F}\mathbf{P}_{\mathrm{v2bv}}(2\mathbf{x}^{(n+1)}-\mathbf{x}^{(n)})$.
            \STATE $\mathbf{t}_2^{(n)}=\mathbf{z}_2^{(n)}+\gamma_2\widetilde{\mathbf{W}}_{\mathrm{b}}\mathbf{D}_{\mathrm{hv}}(2\mathbf{x}^{(n+1)}-\mathbf{x}^{(n)})$.
            \STATE $\mathbf{t}_3^{(n)}=\mathbf{z}_3^{(n)}+\gamma_2\widetilde{\mathbf{\Phi}}(2\mathbf{x}^{(n+1)}-\mathbf{x}^{(n)})$.
            \STATE $\hat{\mathbf{t}}_1^{(n)}=\mathrm{prox}_{\frac{1}{\gamma_2} \|\cdot \|_{1}}\left(\frac{1}{\gamma_2}\mathbf{t}_1^{(n)}\right)$.
            \STATE $\hat{\mathbf{t}}_2^{(n)}=\mathrm{prox}_{\frac{1}{\gamma_2} \|\cdot \|_{1,2}}\left(\frac{1}{\gamma_2}\mathbf{t}_2^{(n)}\right)$.
            \STATE $\hat{\mathbf{t}}_3^{(n)}=\mathrm{prox}_{\frac{1}{\gamma_2}\iota_{\{\mathbf{y}\}}}\left(\frac{1}{\gamma_2}\mathbf{t}_3^{(n)}\right)$.

            \STATE $\mathbf{z}_k^{(n+1)}=\mathbf{t}_k^{(n)}-\gamma_2\hat{\mathbf{t}}_k^{(n)}$ ($k=1,2,3$).
            \STATE $n=n+1$.
            \ENDWHILE
            \STATE Output $\mathbf{u}^{(n)}$.}
    \end{algorithmic}
\end{algorithm}
\bibliographystyle{IEEEtran}
\bibliography{refs1}

\vfill
\end{document}